\documentclass[a4paper, UKenglish, cleveref, autoref, thm-restate]{lipics-v2021}

\usepackage{amsbsy, amsfonts, amsmath, mathabx, amssymb}
\usepackage{complexity}
\usepackage{tikz}
\usepackage{appendix}
\usepackage{color}
\usepackage{subcaption}

\newcommand{\Set}[1]{\left\{\, #1 \,\right\}}
\newcommand{\defproblem}[3]{
    \begin{center}
        \framebox{
            \parbox{0.92\textwidth}{
                #1 \\
                \textbf{Input}: #2 \\
                \textbf{Output}: #3
            }
        }
    \end{center}
}
\makeatletter
    \newcommand*{\Rome}[1]{\expandafter\@slowromancap\romannumeral #1@}
\makeatother
\newtheorem*{reduction*}{Reduction Rule}
\theoremstyle{claimstyle}
\newtheorem{splitstep}{Step}
\theoremstyle{claimstyle}
\newtheorem{chordalstep}{Step}
\theoremstyle{claimstyle}
\newtheorem{exactstep}{Step}

\title{Breaking the Barrier \texorpdfstring{$2^k$}{2\^k} for Subset Feedback Vertex Set in Chordal Graphs (Full Version)} 

\titlerunning{Breaking the Barrier \texorpdfstring{$2^k$}{2^k} for SFVS in Chordal Graphs} 

\author{Tian Bai}{School of Computer Science and Engineering, University of Electronic Science and Technology of China, [Chengdu], P. R. China}{tian.bai.cs@outlook.com}{https://orcid.org/0000-0003-1669-285X}{}

\author{Mingyu Xiao}{School of Computer Science and Engineering, University of Electronic Science and Technology of China, [Chengdu], P. R. China}{myxiao@uestc.edu.cn}{https://orcid.org/0000-0002-1012-2373}{}

\authorrunning{T. Bai and M. Xiao} 

\Copyright{Tian Bai and Mingyu Xiao} 



\category{} 

\relatedversion{} 

\nolinenumbers 

\EventEditors{John Q. Open and Joan R. Access}
\EventNoEds{2}
\EventLongTitle{40nd Conference on Very Important Topics (CVIT 2024)}
\EventShortTitle{CVIT 2024}
\EventAcronym{CVIT}
\EventYear{2024}
\EventDate{2024}
\EventLocation{AAA}
\EventLogo{}
\SeriesVolume{17}
\ArticleNo{XX}

\begin{document}

\maketitle

\begin{abstract}
    The \textsc{Subset Feedback Vertex Set} problem (\textsc{SFVS}) is to delete $k$ vertices from a given graph such that in the remaining graph, any vertex in a subset $T$ of vertices (called a terminal set) is not in a cycle.
    The famous \textsc{Feedback Vertex Set} problem is the special case of SFVS with $T$ being the whole set of vertices.
    In this paper, we study exact algorithms for \textsc{SFVS in Split Graphs} (\textsc{SFVS-S}) and \textsc{SFVS in Chordal Graphs} (\textsc{SFVS-C}).
    \textsc{SFVS-S} generalizes the minimum vertex cover problem and 
    the prize-collecting version of the maximum independent set problem in hypergraphs (\textsc{PCMIS}), and  \textsc{SFVS-C} further generalizes \textsc{SFVS-S}.
    Both \textsc{SFVS-S} and \textsc{SFVS-C} are implicit \textsc{$3$-Hitting Set} problems.
    However, it is not easy to solve them faster than \textsc{$3$-Hitting Set}.
    In 2019, Philip, Rajan, Saurabh, and Tale (Algorithmica 2019) proved that \textsc{SFVS-C} can be solved in $\mathcal{O}^{*}(2^{k})$ time, slightly improving the best result $\mathcal{O}^{*}(2.0755^{k})$ for \textsc{$3$-Hitting Set}.
    In this paper, we break the ``$2^{k}$-barrier'' for \textsc{SFVS-S} and \textsc{SFVS-C} by introducing an $\mathcal{O}^{*}(1.8192^{k})$-time algorithm.
    This achievement also indicates that \textsc{PCMIS} can be solved in $\mathcal{O}^{*}(1.8192^{n})$ time, marking the first exact algorithm for \textsc{PCMIS} that outperforms the trivial $\mathcal{O}^{*}(2^{n})$ threshold.
    Our algorithm uses reduction and branching rules based on the Dulmage-Mendelsohn decomposition and a divide-and-conquer method.
\end{abstract}

\section{Introduction}

The \textsc{Feedback Vertex Set} problem~(\textsc{FVS}), one of Karp's 21 $\NP$-complete problems~\cite{cocoKarp72}, is a fundamental problem in graph algorithms.
Given a graph $G = (V, E)$ and a parameter $k$, \textsc{FVS} is to decide whether there is a subset of vertices of size at most $k$ whose deletion makes the remaining graph acyclic.
\textsc{FVS} arises in a variety of applications in various fields such as circuit testing, network communications, deadlock resolution, artificial intelligence, and computational biology~\cite{ijfcsBodlaender94,siamdmCyganPPW13,algorithmicaIwataK21}.
\textsc{FVS} with the parameter being the solution size $k$ is one of the most important problems in parameterized algorithms.
The first fixed-parameter tractable ($\FPT$) algorithm for \textsc{FVS} traces back to Bodlaender's linear-time algorithm for fixed $k$ \cite{ijfcsBodlaender94}.
After this, there were a large number of contributions to \textsc{FVS} in parameterized algorithms \cite{sctDowneyF92,springerDowneyF99,iwpecKanjPS04,talgRamanSS06}.
Later, Dehne et al. \cite{iwpecDehneFRS04} and Guo et al. \cite{jcssGuoGHNW06} proposed a single-exponential algorithm based on the iterative-compression method, respectively.
The iterative-compression method became an important method to solve \textsc{FVS} fast.
Several following improvements were based on this method.
For example, Chen et al. \cite{jcssChenFLLV08} improved the result to $\mathcal{O}^{*}(5^{k})$, where notation $\mathcal{O}^{*}(\cdot)$ hides all polynomial factors.
 Cao et al.~\cite{algorithmicaCaoCL15} improved the result to $\mathcal{O}^{*}(3.83^{k})$.
Kociumaka and Pilipczuk \cite{iplKociumakaP14} further improved the result to $\mathcal{O}^{*}(3.619^{k})$.
The current fastest deterministic algorithm, introduced by Iwata and Kobayashi \cite{algorithmicaIwataK21}, was an $\mathcal{O}^{*}(3.460^{k})$-time algorithm based on the highest-degree branching.

Because of the importance of \textsc{FVS}, different variants and generalizations have been extensively studied in the literature.
The \textsc{Subset Feedback Vertex Set} problem (\textsc{SFVS}), introduced by Even et al.~\cite{siamcompEvenNZ00} in 2000, is a famous case.
In \textsc{SFVS}, we are further given a vertex subset $T \subseteq V$ called \emph{terminal set}, and we are asked to determine whether there is a set of vertices of size at most $k$ whose removal makes each terminal in $T$ not contained in any cycle in the remaining graph.
When the terminal set is the whole vertex set of the graph, \textsc{SFVS} becomes \textsc{FVS}.
\textsc{SFVS} also generalizes another famous problem, i.e., \textsc{Node Multiway Cut}.
There is an $8$-approximation algorithm for \textsc{SFVS}~\cite{siamcompEvenNZ00}, and whether \textsc{SFVS} is $\FPT$ had been once a well-known open problem \cite{DagSemProcDeamineHM09}.
In 2013, Cygan et al. \cite{siamdmCyganPPW13} proved the fixed-parameter tractability of \textsc{SFVS} by giving an algorithm with running time $\mathcal{O}^{*}(2^{\mathcal{O}(k \log k)})$.
Recently, Iwata et al. \cite{siamcompIwataWY16,focsIwataYY18} showed the first single-exponential algorithm with running time $\mathcal{O}^{*}(4^{k})$ for \textsc{SFVS}.
In 2018, Hols and Kratsch showed that \textsc{SFVS} has a randomized polynomial kernelization with $\mathcal{O}(k^{9})$ vertices~\cite{mstHolsK18}.
In contrast, \textsc{FVS} admits a quadratic kernel~\cite{icalpIwata17,talgThomasse10}.
However, it remains open whether there is a deterministic polynomial kernel for \textsc{SFVS}.

\textsc{SFVS} has also been studied in several graph classes~\cite{damPapadopoulosT19,algorithmicaPhilipRST19}, such as interval graphs, permutation graphs, chordal graphs, and split graphs.
\textsc{SFVS} remains $\NP$-complete even in split graphs~\cite{algorithmicaFominHKPV14}, while \textsc{FVS} in split and chordal graphs are polynomial-time solvable~\cite{iplYannakakisG87}.
It turns out that both \textsc{SFVS in Chordal Graphs} (\textsc{SFVS-C}) and \textsc{SFVS in Split Graphs} (\textsc{SFVS-S}) can be regarded as implicit \textsc{$3$-Hitting Set}.
Its importance stems from the fact that \textsc{$3$-Hitting Set} can be used to recast a wide range of problems, and now it can be solved in time $\mathcal{O}^{*}(2.0755^{k})$~\cite{phdWahlstrom07}.
On the other hand, when we formulate \textsc{SFVS-C} or \textsc{SFVS-S} in terms of \textsc{$3$-Hitting Set}, the structural properties of the input graph are lost.
We believe these structural properties can potentially be exploited to obtain faster algorithms for the original problems.
However, designing a faster algorithm for \textsc{SFVS-S} and \textsc{SFVS-C} seems challenging.
Only recently did Philip et al.~\cite{algorithmicaPhilipRST19} improve the running bound to $\mathcal{O}^{*}(2^{k})$,
where they needed to consider many cases of the clique-tree structures of the chordal graphs.
In some cases, they needed to branch into seven branches.
Note that $2^{k}$ is another barrier frequently considered in algorithm design and analysis.
Some preliminary brute force algorithms, dynamic programming, and advanced techniques, such as inclusion-exclusion, iterative compression, and subset convolution, always lead to the bound $2^{k}$.
Breaking the ``$2^{k}$-barrier'' becomes an interesting question for many problems.

We highlight that \textsc{SFVS-C} and \textsc{SFVS-S} are important since they generalize a natural variation of the maximum independent set problem called \textsc{Prize-Collecting Maximum Independent Set in hypergraphs} (\textsc{PCMIS}).
In PCMIS, we are given a hypergraph $H$ with $n$ vertices.
The object is to find a vertex subset $S$ maximizing the size of $S$ minus the number of hyperedges in $H$ that contain at least two vertices from $S$.
In other words, we may balance the size of the vertex subset against the number of hyperedges on which $S$ violates the independent constraints.
The prize-collecting version of many important fundamental problems has drawn certain attention recently, such as \textsc{Prize-Collecting Steiner Tree}~\cite{algorithmicaPedrosaR22}, \textsc{Prize-Collecting Network Activation}~\cite{talgFukunaga17}, and \textsc{Prize-Collecting Travelling Salesman} Problem~\cite{stocBlauthN23}.
To the best of our knowledge, no exact algorithm for \textsc{PCMIS} faster than $\mathcal{O}^{*}(2^{n})$ is known before.

We also mark that several other implicit \textsc{$3$-Hitting Set} problems, such as \textsc{Cluster Vertex Deletion} and \textsc{Directed FVS in Tournaments}, have been extensively studied in the past few years.
Fomin et al.~\cite{talgFominLLSTZ19} showed that \textsc{Cluster Vertex Deletion} and \textsc{Directed FVS in Tournaments} admit subquadratic kernels with $\mathcal{O}(k^{5/3})$ vertices and $\mathcal{O}(k^{3/2})$ vertices, respectively; while the size of the best kernel for \textsc{SFVS-C} is still quadratic, which can be easily obtained from the kernelization for \textsc{3-Hitting Set}~\cite{jcssAbu-Khzam10}.
As for parameterized algorithms, Dom et al.~\cite{jdaDomGHNT10} first designed an $\mathcal{O}^{*}(2^{k})$-time algorithm for \textsc{Directed FVS in Tournaments}, breaking the barrier of \textsc{$3$-Hitting Set}, and the running time bound of which was later improved to $\mathcal{O}^{*}(1.6191^{k})$ by Kumar and Lokshtanov~\cite{stacsKumarL16}.
For \textsc{Cluster Vertex Deletion}, in 2010, H{\"{u}}ffner et al.~\cite{mstHuffnerKMN10} first broke the barrier of \textsc{$3$-Hitting Set} by obtaining an $\mathcal{O}^{*}(2^{k})$-time algorithm.
Now it can be solved in $\mathcal{O}^{*}(1.7549^{k})$ time~\cite{cocoonTianXY23}.

\paragraph*{Contributions and Techniques.}
In this paper, we contribute to parameterized and exact algorithms for \textsc{SFVS-C}.
Our main contributions are summarized as follows.
\begin{enumerate}[1.]
    \item We firstly break the ``$2^{k}$-barrier'' for \textsc{SFVS-S} and \textsc{SFVS-C} by giving an $\mathcal{O}^{*}(1.8192^{k})$-time algorithm, which significantly improves previous algorithms.
    \item We show that an $\mathcal{O}^{*}(\alpha^{k})$-time algorithm ($\alpha > 1$) for \textsc{SFVS-S} leads to an $\mathcal{O}^{*}(\alpha^{n})$-algorithm for \textsc{PCMIS}.
    Thus, we can solve \textsc{PCMIS} in time $\mathcal{O}^{*}(1.8192^{n})$, also breaking the ``$2^{n}$-barrier'' for this problem for the first time.
    \item We make use of the Dulmage-Mendelsohn decomposition of bipartite graphs to capture structural properties, and then we are able to use a new measure $\mu$ to analyze the running time bound.
    This is the most crucial technique for us to obtain a significant improvement.
    Note that direct analysis based on the original measure $k$ has encountered bottlenecks.
    Any tiny improvement may need complicated case-analysis.
    \item The technique based on Dulmage-Mendelsohn decomposition can only solve \textsc{SFVS-S}.
    We also propose a divide-and-conquer method by dividing the instance of \textsc{SFVS-C} into several instances of \textsc{SFVS-S}.
    We show that \textsc{SFVS-C} can be solved in time $\mathcal{O}^{*}(\alpha^{k} + 1.6191^{k})$ if \textsc{SFVS-S} can be solved in time $\mathcal{O}^{*}(\alpha^{k})$.
    \item By doing a trade-off between our $\mathcal{O}^{*}(1.8192^{k})$-time algorithm and other algorithms, we also get improved exact algorithms: \textsc{SFVS-S} can be solved in time $\mathcal{O}(1.3488^{n})$, and \textsc{SFVS-C} can be solved in time $\mathcal{O}(1.3788^{n})$.
\end{enumerate}

\section{Preliminaries}

\subsection{Graphs}

Let $G = (V, E)$ stand for a simple undirected graph with a set $V$ of vertices and a set $E$ of edges.
We adopt the convention that $n = |V|$ and $m = |E|$.
When a graph $G'$ is mentioned without specifying its vertex and edge sets, we use $V(G')$ and $E(G')$ to denote these sets, respectively.
For a subset $X \subseteq V$ of vertices, we define the following notations.
The \emph{neighbour set} of $X$, denoted by $N_{G}(X)$, is the set of all vertices in $V \setminus X$ that are adjacent to a vertex in $X$, and the \emph{closed neighbour set} of $X$ is expressed as $N_{G}[X] \coloneq N_{G}(X) \cup X$.
The subgraph of $G$ induced by $X$ is denoted by $G[X]$.
We simply write $G - X \coloneq G[V \setminus X]$ as the subgraph obtained from $G$ removing $X$ together with edges incident on any vertex in $X$.
For ease of notation, we may denote a singleton set $\Set{v}$ by $v$.

The \emph{degree} of $v$ in $G$ is defined by $\deg_{G}(v) \coloneq |N_{G}(v)|$.
An edge $e$ is a \emph{bridge} if it is not contained in any cycle of $G$.
A \emph{separator} of a graph is a vertex set such that the deletion of it increases the number of connected components of the graph.
The shorthand $[r]$ is expressed as the set $\Set{1, 2, \ldots, r}$ for $r \in \mathbb{N}^{+}$.

In an undirected graph $G = (V, E)$, a set $X \subseteq V$ is a \emph{clique} if every pair of distinct vertices $u$ and $v$ in $X$ are connected by an edge $uv \in E$; $X$ is an \emph{independent set} if $uv \not\in E$ for every pair of vertices $u$ and $v$ in $X$; $X$ is a \emph{vertex cover} if for any edge $uv \in E$ at least one of $u$ and $v$ is in $X$.
A subset $S \subseteq V$ is a vertex cover of $G$ if and only if $V \setminus S$ is an independent set.
A vertex $v$ is called \emph{simplicial} in $G$ if $N_{G}[v]$ is a clique~\cite{amhaajDirac61}.
A clique in $G$ is \emph{simplicial} if it is maximal and contains at least one simplicial vertex.
A \emph{matching} is a set of edges without common vertices.

\subsection{Chordal Graphs and Split Graphs} \label{CHORDAL AND SPLIT}

A \emph{chord} of a cycle is an edge that connects two non-consecutive vertices of the cycle.
A graph $G$ is said to be \emph{chordal} if every cycle of length at least $4$ contains a chord.
A chordal graph $G$ holds the following facts~\cite{amhaajDirac61} that will be used in the paper:
\begin{enumerate}[1.]
    \item Every induced subgraph of $G$ is chordal.
    \item Every minimal separator of $G$ is a clique.
\end{enumerate}

Consider a connected chordal graph $G$, and let $\mathcal{Q}_{G}$ denote the set of all maximal cliques in $G$.
A \emph{clique graph} of $G$ is an undirected graph $(\mathcal{Q}_{G}, \mathcal{E}_{G}, \sigma)$ with the edge-weighted function $\sigma \colon \mathcal{E}_{G} \to \mathbb{N}$ satisfying that an edge $Q_{1}Q_{2} \in \mathcal{E}_{G}$ if $Q_{1} \cap Q_{2}$ is a minimal separator and $\sigma (Q_{1}Q_{2}) \coloneq |Q_{1} \cap Q_{2}|$.
A \emph{clique tree} $\mathcal{T}_{G}$ of $G$ is a maximum spanning tree of the clique graph of $G$, and the following facts hold~\cite{dmBuneman74,jctGavril74,phdWalter72}:
\begin{enumerate}[1.]
    \item Each leaf node of a clique tree $\mathcal{T}_{G}$ is a simplicial clique in $G$.
    \item For a pair of maximal cliques $Q_{1}$ and $Q_{2}$ such that $Q_{1}Q_{2} \in \mathcal{E}_{G}$, $Q_{1} \cap Q_{2}$ separates each pair of vertices $v_{1} \in Q_{1} \setminus Q_{2}$ and $v_{2} \in Q_{2} \setminus Q_{1}$.
\end{enumerate}

Whether a graph is chordal can be checked in linear time $\mathcal{O}(n + m)$~\cite{siamcompRoseTL76}.
The number of maximal cliques in a chordal graph $G$ is at most $n$~\cite{pjmFulkersonDG65}, and all of them can be listed in linear time $\mathcal{O}(n + m)$~\cite{jctGavril74}.
These properties will be used in our algorithm.
A graph is a \emph{split graph} if its vertex set can be partitioned into a clique $K$ and an independent set $I$~\cite{seiccgtcStephaneH77}.
Such a partition $(I, K)$ is called a \emph{split partition}.
It is worth noting that every split graph is chordal, and whether a graph is a split graph can also be checked in linear time $\mathcal{O}(n + m)$ by definition.

\subsection{Subset Feedback Vertex Set in Split and Chordal Graphs}

Given a terminal set $T \subseteq V$ of an undirected graph $G = (V, E)$, a cycle in $G$ is a \emph{$T$-cycle} if it contains a terminal from $T$, and a \emph{$T$-triangle} is specifically a $T$-cycle of length three.
A \emph{subset feedback vertex set} of a graph $G$ with a terminal set $T$ is a subset of $V$ whose removal makes $G$ contain no $T$-cycle.

In this study, we focus on \textsc{SFVS} in chordal graphs.
The problem takes as input a chordal graph $G = (V, E)$, a terminal set $T \subseteq V$, and an integer $k$.
The task is to determine whether there is a subset feedback vertex set $S$ of size at most $k$.
Moreover, the following lemma shows that the problem can be transformed into the problem of finding a subset of vertices intersecting all $T$-triangles instead of all $T$-cycles.

\begin{lemma}[\cite{algorithmicaPhilipRST19}] \label{T TRIANGLE}
    Let $G = (V, E)$ be a chordal graph and $T \subseteq V$ be the terminal set.
    A vertex set $S \subseteq V$ is a subset feedback vertex set of $G$ if and only if $G - S$ contains no $T$-triangles.
\end{lemma}

For the sake of presentation, this paper considers a slight generalization of \textsc{SFVS-C}.
In this generalized version, a set of marked edges $M \subseteq E$ is further given, and we are asked to decide whether there is a subset feedback vertex set of size at most $k$, which also covers all marked edges, i.e., each marked edge must have at least one of its endpoints included in the set.
This set is called a \emph{solution} to the given instance.
Among all solutions, a \emph{minimum solution} is the one with the smallest size.
The size of a minimum solution to an instance $\mathcal{I}$ is denoted by $s(\mathcal{I})$.

Formally, the generalization of \textsc{SFVS-C} is defined as follows.

\defproblem{\textsc{(Generalized) SFVS-C}}
{A chordal graph $G = (V, E)$, a terminal set $T \subseteq V$, a marked edge set $M \subseteq E$, and an integer $k$.}
{Determine whether there is a subset of vertices $S \subseteq V$ of size at most $k$, such that neither edges in $M$ nor $T$-cycles exist in $G - S$.}

We have the following simple observations.
Let $abc$ be a $T$-triangle with a degree-$2$ vertex $b$ in the graph.
Any solution must contain at least one of the vertices $a$, $b$, and $c$.
If vertex $b$ is included in the solution, we can replace it with either $a$ or $c$ without affecting the solution's feasibility.
Consequently, we can simplify the graph by removing $b$ and marking edge $ac$.
This observation motivates the consideration of the generalized version.

The equivalence between the original and generalized versions is apparent, allowing us to simply use \textsc{SFVS-C} to denote the generalized version.
An instance of our problem is denoted by $\mathcal{I} = (G, T, M, k)$.

During the algorithm, it may be necessary to consider some sub-instances where the graph is a subgraph of $G$.
We define the \emph{instance induced by $X \subseteq V$ or $G[X]$} as $(G[X], T \cap X, M \cap E(G[X]), k)$.

\subsection{Branching Algorithms}

Our algorithm consists of several reduction rules and branching rules.
To estimate the progress made by each step, we need to select a measure $\mu(\cdot)$ for the instances.
In each step of the algorithm, the measure does not increase, and if it becomes non-positive, the problem can be solved in polynomial time.

Branching rules divide the current instance $\mathcal{I}$ into several sub-instances $\mathcal{I}_{1}, \mathcal{I}_{2}, \ldots, \mathcal{I}_{r}$, where $\mu(\mathcal{I}_{i}) < \mu(\mathcal{I})$ for each $i \in [r]$.
A branching rule is safe if $\mathcal{I}$ is a Yes-instance if and only if at least one of the sub-instances is a Yes-instance.
Reduction rules transform the input instance into a ``smaller'' one.
A reduction rule is safe if the input instance is a Yes-instance if and only if the output instance is a Yes-instance.

To analyze the running time, we bound the size of the search tree generated by the algorithm.
We use $R(\mu)$ to denote the worst-case size of the search tree in our algorithm with respect to the measure $\mu$.
For a branching rule, we establish the recurrence relation $R(\mu) \leq 1 + \sum_{i \in [r]} R(\mu_{i})$, where $\mu_{i} = \mu(\mathcal{I}_{i})$ for $i \in [r]$.
This recurrence relation is typically represented by the \emph{branching vector} $(\mu - \mu_{1}, \mu - \mu_{2}, \ldots, \mu - \mu_{r})$.
The \emph{branching factor}, which is the unique positive real root of the equation $x^{\mu} - \sum_{i = 1}^{r} x^{\mu_{i}} = 0$, is associated with this branching vector.

A branching vector $\mathrm{\mathbf{a}}$ is \emph{not worse than} $\mathrm{\mathbf{b}}$ if the branching factor of $\mathrm{\mathbf{a}}$ is not greater than that of $\mathrm{\mathbf{b}}$.
We will use the property that if $a_{i} \geq b_{i}$ for each $i \in [r]$, the branching vector $\mathrm{\mathbf{a}} = (a_{1}, a_{2}, \ldots, a_{r})$ is not worse than $\mathrm{\mathbf{b}} = (b_{1}, b_{2}, \ldots, b_{r})$.
If the maximum branching factor among that of all branching vectors is $\alpha$, then the size of the search tree in the algorithm is bounded by $\mathcal{O}(\alpha^{\mu})$.

\section{The Dulmage-Mendelsohn Decomposition and Reduction} \label{SEC: DM REDUCTION}
This section introduces the Dulmage-Mendelsohn decomposition of a bipartite graph~\cite{canjmathDulmageM58,tranrscDulmageA59}.
The Dulmage-Mendelsohn decomposition will play a crucial role in our algorithm for \textsc{SFVS-S}.

\begin{definition}[Dulmage-Mendelsohn Decomposition~\cite{elsevierLocaszP86,jcssChenK03}]
    Let $F$ be a bipartite graph with bipartition $V(F) = A \cup B$.
    The Dulmage-Mendelsohn decomposition \textup{(}cf.~Fig.~\ref{FIG: DM DECOMPOSITION}\textup{)} of $F$ is a partitioning of $V(F)$ into three disjoint parts $C$, $H$ and $R$, such that
    \begin{enumerate}[1.]
        \item $C$ is an independent set and $H = N_{F}(C)$;
        \item $F[R]$ has a perfect matching;
        \item $H$ is the intersection of all minimum vertex covers of $F$; and
        \item any maximum matching in $F$ includes all vertices in $R \cup H$.
    \end{enumerate}
\end{definition}

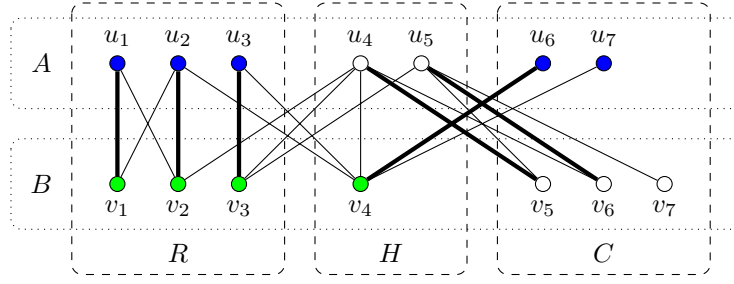
\begin{figure}[t!]
    \centering
        \begin{tikzpicture}
            [
            scale = 0.8,
            nonterminal/.style={draw, shape = circle, inner sep=2pt},
            terminal/.style={draw, shape = circle, fill = black, inner sep=2pt},
            ]
            \node[] (A) at(0.75, 0) {$A$};
            \node[fill = blue, nonterminal, label={[yshift = 0mm]$u_{1}$}] (u1) at(2, 0) {};
            \node[fill = blue, nonterminal, label={[yshift = 0mm]$u_{2}$}] (u2) at(3, 0) {};
            \node[fill = blue, nonterminal, label={[yshift = 0mm]$u_{3}$}] (u3) at(4, 0) {};
            \node[nonterminal, label={[yshift = 0mm]$u_{4}$}] (u4) at(6, 0) {};
            \node[nonterminal, label={[yshift = 0mm]$u_{5}$}] (u5) at(7, 0) {};
            \node[fill = blue, nonterminal, label={[yshift = 0mm]$u_{6}$}] (u6) at(9, 0) {};
            \node[fill = blue, nonterminal, label={[yshift = 0mm]$u_{7}$}] (u7) at(10, 0) {};

            \node[] (B) at(0.75, -2) {$B$};
            \node[fill = green, nonterminal, label={[yshift = -6.5mm]$v_{1}$}] (v1) at(2, -2) {};
            \node[fill = green, nonterminal, label={[yshift = -6.5mm]$v_{2}$}] (v2) at(3, -2) {};
            \node[fill = green, nonterminal, label={[yshift = -6.5mm]$v_{3}$}] (v3) at(4, -2) {};
            \node[fill = green, nonterminal, label={[yshift = -6.5mm]$v_{4}$}] (v4) at(6, -2) {};
            \node[nonterminal, label={[yshift = -6.5mm]$v_{5}$}] (v5) at(9, -2) {};
            \node[nonterminal, label={[yshift = -6.5mm]$v_{6}$}] (v6) at(10, -2) {};
            \node[nonterminal, label={[yshift = -6.5mm]$v_{7}$}] (v7) at(11, -2) {};

            \node[] (R) at(3, -3.15) {$R$};
            \node[] (H) at(6.5, -3.15) {$H$};
            \node[] (C) at(10, -3.15) {$C$};

            \draw[ultra thick]
            (u1) -- (v1) (u2) -- (v2) (u3) -- (v3)
            (u4) -- (v5) (u5) -- (v6) (v4) -- (u6);
            \draw
            (u1) -- (v2) (u2) -- (v1) (u2) -- (v4) (u3) -- (v4) (u4) -- (v2) (u4) -- (v3) (u5) -- (v3)
            (u4) -- (v4) (u7) -- (v4) (u4) -- (v6) (u5) -- (v5) (u5) -- (v7);

            \draw[rounded corners, dashed]
            (1.25, 1) rectangle (4.75, -3.5)
            (5.25, 1) rectangle (7.75, -3.5)
            (8.25, 1) rectangle (11.75, -3.5);

            \draw[rounded corners, dotted]
            (0.25, 0.75) rectangle (12.25, -0.75)
            (0.25, -1.25) rectangle (12.25, -2.75);
        \end{tikzpicture}
        \caption{A bipartite graph $F$ with bipartition $V(F) = A \cup B$, where $A = \Set{u_{i}}_{i = 1}^{7}$ and $B = \Set{v_{i}}_{i = 1}^{7}$.
        The thick edges form a maximum matching of $F$.
        The Dulmage-Mendelsohn decomposition of $F$ is $(C, H, R)$ with $C = \Set{u_{6}, u_{7}, v_{5}, v_{6}, v_{7}}$, $H = \Set{u_{4}, u_{5}, v_{4}}$, and $R = \Set{u_{1}, u_{2}, u_{3}, v_{1}, v_{2}, v_{3}}$.
        If $F$ is an auxiliary subgraph of an instance of \textsc{SFVS-S}, then $\hat{A} = \Set{u_{1}, u_{2}, u_{3}, u_{6}, u_{7}}$ (denoted by blue vertices) and $\hat{B} = \Set{v_{1}, v_{2}, v_{3}, v_{4}}$ (denoted by green vertices).}
    \label{FIG: DM DECOMPOSITION}
\end{figure}

The Dulmage-Mendelsohn decomposition always exists and is unique~\cite{elsevierLocaszP86}, which can be computed in time $\mathcal{O}(m\sqrt{n})$ by finding the maximum matching of the graph $F$~\cite{siamcompHopcroftK73}.
Leveraging this decomposition, we propose a crucial reduction rule for \textsc{SFVS-S}.

Consider an instance $\mathcal{I} = (G = (V, E), T, M, k)$ of \textsc{SFVS-S}.
Let $(I, K)$ be a split partition of $G$, where $I$ is an independent set and $K$ is a clique.
Based on the split partition $(I, K)$ of $G$, we can uniquely construct an auxiliary bipartite subgraph $F$ with bipartition $V(F) = A \cup B$.
In this subgraph $F$, partition $A$ is the subset of the vertices in $I$ that are only incident to marked edges and $B = N_{G}(A)$.
In addition, $E(F)$ is the set of all edges between $A$ and $B$, i.e., $E(F) \coloneq \{ ab \in E : a \in A,~b \in B \}$.
Notice that $F$ contains no isolated vertex, and all edges in $F$ are marked by the definitions of $A$ and $B$.

Let $(R, H, C)$ denote the Dulmage-Mendelsohn decomposition of the auxiliary subgraph $F$.
Define $\hat{A} \coloneq A \cap (R \cup C)$ and $\hat{B} \coloneq B \cap (R \cup H)$ (see Fig.~\ref{FIG: DM DECOMPOSITION}).
We have $\hat{B} = N_{G}(\hat{A})$, and there exists a matching saturating all vertices in $\hat{B}$.
This indicates that every solution must contain at least $|\hat{B}|$ vertices in $\hat{A} \cup \hat{B}$.
On the other hand, $\hat{B}$ is a minimum vertex cover of the subgraph induced by $\hat{A} \cup \hat{B}$.
Consequently, a minimum solution to $\mathcal{I}$ containing $\hat{B}$ exists.
Now, we formally introduce the following reduction rule, which we call DM Reduction.

\begin{reduction*}[DM Reduction] \label{REDUCTION: DM}
    Let $F$ be the auxiliary subgraph with bipartition $V(F) = A \cup B$, and let $(R, H, C)$ denote the Dulmage-Mendelsohn decomposition of $F$.
    If $\hat{A}$ and $\hat{B}$ are non-empty, delete $\hat{A}$ and $\hat{B}$ from the graph $G$ and decrease $k$ by $|\hat{B}| = |R| / 2 + |H \cap B|$.
\end{reduction*}

\begin{lemma} \label{SAFENESS: DM DECOMPOSITION}
    The DM Reduction is safe.
\end{lemma}

\begin{proof}
    Recall that $\hat{A} \coloneq A \cap (R \cup C)$ and $\hat{B} \coloneq B \cap (R \cup H)$.
    According to the definition of the Dulmage-Mendelsohn decomposition, we know that $N_{F}(\hat{A}) = \hat{B}$, and $\hat{B}$ is a minimum vertex cover of the subgraph induced by $\hat{A} \cup \hat{B}$.
    For a solution $S$ to the input instance $\mathcal{I}$, the size of $S \cap (\hat{A} \cup \hat{B})$ is no less than $|\hat{B}|$ since $S$ covers every edge in $M$.
    Let $S' = (S \setminus \hat{A}) \cup \hat{B}$.
    Observe that $|\hat{A}| > |\hat{B}|$; otherwise, $A$ would be a minimum vertex cover of $F$, contradicting that $H$ is a subset of any minimum vertex cover.
    Consequently, we derive that $|S'| \leq |S|$.
    In addition, we can see that $S'$ is also a solution, leading to the safeness of the DM Reduction.
\end{proof}

\begin{lemma} \label{STRUCTURE: DM DECOMPOSITION}
    Given an instance $\mathcal{I} = (G, T, M, k)$ of \textsc{SFVS-S}, let $F$ be the auxiliary subgraph of $G$ with bipartition $V(F) = A \cup B$.
    If the DM Reduction cannot be applied, for any non-empty subset $A' \subseteq A$, it holds that $|A'| < |N_{G}(A')|$.
\end{lemma}

\begin{proof}
    If the DM Reduction cannot be applied, the Dulmage-Mendelsohn decomposition of $F$ must be $(R, H, C) = (\varnothing, A, B)$.
    According to the definition of the Dulmage-Mendelsohn decomposition, $A$ is a vertex cover, and $H = A$ is the intersection of all minimum vertex covers of $F$.
    As a result, we know that $A$ is the unique minimum vertex cover of the auxiliary subgraph $F$.
    
    We assume to the contrary that there exists a subset $A' \subseteq A$ such that $|A'| \geq |N_{G}(A')|$.
    Then we immediately know that $(A \setminus A') \cup N_{G}(A')$ is a minimum vertex cover distinct from $A$, leading to a contradiction.
\end{proof}

\begin{lemma} \label{SAFENESS: MEASURE}
    Given an instance $\mathcal{I} = (G, T, M, k)$ of \textsc{SFVS-S}, let $F$ be the auxiliary subgraph of $G$ with bipartition $V(F) = A \cup B$.
    If the DM Reduction cannot be applied and $k < |A|$, the instance $\mathcal{I}$ is a No-instance.
\end{lemma}

\begin{proof}
    The size of the solution to $\mathcal{I} = (G, T, M, k)$ is no less than the size of the minimum vertex cover of $F$ since all marked edges need to be covered.
    If the DM Reduction cannot be applied, Lemma~\ref{STRUCTURE: DM DECOMPOSITION} implies that $A$ is the minimum vertex cover of $F$.
    Consequently, the size of the minimum solution to $\mathcal{I}$ must be no less than $|A|$, which means that an instance $\mathcal{I}$ is a No-instance if $k < |A|$.
\end{proof}

\section{Algorithms for \textsc{SFVS in Split Graphs}} \label{SEC: SFVS-S}

This section presents an algorithm for \textsc{SFVS-S}.
This algorithm plays a critical role in the algorithm for \textsc{SFVS-C}.

\subsection{Good Instances}

We begin by introducing a special instance of \textsc{SFVS-S}, which we refer to as a good instance.
In this subsection, we show that solving good instances is as hard as solving normal instances of \textsc{SFVS-S} in some sense.

\begin{definition}[Good Instances] \label{DEF: GOOD}
    An instance $\mathcal{I} = (G = (V, E), T, M, k)$ of \textsc{SFVS-S} is called \emph{good} if it satisfies the following properties:
    \begin{enumerate}[\textup{(}i\textup{)}]
        \item $(T, V \setminus T)$ is the split partition, where terminal set $T$ is the independent set and $V \setminus T$ forms the clique;
        \item every marked edge connects one terminal and one non-terminal; and
        \item the DM reduction cannot be applied on the auxiliary subgraph determined by $(T, V \setminus T)$.
    \end{enumerate}
\end{definition}

\begin{lemma} \label{ALG: GOOD = SFVS-S}
    For any constant $\alpha > 1$, \textsc{SFVS-S} can be solved in $\mathcal{O}^{*}(\alpha^{k})$ if and only if \textsc{SFVS-S} on good instances can be solved in $\mathcal{O}^{*}(\alpha^{k})$.
\end{lemma}

\begin{proof}
    We only need to show that if there exists an algorithm \texttt{GoodAlg} solving good instances in time $\mathcal{O}^{*}(\alpha^{k})$, there also exists an algorithm for \textsc{SFVS-S} running in the same time bound $\mathcal{O}^{*}(\alpha^{k})$.

    Let $\mathcal{I} = (G = (V, E), T, M, k)$ be an instance of \textsc{SFVS-S}.
    Notice that $\alpha$ is a constant.
    We select a sufficiently large constant $C$ such that the branching factor of the branching vector $(1, C, C)$ does not exceed $\alpha$.
    Our algorithm for \textsc{SFVS-S} is constructed below.
         
    First, we find the split partition $(I, K)$ of $G$ in polynomial time.
    If $|K| \leq 2C$, we solve the instance directly in polynomial time by brute-force enumerating subsets of $K$ in the solution.
    Otherwise, we assume that the size of $K$ is at least $2C + 1$.
    
    Next, we consider two cases.
    
    \textbf{Case 1:} There is a terminal $t \in K$.
    In this case, we partition $K \setminus \Set{t}$ into two parts $K'$ and $K''$ such that $|K'| \geq C$ and $|K''| \geq C$.
    If $t$ is not included in the solution, at most one vertex in the clique $K \setminus \Set{t}$ is not contained in the solution.
    Consequently, either $K'$ or $K''$ must be part of the solution.
    We can branch into three instances by either 
    \begin{itemize}
        \item removing $t$, and decreasing $k$ by $1$;
        \item removing $K'$, and decreasing $k$ by $|K'|$; or
        \item removing $K''$, and decreasing $k$ by $|K''|$.
    \end{itemize}
    This branching rule yields a branching vector $(1, |K'|, |K''|)$ (w.r.t. the measure $k$) with the branching factor not greater than $\alpha$ since $|K'| \geq C$ and $|K''| \geq C$.
    
    \textbf{Case 2:} No terminal is in $K$.
    For this case, each non-terminal $v \in I$ is not contained in any $T$-triangle.
    However, we cannot directly remove $v$ since it may be incident to marked edges.
    We can add an edge between $v$ and each vertex $u \in K$ not adjacent to $v$ without creating any new $T$-triangle or marked edge.
    This operation will change $v$ from a vertex in $I$ to a vertex in $K$, preserving the graph as a split graph.
    After handling all non-terminal $v \in I$, we know that the terminal set and non-terminal set form a split partition.
    Subsequently, for each marked edge between two non-terminals $v$ and $u$, we add a new $2$-degree terminal $t_{uv}$ adjacent to $u$ and $v$ and unmark the edge $uv$.
    We then apply the DM Reduction and obtain a good instance.
    Finally, we call the algorithm \texttt{GoodAlg} to solve the good instance in time $\mathcal{O}^{*}(\alpha^{k})$.

    By either branching with a branching factor not greater than $\alpha$ or solving the instance directly in $\mathcal{O}^{*}(\alpha^{k})$ time, our algorithm runs in time $\mathcal{O}^{*}(\alpha^{k})$.
\end{proof}

In the rest of this section, we only need to focus on the algorithm, denoted as \texttt{GoodAlg}, for good instances of \textsc{SFVS-S}.

\subsection{The Measure and Its Properties}

With the help of the auxiliary subgraph and DM Reduction (defined in Section~\ref{SEC: DM REDUCTION}), we use the following specific measure to analyze our algorithm.

\begin{definition} [The Measure of Good Instances] \label{DEF: MEASURE}
    Given a good instance $\mathcal{I} = (G, T, M, k)$ of \textsc{SFVS-S}, let $F$ be the auxiliary subgraph of $G$ with bipartition $V(F) = A \cup B$.
    We define the measure of the instance $\mathcal{I}$ as
    \begin{equation*}
        \mu(\mathcal{I}) \coloneq k - \frac{2}{3} |A|.
    \end{equation*}
\end{definition}

In our algorithm \texttt{GoodAlg}, the DM Reduction will be applied as much as possible once the graph changes to keep the instance always good.
Additionally, according to Lemma~\ref{SAFENESS: MEASURE}, an instance $\mathcal{I}$ can be solved in polynomial time when $\mu(\mathcal{I}) \leq 0$.
Thus, we can use $\mu(\cdot)$, defined in Definition~\ref{DEF: MEASURE}, as our measure to analyze the algorithm.

We may branch on a vertex by including it in the solution or excluding it from the solution in our algorithm.
In the first branch, we delete the vertex from the graph and decrease the parameter $k$ by $1$.
In the second branch, we execute a basic operation of \emph{hiding} the vertex, which is defined according to whether the vertex is a terminal.
\begin{description}
    \item[Hiding a terminal $t$:] delete every vertex in $N_{M}(t)$ and decrease $k$ by $|N_{M}(t)|$.
    \item[Hiding a non-terminal $v$:] delete every terminal in $N_{M}(v)$ and decrease $k$ by $|N_{M}(v)|$; for each $T$-triangle $vtu$ containing $v$, mark edge $tu$; and last, delete $v$ from the graph.
\end{description}
Here, the notation $N_{M}(v)$ represents the set of the vertices adjacent to $v$ via a marked edge.

\begin{lemma} \label{SAFENESS: OPERATIION}
    If there exists a solution containing a vertex $v$, then it is safe to delete $v$, decrease $k$ by $1$, and do the DM Reduction.
    If there exists a solution not containing a vertex $v$, then it is safe to hide $v$ and do the DM Reduction.
    Moreover, the resulting instance is good after applying either of the above two operations.
\end{lemma}

\begin{proof}
    Assuming a solution $S$ contains $v$, it is trivial that $S \setminus \Set{v}$ is also a solution to the instance $(G - v, T \setminus \Set{v}, k - 1)$.
    Moreover, since the DM Reduction is safe by Lemma~\ref{SAFENESS: DM DECOMPOSITION}, the first operation in the lemma is safe.
    
    Now, we assume that a solution $S$ does not contain a vertex $v$.
    Since $S$ must cover all edges in $M$, we know that $S$ contains all neighbours of $v$ in $M$.
    This shows that hiding $v$ is safe if $v$ is a terminal.
    Suppose that $v$ is a non-terminal, for every $T$-triangle $vut$ containing $v$, we have that $S \cap \Set{u, t} \neq \varnothing$.
    Consequently, it is safe to mark the edge $ut$ further.
    Moreover, since the DM Reduction is safe by Lemma~\ref{SAFENESS: DM DECOMPOSITION}, the second operation in the lemma is safe.

    Finally, either operation only deletes some vertices and marks some edges between terminals and non-terminals.
    Hence, the terminal set and the non-terminal set still form an independent set and a clique, repetitively.
    Additionally, the DM Reduction cannot be applied on the resulting instances.
    Therefore, the resulting instance is good after applying either of the above two operations.
\end{proof}

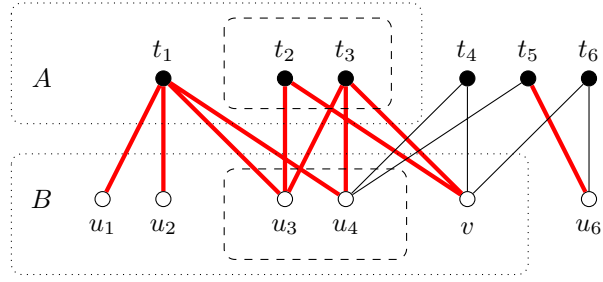
\begin{figure}[t!]
    \centering
        \begin{tikzpicture}
            [
            scale = 0.8,
            nonterminal/.style={draw, shape = circle, fill = white, inner sep=2pt},
            terminal/.style={draw, shape = circle, fill = black, inner sep=2pt},
            ]
            \node[] (A) at(1, 0) {$A$};
            \node[terminal, label={[yshift = 0mm]$t_{1}$}] (t1) at(3, 0) {};
            \node[terminal, label={[yshift = 0mm]$t_{2}$}] (t2) at(5, 0) {};
            \node[terminal, label={[yshift = 0mm]$t_{3}$}] (t3) at(6, 0) {};
            \node[terminal, label={[yshift = 0mm]$t_{4}$}] (t4) at(8, 0) {};
            \node[terminal, label={[yshift = 0mm]$t_{5}$}] (t5) at(9, 0) {};
            \node[terminal, label={[yshift = 0mm]$t_{6}$}] (t6) at(10, 0) {};

            \node[] (B) at(1, -2) {$B$};
            \node[nonterminal, label={[yshift = -7mm]$u_{1}$}] (v1) at(2, -2) {};
            \node[nonterminal, label={[yshift = -7mm]$u_{2}$}] (v2) at(3, -2) {};
            \node[nonterminal, label={[yshift = -7mm]$u_{3}$}] (v3) at(5, -2) {};
            \node[nonterminal, label={[yshift = -7mm]$u_{4}$}] (v4) at(6, -2) {};
            \node[nonterminal, label={[yshift = -7mm]$v$}] (v5) at(8, -2) {};
            \node[nonterminal, label={[yshift = -7mm]$u_{6}$}] (v6) at(10, -2) {};

            \draw[ultra thick, red]
            (t1) -- (v1) (t1) -- (v2) (t1) -- (v3) (t1) -- (v4)
            (t2) -- (v3) (t2) -- (v5)
            (t3) -- (v3) (t3) -- (v4) (t3) -- (v5)
            (t5) -- (v6);
            \draw
            (t4) -- (v4) (t4) -- (v5)
            (t5) -- (v4)
            (t6) -- (v5) (t6) -- (v6);
            \draw[rounded corners, dotted]
            (0.5, 1.25) rectangle (7.25, -.75)
            (0.5, -1.25) rectangle (9, -3.25);
            \draw[rounded corners, dashed]
            (4, 1) rectangle (6.75, -.5)
            (4, -1.5) rectangle (7, -3);
        \end{tikzpicture}
        \caption{
        The graph $G$, where black vertices are terminals, white vertices are non-terminals, thick and red edges are marked edges, and edges between two non-terminals are not presented in the graph; the auxiliary subgraph is $F$ with bipartition $V(F) = A \cup B$, where $A = \Set{t_{1}, t_{2}, t_{3}}$ and $B = \Set{u_{1}, u_{2}, u_{3}, u_{4}, v}$ (denoted by dotted boxes).
        After deleting $v$, the DM Reduction can be applied.
        When doing the DM Reduction, $\hat{A} = \Set{t_{2}, t_{3}}$ and $\hat{B} = \Set{u_{3}, u_{4}}$ (denoted by dashed boxes) are deleted.}
        \label{FIG: DELETE + DM}
\end{figure}

\begin{lemma} \label{MEASURE: DELETE}
    Given the good instance $\mathcal{I} = (G = (V, E), T, M, k)$, let $F$ be the auxiliary subgraph of $G$ with bipartition $V(F) = A \cup B$, and $v$ be a vertex in $V \setminus A$.
    Let $\mathcal{I}_{1}$ be the instance obtained from $\mathcal{I}$ by first deleting $v$ and then doing the DM Reduction \textup{(}cf.~Fig.~\ref{FIG: DELETE + DM}\textup{)}.
    Then $\mathcal{I}_{1}$ is a good instance such that $\mu(\mathcal{I}) - \mu(\mathcal{I}_{1}) \geq 0$.
\end{lemma}

\begin{proof}
    Let $\mathcal{I}_{0} = (G_{0}, T_{0}, M_{0}, k_{0})$ be the instance after deleting $v$ from $\mathcal{I}$.
    Then, $\mathcal{I}_{1} = (G_{1}, T_{1}, M_{1}, k_{1})$ is the instance after doing the DM Reduction from $\mathcal{I}_{0}$.
    Let $F_{i}$ with bipartition $V(F_{i}) = A_{i} \cup B_{i}$ be the auxiliary subgraph of $G_{i}$, where $i \in \Set{0, 1}$.
    
    It is clear that $\mu(\mathcal{I}_{0}) = \mu(\mathcal{I})$ since no edge is newly marked and $k_{0} = k$ holds.
    
    Assume that $\hat{A}_{0} \subseteq A_{0}$ and $\hat{B}_{0} \subseteq B_{0}$ are deleted.
    We note that the DM Reduction cannot be applied if and only if $\hat{A}_{0} = \hat{B}_{0} = \varnothing$.
    Observe that $A_{1} = A_{0} \setminus \hat{A}_{0}$, $B_{1} = B_{0} \setminus \hat{B}_{0}$ and $k_{1} = k_{0} - |\hat{B}_{0}|$.
    Thus, we have
    \begin{equation*}
         \mu(\mathcal{I}) - \mu(\mathcal{I}_{1}) = \mu(\mathcal{I}_{0}) - \mu(\mathcal{I}_{1}) = \left(k_{0} - \frac{2}{3} |A_{0}| \right) - \left(k_{1} - \frac{2}{3} |{A}_{1}|\right) = |\hat{B}_{0}| - \frac{2}{3} |\hat{A}_{0}|.
    \end{equation*}
    
    If the DM Reduction cannot be applied, then $\hat{A}_{0} = \hat{B}_{0} = \varnothing$, which already implies that $\mu(\mathcal{I}) - \mu(\mathcal{I}_{1}) = 0$.
    Otherwise, the DM Reduction can be applied after deleting $v$.
    In this case, we have $v \in B$.
    Additionally, according to Lemma~\ref{STRUCTURE: DM DECOMPOSITION}, $\mathcal{I}$ is a good instance implying that $N_{G}(\hat{A}_{0}) = \hat{B}_{0} \cup \Set{v}$ and $|\hat{B}_{0} \cup \Set{v}| > |\hat{A}_{0}|$.
    Therefore, we obtain that $|\hat{A}_{0}| = |\hat{B}_{0}|$.
    Thus we get $\mu(\mathcal{I}) - \mu(\mathcal{I}_{1}) \geq 0$.
    The lemma holds.
\end{proof}

\begin{lemma} \label{MEASURE: HIDE TERMINAL}
    Given a good instance $\mathcal{I} = (G = (V, E), T, M, k)$, let $F$ be the auxiliary subgraph of $G$ with bipartition $V(F) = A \cup B$, and $t$ be a terminal in $T$.
    Let $\mathcal{I}_{1}$ be the instance obtained from $\mathcal{I}$ by first hiding $t$ and then doing the DM Reduction \textup{(}cf.~Fig.~\ref{FIG: HIDE TERMINAL + DM}\textup{)}.
    Then $\mathcal{I}_{1}$ is a good instance such that
    \begin{itemize}
        \item If $t \in A$, it holds $ \mu(\mathcal{I}) - \mu(\mathcal{I}_{1}) \geq 4/3$; and
        \item If $t \not\in A$, it holds $\mu(\mathcal{I}) - \mu(\mathcal{I}_{1}) \geq 0$.
    \end{itemize}
\end{lemma}

\begin{figure}[t!]
    \centering
        \begin{tikzpicture}
            [
            scale = 0.8,
            nonterminal/.style={draw, shape = circle, fill = white, inner sep = 2pt},
            terminal/.style={draw, shape = circle, fill = black, inner sep = 2pt},
            ]
            \node[] (A) at(1, 0) {$A$};
            \node[terminal, label={[yshift = 0mm]$t_{1}$}] (t1) at(3, 0) {};
            \node[terminal, label={[yshift = 0mm]$t_{2}$}] (t2) at(5, 0) {};
            \node[terminal, label={[yshift = 0mm]$t$}] (t3) at(6, 0) {};
            \node[terminal, label={[yshift = 0mm]$t_{4}$}] (t4) at(8, 0) {};
            \node[terminal, label={[yshift = 0mm]$t_{5}$}] (t5) at(10, 0) {};
            \node[terminal, label={[yshift = 0mm]$t_{6}$}] (t6) at(11, 0) {};

            \node[] (B) at(1, -2) {$B$};
            \node[nonterminal, label={[yshift = -7mm]$u_{1}$}] (v1) at(2, -2) {};
            \node[nonterminal, label={[yshift = -7mm]$u_{2}$}] (v2) at(3, -2) {};
            \node[nonterminal, label={[yshift = -7mm]$u_{3}$}] (v3) at(5, -2) {};
            \node[nonterminal, label={[yshift = -7mm]$u_{4}$}] (v4) at(6, -2) {};
            \node[nonterminal, label={[yshift = -7mm]$u_{5}$}] (v5) at(7, -2) {};
            \node[nonterminal, label={[yshift = -7mm]$v_{6}$}] (v6) at(9, -2) {};
            \node[nonterminal, label={[yshift = -7mm]$u_{7}$}] (v7) at(10, -2) {};

            \draw[ultra thick, red]
            (t1) -- (v1) (t1) -- (v2) (t1) -- (v3) (t1) -- (v4)
            (t2) -- (v3) (t2) -- (v5)
            (t3) -- (v3) (t3) -- (v4) (t3) -- (v5)
            (t5) -- (v7);
            \draw
            (t4) -- (v4) (t4) -- (v6) (t4) -- (v5)
            (t5) -- (v4)
            (t6) -- (v5);
            \draw[rounded corners, dotted]
            (0.5, 1.25) rectangle (7, -.75)
            (0.5, -1.25) rectangle (8.25, -3.25);
        \end{tikzpicture}
        \caption{
        The graph $G$, where black vertices are terminals, white vertices are non-terminals, thick and red edges are marked edges, and edges between two non-terminals are not presented in the graph; the auxiliary subgraph is $F$ with bipartition $V(F) = A \cup B$, where $A = \Set{t_{1}, t_{2}, t_{3}}$ and $B = \Set{u_{1}, u_{2}, u_{3}, u_{4}, u_{5}}$ (denoted by dotted boxes).
        After hiding $t$, the terminal $t_{2}$ becomes isolated, and the DM Reduction cannot be applied.}
        \label{FIG: HIDE TERMINAL + DM}
\end{figure}
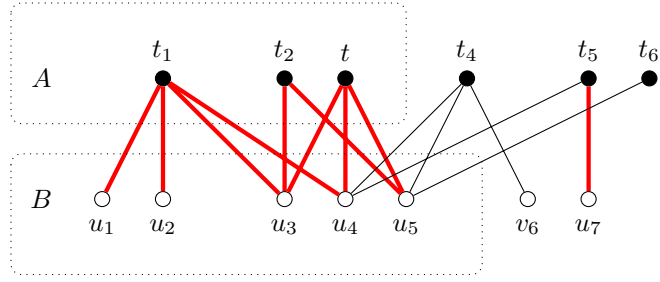

\begin{proof}
    Let instance $\mathcal{I}_{0} = (G_{0}, T_{0}, M_{0}, k_{0})$ denote the instance after hiding $t$ from $\mathcal{I}$.
    Then, $\mathcal{I}_{1} = (G_{1}, T_{1}, M_{1}, k_{1})$ is the instance after doing the DM Reduction from $\mathcal{I}_{0}$.
    Let $F_{i}$ with bipartition $V(F_{i}) = A_{i} \cup B_{i}$ be the auxiliary subgraph of $G_{i}$, where $i \in \Set{0, 1}$.
    
    After hiding terminal $t$, a non-terminal is removed if and only if it is adjacent to $t$ via a marked edge.
    Thus, $B_{0} = B \setminus N_{M}(t)$ and $k_{0} = k - |N_{M}(t)|$ hold.
    It follows that
    \begin{equation*}
        \mu(\mathcal{I}) - \mu(\mathcal{I}_{0}) = \left(k - \frac{2}{3}|A|\right) - \left(k_{0} - \frac{2}{3}|A_{0}|\right) = |N_{M}(t)| - \frac{2}{3}(|A| - |A_{0}|).
    \end{equation*}
    
    Notice that after hiding $t$, the only deleted terminal is $t$.
    Besides, a terminal $t' \neq t$ is removed from $A$ if and only if it becomes an isolated vertex.
    It follows that $N_{G}(t') \subseteq N_{M}(t)$ if $t' \in A \setminus A_{0}$.
    
    Assume that $\hat{A}_{0} \subseteq A_{0}$ and $\hat{B}_{0} \subseteq B_{0}$ are deleted after doing the DM Reduction.
    We note that the DM Reduction cannot be applied if and only if $\hat{A}_{0} = \hat{B}_{0} = \varnothing$.
    Observe that $k_{1} = k_{0} - |\hat{B}_{0}|$ and $|A_{1}| = |A_{0}| - |\hat{A}_{0}|$.
    Thus, we derive that
    \begin{equation*}
        \begin{aligned}
        \mu(\mathcal{I}) - \mu(\mathcal{I}_{1}) 
        = &|N_{M}(t)| - \frac{2}{3}(|A| - |A_{0}|) + |\hat{B}_{0}| - \frac{2}{3}|\hat{A}_{0}| \\
        = &|N_{M}(t) \cup \hat{B}_{0}| - \frac{2}{3}(|A| - |A_{1}|) \\
        \geq &|N_{M}(t) \cup \hat{B}_{0}| - \frac{2}{3}|A \setminus A_{1}|. \\
        \end{aligned}
    \end{equation*}
    
    Consider a terminal $t' \in A \setminus A_{1}$.
    If $t' \in A \setminus A_{0}$, we know $N_{G}(t') \subseteq N_{M}(t)$.
    Otherwise, we have $t' \in \hat{A}_{0}$, and all neighbours of $t'$ in $G$ are deleted, which indicates that $N_{G}(t') \subseteq \hat{B}_{0} \cup N_{M}(t)$.
    It follows that $N_{G}(A \setminus A_{1}) \subseteq \hat{B}_{0} \cup N_{M}(t)$.
    Since $\mathcal{I}$ is good, by Lemma~\ref{STRUCTURE: DM DECOMPOSITION}, we have $A \setminus A_{1} = \varnothing$ or $|N_{G}(A \setminus A_{1})| > |A \setminus A_{1}|$.
    
    If $A \setminus A_{1} = \varnothing$ holds, every terminal in $A$ is not deleted, which implies that $t \not\in A$.
    In this case, we have
    \begin{equation*}
        \mu(\mathcal{I}) - \mu(\mathcal{I}_{1}) \geq |N_{M}(t) \cup \hat{B}_{0}| - \frac{2}{3}|A \setminus A_{1}| \geq  |N_{M}(t) \cup \hat{B}_{0}| \geq 0.
    \end{equation*}
    If $|N_{G}(A \setminus A_{1})| > |A \setminus A_{1}|$ holds, we have
    \begin{equation*}
        \mu(\mathcal{I}) - \mu(\mathcal{I}_{1}) \geq |N_{M}(t) \cup \hat{B}_{0}| - \frac{2}{3}|A \setminus A_{1}| \geq |N_{G}(A \setminus A_{1})| - \frac{2}{3}|A \setminus A_{1}| \geq \frac{4}{3}.
    \end{equation*}
    Therefore, we complete our proof.
\end{proof}

\begin{lemma} \label{MEASURE: HIDE NONTERMINAL}
    Given a good instance $\mathcal{I} = (G = (V, E), T, M, k)$, let $F$ be the auxiliary subgraph of $G$ with bipartition $V(F) = A \cup B$, and $v$ be a non-terminal in $V \setminus T$.
    Let $\mathcal{I}_{1}$ be the instance obtained from $\mathcal{I}$ by first hiding $v$ and then doing the DM Reduction \textup{(}cf.~Fig.~\ref{FIG: HIDE NONTERMINAL + DM}\textup{)}.
    Then $\mathcal{I}_{1}$ is a good instance such that $\mu(\mathcal{I}) - \mu(\mathcal{I}_{1}) \geq 0$.
    
    Furthermore, if every terminal  \textup{(}resp. non-terminal\textup{)} is adjacent to at least two non-terminals \textup{(}resp. terminals\textup{)} and no two $2$-degree terminals have identical neighbours in $G$, it satisfies that
    \begin{itemize}
        \item If $v \in B$, it holds 
        \begin{equation*}
            \mu(\mathcal{I}) - \mu(\mathcal{I}_{1}) \geq \min\Set{\frac{2}{3}|N_{G}(v) \cap T| - \frac{1}{3}|N_{M}(v) \cap A|, \frac{4}{3}}.
        \end{equation*}
        \item If $v \not\in B$, it holds 
        \begin{equation*}
            \mu(\mathcal{I}) - \mu(\mathcal{I}_{1}) \geq \min\Set{\frac{2}{3}|N_{G}(v) \cap T| + \frac{1}{3}|N_{M}(v) \cap A|, \frac{4}{3}} = \frac{4}{3}.
        \end{equation*}
    \end{itemize}
\end{lemma}

\begin{figure}[t!]
    \centering
        \begin{tikzpicture}
            [
            scale = 0.8,
            nonterminal/.style={draw, shape = circle, fill = white, inner sep = 2pt},
            terminal/.style={draw, shape = circle, fill = black, inner sep = 2pt},
            ]
            \node[] (A) at(1, 0) {$A$};
            \node[terminal, label={[yshift = 0mm]$t_{1}$}] (t1) at(3, 0) {};
            \node[terminal, label={[yshift = 0mm]$t_{2}$}] (t2) at(5, 0) {};
            \node[terminal, label={[yshift = 0mm]$t_{3}$}] (t3) at(6, 0) {};
            \node[terminal, label={[yshift = 0mm]$t_{4}$}] (t4) at(8, 0) {};
            \node[terminal, label={[yshift = 0mm]$t_{5}$}] (t5) at(10, 0) {};
            \node[terminal, label={[yshift = 0mm]$t_{6}$}] (t6) at(11, 0) {};

            \node[] (B) at(1, -2) {$B$};
            \node[nonterminal, label={[yshift = -7mm]$u_{1}$}] (v1) at(2, -2) {};
            \node[nonterminal, label={[yshift = -7mm]$u_{2}$}] (v2) at(3, -2) {};
            \node[nonterminal, label={[yshift = -7mm]$u_{3}$}] (v3) at(5, -2) {};
            \node[nonterminal, label={[yshift = -7mm]$u_{4}$}] (v4) at(6, -2) {};
            \node[nonterminal, label={[yshift = -7mm]$u_{5}$}] (v5) at(7, -2) {};
            \node[nonterminal, label={[yshift = -7mm]$v$}] (v6) at(9, -2) {};
            \node[nonterminal, label={[yshift = -7mm]$u_{7}$}] (v7) at(10, -2) {};

            \draw[ultra thick, red]
            (t1) -- (v1) (t1) -- (v2) (t1) -- (v3) (t1) -- (v4)
            (t2) -- (v3) (t2) -- (v5)
            (t3) -- (v3) (t3) -- (v4) (t3) -- (v5)
            (t5) -- (v7);
            \draw
            (t4) -- (v4) (t4) -- (v6) (t4) -- (v5)
            (t5) -- (v4)
            (t6) -- (v5);
            \draw[rounded corners, dotted]
            (0.5, 1.25) rectangle (7, -.75)
            (0.5, -1.25) rectangle (8.25, -3.25);
            \draw[rounded corners, dashed]
            (4, 1) rectangle (9, -.5)
            (4, -1.5) rectangle (7.75, -3);
        \end{tikzpicture}
        \caption{
        The graph $G$, where black vertices are terminals, white vertices are non-terminals, thick and red edges are marked edges, and edges between two non-terminals are not presented in the graph; the auxiliary subgraph is $F$ with bipartition $V(F) = A \cup B$, where $A = \Set{t_{1}, t_{2}, t_{3}}$ and $B = \Set{u_{1}, u_{2}, u_{3}, u_{4}, u_{5}}$ (denoted by dotted boxes).
        After hiding $v$, the DM Reduction can be applied.
        When doing the DM Reduction, $\hat{A} = \Set{t_{2}, t_{3}, t_{4}}$ and $\hat{B} = \Set{u_{3}, u_{4}, u_{5}}$ (denoted by dashed boxes) are deleted.}
        \label{FIG: HIDE NONTERMINAL + DM}
\end{figure}
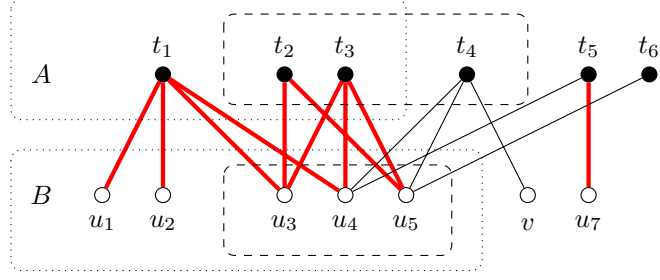

\begin{proof}
    Let instance $\mathcal{I}_{0} = (G_{0}, T_{0}, M_{0}, k_{0})$ denote the instance after hiding $v$ from $\mathcal{I}$.
    Then, $\mathcal{I}_{1} = (G_{1}, T_{1}, M_{1}, k_{1})$ is the instance after doing the DM Reduction from $\mathcal{I}_{0}$.
    Let $F_{i}$ with bipartition $V(F_{i}) = A_{i} \cup B_{i}$ be the auxiliary subgraph of $G_{i}$, where $i \in \Set{0, 1}$.

    After hiding $v$, a vertex is removed if and only if it is a terminal adjacent to $v$ via a marked edge.
    Thus, $k_{0} = k - |N_{M}(v)|$ and $A \setminus A_{0} = A \cap N_{M}(v)$.
    It follows that
    \begin{equation*}
        \begin{aligned}
        \mu(\mathcal{I}) - \mu(\mathcal{I}_{0}) 
        &= (k - \frac{2}{3} |A|) - (k_{0} -  \frac{2}{3} |A_{0}|) \\
        &= (k - k_{0}) + \frac{2}{3}|A_{0} \setminus A| - \frac{2}{3}|A \setminus A_{0}| \\
        &= |N_{M}(v)| + \frac{2}{3}|A_{0} \setminus A| - \frac{2}{3}|A \cap N_{M}(v)|. \\
        \end{aligned}
    \end{equation*}   
    It is easy to see that $\mu(\mathcal{I}) - \mu(\mathcal{I}_{0}) \geq 0$ since $|A \cap N_{M}(v)| \leq |N_{M}(v)|$.
    Thus, if the DM Reduction cannot be applied, we have $\mu(\mathcal{I}_{0}) = \mu(\mathcal{I}_{1})$, leading that $\mu(\mathcal{I}) - \mu(\mathcal{I}_{1}) \geq 0$.
    
    Next, we consider what terminals belong to set $A_{0} \setminus A$.
    On the one hand, a terminal $t \in A_{0} \setminus A$ must be adjacent to $v$ via an unmarked edge.
    On the other hand, $t$ should not be an isolated vertex after hiding $v$, which implies that the terminal $t$ is adjacent to at least one vertex distinct from $v$.
    Thus, if every terminal is adjacent to at least two non-terminals, we derive that $A_{0} \setminus A = (N_{G}(v) \cap T) \setminus N_{M}(v)$.
    It follows that
    \begin{equation*}
        \begin{aligned}
        \mu(\mathcal{I}) - \mu(\mathcal{I}_{0})
        &= |N_{M}(v)| + \frac{2}{3}|A_{0} \setminus A| - \frac{2}{3}|A \cap N_{M}(v)| \\
        &= |N_{M}(v)| + \frac{2}{3}|(N_{G}(v) \cap T) \setminus N_{M}(v)| - \frac{2}{3}|A \cap N_{M}(v)| \\
        &= |N_{M}(v)| + \frac{2}{3}(|N_{G}(v) \cap T| - |N_{M}(v)|) - \frac{2}{3}|A \cap N_{M}(v)| \\
        &\geq \frac{2}{3}|N_{G}(v) \cap T| - \frac{1}{3}|A \cap N_{M}(v)|. \\
        \end{aligned}
    \end{equation*}
    If the DM Reduction cannot be applied on $\mathcal{I}_{0}$, we can derive that
    \begin{equation*}
        \mu(\mathcal{I}) - \mu(\mathcal{I}_{1}) = \mu(\mathcal{I}) - \mu(\mathcal{I}_{0}) \geq \frac{2}{3}|N_{G}(v) \cap T| - \frac{1}{3}|A \cap N_{M}(v)|.
    \end{equation*}
    Furthermore, if $v \not\in B$, then $v$ is not adjacent to any terminal in $A$, leading that $|A \cap N_{M}(v)| = 0$.
    In the case that $v$ belongs to $B$ and it is adjacent to at least two terminals, we further have
    \begin{equation}\label{EQ: DM NOT APPLY}
        \mu(\mathcal{I}) - \mu(\mathcal{I}_{1}) = \mu(\mathcal{I}) - \mu(\mathcal{I}_{0}) \geq \frac{2}{3}|N_{G}(v) \cap T| \geq \frac{4}{3}.
    \end{equation}
    
    Now, we assume that the DM Reduction can be applied on $\mathcal{I}_{0}$.
    Suppose $\hat{A}_{0} \subseteq A_{0}$ and $\hat{B}_{0} \subseteq B_{0}$ are deleted.
    We know that $\hat{A}_{0}$ and $\hat{B}_{0}$ are non-empty, and $k_{1} = k_{0} - |\hat{B}_{0}| = k - |N_{M}(v)| - |\hat{B}_{0}|$ holds.
    Consider a terminal $t' \in A \setminus A_{1}$.
    If $t' \in A \cap N_{M}(v)$ it is deleted when hiding $v$; otherwise, it is deleted when doing the DM Reduction which means that $t' \in A \cap \hat{A}_{0}$.
    It follows that
    \begin{equation*}
        \begin{aligned}
            \mu(\mathcal{I}) - \mu(\mathcal{I}_{1})
            = &(k - \frac{2}{3} |A|) - (k_{1} -  \frac{2}{3} |A_{1}|) \\
            = &|N_{M}(v)| + |\hat{B}_{0}| + \frac{2}{3} |A_{1} \setminus A| - \frac{2}{3}|A \setminus A_{1}| \\
            \geq &|N_{M}(v)| + |\hat{B}_{0}| + \frac{2}{3} |A_{1} \setminus A| - \frac{2}{3} (|A \cap \hat{A}_{0}| + |A \cap N_{M}(v)|) \\
            \geq &\frac{1}{3}|N_{M}(v)| + |\hat{B}_{0}| + \frac{2}{3} |A_{1} \setminus A| - \frac{2}{3} |A \cap \hat{A}_{0}|. \\
        \end{aligned}
    \end{equation*}

    Now, we analysis the lower bound of $\mu(\mathcal{I}) - \mu(\mathcal{I}_{1})$ and there are two cases.
    
    \textbf{Case 1.1}: $A \cap \hat{A}_{0}$ is non-empty.
    We observe that in graph $G$, $v$ is not adjacent to any terminal in $\hat{A}_{0}$.
    This is because all the terminals adjacent to $v$ via a marked edge are deleted after hiding $v$ and they do not appear in the graph $G_{0}$.
    Hence we know $\hat{B}_{0} = N_{G}(\hat{A}_{0})$. 
    Besides, we have $B = N_{G}(A)$, and thus $B \cap \hat{B}_{0} = N_{G}(A \cap \hat{A}_{0})$ holds.
    since $\mathcal{I}$ is good and $A \cap \hat{A}_{0}$ is non-empty, we get $|B \cap \hat{B}_{0}| > |A \cap \hat{A}_{0}|$.
    Then we derive that the decrease of the measure is  
    \begin{equation*}
        \mu(\mathcal{I}) - \mu(\mathcal{I}_{1}) \geq |\hat{B}_{0}| - \frac{2}{3} |A \cap \hat{A}_{0}| \geq 1 + \frac{1}{3}|A \cap \hat{A}_{0}| \geq \frac{4}{3}.
    \end{equation*}
    
    \textbf{Case 1.2}: $A \cap \hat{A}_{0}$ is empty.
    In this case, we can directly obtain that 
    \begin{equation}\label{EQ: DECREASE>=1}
        \mu(\mathcal{I}) - \mu(\mathcal{I}_{1}) \geq \frac{1}{3}|N_{M}(v)| + |\hat{B}_{0}| + \frac{2}{3}|A_{1} \setminus A| \geq \frac{1}{3}|N_{M}(v)| + 1 \geq 1.
    \end{equation}
    Thus, we have proven the measure does not increase.

    Finally, we show that $\mu(\mathcal{I}) - \mu(\mathcal{I}_{1}) \geq 4/3$ always holds when the input graph satisfies the condition in the lemma.
    We consider two subcases.
    
    \textbf{Case 2.1:} $v \in B$.
    For this subcase, $v$ is adjacent to at least one terminal in $A$ via a marked edge.
    Hence, set $N_{M}(v)$ is non-empty, and we get $\mu(\mathcal{I}) - \mu(\mathcal{I}_{1}) \geq 1/3 + 1 = 4/3$.
    This completes that the measure is decreased by at least $4/3$ for $v \in B$.

    \textbf{Case 2.2:} $v \not\in B$.
    We assume to the contrary that if the DM Reduction can be applied and $\mu(\mathcal{I}) - \mu(\mathcal{I}_{1}) < 4/3$.
    According to \eqref{EQ: DECREASE>=1}, we can obtain that $|N_{M}(v)| = 0$, $|\hat{B}_{0}| = 1$ and $|A_{1} \setminus A| = 0$.
    It means that for every terminal $t$ adjacent to $v$, it satisfies that
    \begin{enumerate}
        \item $\deg_{G}(t) \geq 2$ holds according to the condition in the lemma;
        \item $tv$ is unmarked and $t \not\in A$ since $N_{M}(v)$ is empty;
        \item $t$ is in $A_{0}$ since $\deg_{G}(t) \geq 2$ and $v$ is the unique non-terminal deleted after hiding $v$;
        \item $t$ is deleted when doing the DM Reduction (i.e., $t \in \hat{A}_{0}$) since $A_{1} \subseteq A$ but $t \not\in A$; and
        \item $t$ has exactly two neighbours in $G$ which are $v$ and the unique vertex in $\hat{B}_{0}$ (i.e., $N_{G}(t) = \Set{v} \cup \hat{B}_{0}$) since $|\hat{B}_{0}| = 1$.
    \end{enumerate}
    Above all, we derive that all terminals in $N_{G}(v)$ are $2$-degree vertices and have the same neighbours.
    By the condition in the lemma, vertex $v$ is adjacent to at least two terminals, contradicting the condition that no two $2$-degree terminals have identical neighbours in $G$.
    Combine with \eqref{EQ: DM NOT APPLY}, we conclude that for any vertex $v \not\in B$, it holds
    \begin{equation*}
        \mu(\mathcal{I}) - \mu(\mathcal{I}_{1}) \geq \min\Set{\frac{2}{3}|N_{G}(v) \cap T|, \frac{4}{3}} = \frac{4}{3}.
    \end{equation*}
\end{proof}

\subsection{An Algorithm for Good Instances}

We now give the algorithm \texttt{GoodAlg} that solves the good instances.
When introducing a step, we assume all previous steps cannot be applied.

\begin{splitstep} \label{SPLIT REDUCTION: MUESURE}
    If $|A| > k$ or even $\mu(\mathcal{I}) < 0$, return No and quit.
    If $|T| \leq k$, return Yes and quit.
\end{splitstep}

The safeness of this step directly follows from Lemma~\ref{SAFENESS: MEASURE}.

\begin{splitstep} \label{SPLIT REDUCTION: 0TERMINAL}
    Delete any vertex in $G$ that is not contained in any $T$-triangle or marked edge, and then do the DM Reduction on the instance.
\end{splitstep}

\begin{splitstep} \label{SPLIT REDUCTION: 1TERMINAL}
    If there exists a non-terminal $v$ such that $|N_{G}(v) \cap T| = 1$, hide $v$ and then do the DM Reduction on the instance.
\end{splitstep}


The safeness of Steps~\ref{SPLIT REDUCTION: 0TERMINAL} and~\ref{SPLIT REDUCTION: 1TERMINAL} are trivial.
Based on Lemmas~\ref{SAFENESS: OPERATIION},~\ref{MEASURE: DELETE} and~\ref{MEASURE: HIDE NONTERMINAL}, after applying Step~\ref{SPLIT REDUCTION: 0TERMINAL} or Step~\ref{SPLIT REDUCTION: 1TERMINAL}, the resulting instance is good, and the measure $\mu$ does not increase.
Furthermore, every terminal is in some $T$-triangle after Step~\ref{SPLIT REDUCTION: 1TERMINAL}.
It follows that every terminal (resp. non-terminal) is adjacent to at least two non-terminals (resp. terminals).

\begin{splitstep} \label{SPLIT REDUCTION: 2TERMINAL}
    This step deals with some $2$-degree terminals in $T \setminus A$, and there are two cases.
    \begin{enumerate}
        \item Let $t$ be a $2$-degree terminal in $T \setminus A$. 
        If $t$ is adjacent to exactly one marked edge, hide $t$ and do the DM Reduction.
        \item Let $t$ and $t'$ be two $2$-degree terminals in $T \setminus A$.
        If $t$ and $t'$ have the same neighbours and none of them is adjacent to a marked edge, delete one of them and do the DM reduction.
    \end{enumerate}
\end{splitstep}

After this step, one can easily find the condition in Lemma~\ref{MEASURE: HIDE NONTERMINAL} holds.

\begin{lemma} \label{SPLIT CORRECTNESS: 2TERMINAL}
    Step~\ref{SPLIT REDUCTION: 2TERMINAL} is safe.
    After applying Step~\ref{SPLIT REDUCTION: 2TERMINAL}, the resulting instance is good, and the measure $\mu$ does not increase.
\end{lemma}

\begin{proof}
    According to Lemma~\ref{SAFENESS: OPERATIION} the resulting instance is a good instance.
    Based on Lemmas~\ref{MEASURE: DELETE},~\ref{MEASURE: HIDE TERMINAL} and~\ref{MEASURE: HIDE NONTERMINAL}, it follows that the measure $\mu$ does not increase after applying Step~\ref{SPLIT REDUCTION: 2TERMINAL}.
    We only need to prove the safeness of this step.

    We consider the first case in Step~\ref{SPLIT REDUCTION: 2TERMINAL}.
    The terminal $t$ is contained in exactly one $T$-triangle $tvu$ and one marked edge, say $tv \in M$.
    If there exists a solution $S$ containing $t$, then $S' = (S \setminus \Set{t}) \cup \Set{v}$ is also a solution since $v$ hits the marked edge $tv$ and the unique $T$-triangle $tvu$.
    Hence, there exists a solution not containing $t$.

    Next, we consider the second case in Step~\ref{SPLIT REDUCTION: 2TERMINAL}.
    Let $\mathcal{I} = (G, T, M, k)$ be the input instance.
    Let $\mathcal{I}_{0} = (G_{0}, T_{0}, M_{0}, k_{0})$ be the instance after removing one of $t$ and $t'$ from $\mathcal{I}$, say $t$.
    For any solution $S$ to $\mathcal{I}$, it is also a solution to $\mathcal{I}_{0}$, since $G_{0} = G - t$ is a subgraph of $G$.
    Let $S_{0}$ be a solution to $\mathcal{I}_{0}$.
    If $S_{0}$ contains $t'$, we replace $t'$ with a neighbour of $t'$ in $S_{0}$ to get $S'_{0}$.
    $S'_{0}$ is also a solution to $\mathcal{I}$ because every $T$-triangle containing $t$ also containing its two neighbours.
    Therefore, Step~\ref{SPLIT REDUCTION: 2TERMINAL} is safe.
\end{proof}

\begin{splitstep} \label{SPLIT BRANCH: 2NONTERMINAL}
    If there exists a non-terminal $v \in B$ adjacent to exactly one terminal $t$ via a marked edge and exactly one terminal $t'$ via an unmarked edge, we branch into two instances by either
    \begin{itemize}
        \item hiding the vertex $v$ and doing the DM Reduction; or
        \item hiding the vertex $t$ and doing the DM Reduction.
    \end{itemize}
\end{splitstep}

\begin{splitstep} \label{SPLIT BRANCH: UNMARKED NONTERMINAL}
    If there exists a non-terminal $v \in V \setminus T$ incident to at least one unmarked edge, we branch into two instances by either
    \begin{itemize}
        \item deleting the vertex $v$, decreasing $k$ by $1$, and doing the DM Reduction; or
        \item hiding $v$ and doing the DM Reduction (cf.~Fig.~\ref{FIG: NONTERMINAL DEGREE 2}).
    \end{itemize}
\end{splitstep}

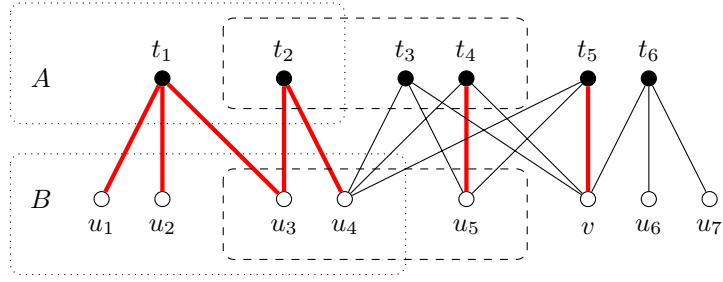
\begin{figure}[t!]
    \centering
        \begin{tikzpicture}
            [
            scale = 0.8,
            nonterminal/.style={draw, shape = circle, fill = white, inner sep=2pt},
            terminal/.style={draw, shape = circle, fill = black, inner sep=2pt},
            ]
            \node[] (A) at(1, 0) {$A$};
            \node[terminal, label={[yshift = 0mm]$t_{1}$}] (t1) at(3, 0) {};
            \node[terminal, label={[yshift = 0mm]$t_{2}$}] (t2) at(5, 0) {};
            \node[terminal, label={[yshift = 0mm]$t_{3}$}] (t3) at(7, 0) {};
            \node[terminal, label={[yshift = 0mm]$t_{4}$}] (t4) at(8, 0) {};
            \node[terminal, label={[yshift = 0mm]$t_{5}$}] (t5) at(10, 0) {};
            \node[terminal, label={[yshift = 0mm]$t_{6}$}] (t6) at(11, 0) {};

            \node[] (B) at(1, -2) {$B$};
            \node[minimum height = 5mm,minimum width = 15mm] (v0) at(0.25, -2) {};
            \node[nonterminal, label={[yshift = -7mm]$u_{1}$}] (v1) at(2, -2) {};
            \node[nonterminal, label={[yshift = -7mm]$u_{2}$}] (v2) at(3, -2) {};
            \node[nonterminal, label={[yshift = -7mm]$u_{3}$}] (v3) at(5, -2) {};
            \node[nonterminal, label={[yshift = -7mm]$u_{4}$}] (v4) at(6, -2) {};
            \node[nonterminal, label={[yshift = -7mm]$u_{5}$}] (v5) at(8, -2) {};
            \node[nonterminal, label={[yshift = -7mm]$v$}] (v6) at(10, -2) {};
            \node[nonterminal, label={[yshift = -7mm]$u_{6}$}] (v7) at(11, -2) {};
            \node[nonterminal, label={[yshift = -7mm]$u_{7}$}] (v8) at(12, -2) {};

            \draw[ultra thick, red]
            (t1) -- (v1) (t1) -- (v2) (t1) -- (v3)
            (t2) -- (v3) (t2) -- (v4)
            (t4) -- (v5) (t5) -- (v6);
            \draw
            (t3) -- (v4) (t3) -- (v5) (t3) -- (v6)
            (t4) -- (v4) (t4) -- (v6)
            (t5) -- (v4) (t5) -- (v5)
            (t6) -- (v6) (t6) -- (v7) (t6) -- (v8);
            \draw[rounded corners, dotted]
            (0.5, 1.25) rectangle (6, -0.75)
            (0.5, -1.25) rectangle (7, -3.25);
            \draw[rounded corners, dashed]
            (4, 1) rectangle (9, -.5)
            (4, -1.5) rectangle (9, -3);
        \end{tikzpicture}
        \caption{
        The graph $G$, where black vertices are terminals, white vertices are non-terminals, thick and red edges are marked edges, and edges between two non-terminals are not presented in the graph; the auxiliary subgraph is $F$ with bipartition $V(F) = A \cup B$, where $A = \Set{t_{1}, t_{2}, t_{3}, t_{4}}$ and $B = \Set{u_{1}, u_{2}, u_{3}, u_{4}, v_{5}}$ (denoted by dotted boxes).
        For the non-terminal $v \in V \setminus (T \cap B)$, after hiding $v$, the DM Reduction can be applied.
        When doing the DM Reduction, $\hat{A} = \Set{t_{2}, t_{3}, t_{4}}$ and $\hat{B} = \Set{v_{3}, v_{4}, v_{5}}$ (denoted by dashed boxes) are deleted.}
        \label{FIG: NONTERMINAL DEGREE 2}
\end{figure}

\begin{lemma} \label{SPLIT CORRECTNESS: NONTERMINAL}
    After applying Step~\ref{SPLIT BRANCH: 2NONTERMINAL} or~\ref{SPLIT BRANCH: UNMARKED NONTERMINAL}, the resulting instance in each branch is good, and we can get a branching vector not worse than $(1, 4/3)$ in both steps.
\end{lemma}

\begin{proof}
    According to Lemma~\ref{SAFENESS: OPERATIION} two instances in the two sub-branches are good instances.

    Next, we consider the safeness of these two steps.
    It is trivial that Step~\ref{SPLIT BRANCH: UNMARKED NONTERMINAL} is safe.
    The safeness of Step~\ref{SPLIT BRANCH: 2NONTERMINAL} is based on the following facts: there always exists a minimum solution containing exactly one of $v$ and $t$.
    Indeed, on the one hand, any minimum solution contains at least one of them since $vt$ is marked.
    On the other hand, if there is a minimum solution $S$ containing both $v$ and $t$, it is not hard to see that $(S \setminus \Set{v}) \cup \Set{t'}$ is a solution.
    This is because $v$ is adjacent to only two terminals.
    
    Finally, we analyze how much the measure can reduce in each branch.
    In Step~\ref{SPLIT BRANCH: 2NONTERMINAL}, $v$ is a non-terminal in $B$, and it holds $|N_{M}(v)| = 1$ and $|N_{G}(v) \cap T| = 2$.
    Furthermore, the condition in Lemma~\ref{MEASURE: HIDE NONTERMINAL} holds since Steps~\ref{SPLIT REDUCTION: 1TERMINAL} and~\ref{SPLIT REDUCTION: 2TERMINAL} cannot be applied.
    Then we get that the decreasing of the first branch is at least $1$ by Lemma~\ref{MEASURE: HIDE NONTERMINAL}.
    In the second branch, vertex $v \in B$ implies that the terminal $t$ must belong to $A$, which means that the decrease of the measure is at least $4/3$ according to Lemma~\ref{MEASURE: HIDE TERMINAL}.
    Thus, the branching vector is not worse than $(1, 4/3)$ in Step~\ref{SPLIT BRANCH: 2NONTERMINAL}.

    In the first branch of Step~\ref{SPLIT BRANCH: UNMARKED NONTERMINAL}, by Lemma~\ref{MEASURE: DELETE}, we know that the measure does not decrease after deleting $v$ and doing the DM Reduction.
    Besides, we also put $v$ into the solution and decrease $k$ by $1$.
    Hence, the measure will be decreased by at least $1$.
    Next, we consider the second branch.
    If $v$ belongs to $B$, we have $|N_{G}(v) \cap T| \geq 3$ since Step~\ref{SPLIT BRANCH: 2NONTERMINAL} cannot be applied.
    By choice of $v$ in Step~\ref{SPLIT BRANCH: UNMARKED NONTERMINAL}, the set $(N_{G}(v) \cap T) \setminus A$ is non-empty.
    Furthermore, the condition in Lemma~\ref{MEASURE: HIDE NONTERMINAL} holds since Steps~\ref{SPLIT REDUCTION: 1TERMINAL} and~\ref{SPLIT REDUCTION: 2TERMINAL} cannot be applied.
    Hence, the measure is decreased by at least $2/3 \times 3 - 1/3 \times 2 = 4/3$ according to Lemma~\ref{MEASURE: HIDE NONTERMINAL}.
    Otherwise, $v$ does not belong to $B$.
    The measure also decreases at least $4/3$ by Lemma~\ref{MEASURE: HIDE NONTERMINAL}.
    Therefore, the branching vector is not worse than $(1, 4/3)$ in Step~\ref{SPLIT BRANCH: UNMARKED NONTERMINAL}.
\end{proof}

One can easily find that every edge between a terminal and a non-terminal is marked if Step~\ref{SPLIT BRANCH: UNMARKED NONTERMINAL} cannot be applied.
Thus, we have $T = A$, and Step~\ref{SPLIT REDUCTION: MUESURE} will be applied and return the answer.
Therefore, we obtain the following result.

\begin{lemma} \label{TIME: GOOD}
    \textsc{SFVS-S} can be solved in time $\mathcal{O}^{*}(1.8192^{k})$.
\end{lemma}

\begin{proof}
    \texttt{GoodAlg} contains only two branching operations in Steps~\ref{SPLIT BRANCH: 2NONTERMINAL} and~\ref{SPLIT BRANCH: UNMARKED NONTERMINAL}.
    By Lemmas~\ref{MEASURE: DELETE},~\ref{MEASURE: HIDE TERMINAL}, and~\ref{MEASURE: HIDE NONTERMINAL}, their branching vectors are not worse than $(1, 4/3)$ whose branching factor is $1.81918$.
    Thus, we conclude that \texttt{GoodAlg} solves good instances of \textsc{SFVS-S} in time $\mathcal{O}^{*}(1.81918^{k})$.  
    According to Lemma~\ref{ALG: GOOD = SFVS-S}, \textsc{SFVS-S} can be solved in time $\mathcal{O}^{*}(1.81918^{k}) \leq \mathcal{O}^{*}(1.8192^{k})$.
\end{proof}

\section{Algorithms for \textsc{SFVS in Chordal Graphs}} \label{SEC: SFVS-C}

In this section, we address \textsc{SFVS} in general chordal graphs.
We mention that we will use the traditional parameter $k$ as the measure to analyze the algorithm in this section.
Thus, the analysis will be totally different from that of \texttt{GoodAlg} in the previous section.
Many properties required in \texttt{GoodAlg} do not hold here.
The algorithm in this section will only need to call \texttt{GoodAlg} as a sub-algorithm and use Lemma~\ref{TIME: GOOD} for the analysis.

Our algorithm for \textsc{SFVS-C}, denoted by \texttt{WholeAlg}, is a recursive algorithm that contains several steps.
After completing a step, \texttt{WholeAlg} recursively calls itself on the newly generated instances.
Each step is introduced under the assumption that none of the previous steps can be applied to the current instance.

For the sake of presentation, we divide the algorithm into two parts.
In the first part, we introduce some reduction rules and branching rules to deal with several easy cases and simplify the instance.
If none of the steps in the first part can be applied, we call the instance a ``thin'' instance.
In a thin instance, if all terminals are simplicial, we can easily reduce it to a good instance of \textsc{SFVS-S} and solve it by calling \texttt{GoodAlg}.
However, if there are ``inner'' terminals (terminals not being simplicial), we employ a divide-and-conquer technique in the second part.
This technique involves branching on a minimal separator containing inner terminals.
In each branch, we will obtain a good instance of \textsc{SFVS-S} for each sub-instance and call \texttt{GoodAlg} to solve it.

\subsection{Part~\Rome{1}: Simplifying the Instance} \label{SIMPLIFICATION}

This subsection presents the first part of \texttt{WholeAlg}, which contains nine steps.
The first seven steps are almost trivial and involve simple reductions and branching rules.

\begin{chordalstep} \label{CHORDAL REDUCTION: 0TERMINAL}
    Delete every vertex not contained in any $T$-triangle or marked edge.
\end{chordalstep}

\begin{chordalstep} \label{CHORDAL REDUCTION: BRIDGE}
    Delete all unmarked bridges from the graph.
\end{chordalstep}

Apparently, Steps~\ref{CHORDAL REDUCTION: 0TERMINAL} and~\ref{CHORDAL REDUCTION: BRIDGE} are safe.
These operations keep the graph chordal and do not change the set of $T$-triangles and the set of marked edges in the graph.

\begin{chordalstep} \label{CHORDAL BRANCH: DEGREE2 MARED}
    If there is a vertex $v$ adjacent to at least two marked edges, we branch into two instances by either
    \begin{itemize}
        \item removing $v$ from the graph and decreasing $k$ by $1$; or
        \item removing all vertices in $N_{M}(v)$ and decreasing $k$ by $|N_{M}(v)| \geq 2$.
    \end{itemize}
\end{chordalstep}

\begin{lemma} \label{SAFENESS: DEGREE2 MARKED}
    Step~\ref{CHORDAL BRANCH: DEGREE2 MARED} is safe, and the branching vector is not worse than $(1, 2)$.
\end{lemma}

\begin{proof}
    If there is a solution $S$ containing $v$, then $S \setminus \Set{v}$ is a solution to the instance induced by $G - v$.
    If $v$ is not contained in any solution $S$, then $S$ must contain all vertices in $N_{M}(v)$.
    The branching rule is safe, and we can get a branching vector $(1, |N_{M}(v)|)$.
    Since the size of $N_{M}(v)$ is at least $2$, the branching vector is not worse than $(1, 2)$.
\end{proof}

After the first three steps, any vertex is contained in a $T$-triangle or adjacent to at most one marked edge.
We will focus on vertices of degree at most two.

\begin{chordalstep} \label{CHORDAL REDUCTION: DEGREE2 0MARKED}
    If there is a $2$-degree vertex $v$ not incident to a marked edge, mark the edge between two neighbours of $v$ and delete $v$.
\end{chordalstep}

\begin{lemma}
    Step~\ref{CHORDAL REDUCTION: DEGREE2 0MARKED} is safe.
\end{lemma}

\begin{proof}
    Since Step~\ref{SPLIT REDUCTION: 0TERMINAL} cannot be applied, for any $2$-degree vertex $v$ with neighbours $N_{G}(v) = \Set{v_{1}, v_{2}}$, the three vertices $\Set{v, v_{1}, v_{2}}$ form a $T$-triangle.
    If there exists a solution containing $v$, we can replace $v$ with either $v_{1}$ or $v_{2}$.
    In either case, we obtain a solution of the same size that does not contain $v$.
    Therefore, Step~\ref{CHORDAL REDUCTION: DEGREE2 0MARKED} is safe.
\end{proof}

\begin{chordalstep} \label{CHORDAL REDUCTION: DEGREE2 1MARKED}
    If there is a $1$-degree or $2$-degree vertex $v$ incident to the marked edge $vu \in M$, delete $u$ from the graph and decrease $k$ by $1$.
\end{chordalstep}

\begin{lemma} \label{CHORDAL CORRECTNESS: DEGREE2 1MARKED}
    Step~\ref{CHORDAL REDUCTION: DEGREE2 1MARKED} is safe.
\end{lemma}

\begin{proof}
    After Step~\ref{CHORDAL BRANCH: DEGREE2 MARED}, we know that vertex $v$ is incident to at most one marked edge $vu$.
    At least one of $u$ and $v$ must be in the solution $S$.
    If $u$ is in $S$, but $v$ is not, we can replace $u$ with $v$ in $S$ to get another solution.
    This is because any triangle containing $v$ also contains $u$, and any marked edge incident to $v$ is also incident to $u$.
\end{proof}

Next, the main task of Steps~\ref{CHORDAL REDUCTION: DOMINATING TERMINALS} and~\ref{CHORDAL BRANCH: DEGREE3 SIMPLICIAL} is to deal with some simplicial vertices.
For a simplicial vertex $v$, it possesses the property that any neighbour $u$ of $v$ satisfies $N_{G}[v] \subseteq N_{G}[u]$.

\begin{chordalstep} \label{CHORDAL REDUCTION: DOMINATING TERMINALS}
    If there is a marked edge $vu \in M$ such that $v$ is a non-terminal or $u$ is a terminal, and $N_{G}[v] \subseteq N_{G}[u]$, then delete $u$, and decrease $k$ by $1$.
\end{chordalstep}

\begin{lemma} \label{CHORDAL CORRECTNESS: DOMINATING TERMINALS}
    Step~\ref{CHORDAL REDUCTION: DOMINATING TERMINALS} is safe.
\end{lemma}

\begin{proof}
    Since $vu$ is a marked edge, at least one of $v$ and $u$ is in any solution $S$.
    We show that if $v$ is in $S$ but $u$ is not, then $S' = (S \setminus \Set{v}) \cup \Set{u}$ is a solution containing $u$.
    Assume to the contrary that $S'$ is not a solution. 
    We know that there is a $T$-triangle $vab$ containing $v$ in $G - S'$.
    Thus $uab$ is also a triangle in $G$ since $N_{G}[v] \subseteq N_{G}[u]$.
    If $v$ is a non-terminal, at least one of $a$ and $b$ is a terminal.
    If $v$ is a terminal, $u$ must be a terminal by the condition.
    It follows that $uab$ is a $T$-triangle.
    These two vertices, $a$ and $b$, are not in $S'$ or $S$.
    We get a contradiction that $u$, $a$, and $b$ form a $T$-triangle in $G - S$.
\end{proof}

\begin{chordalstep} \label{CHORDAL BRANCH: DEGREE3 SIMPLICIAL}
    Let $Q$ be a simplicial clique of size at least $4$.
    If there is a simplicial vertex $v \in Q$ adjacent to a terminal $t$, branch into two instances by either
    \begin{itemize}
        \item removing $t$ from the graph and decreasing $k$ by $1$; or
        \item removing $Q \setminus \Set{t, v}$ from the graph and decreasing $k$ by $|Q| - 2 \geq 2$.
    \end{itemize}
\end{chordalstep}

\begin{lemma} \label{CHORDAL CORRECTNESS: DEGREE3 SIMPLICIAL}
    Step~\ref{CHORDAL BRANCH: DEGREE3 SIMPLICIAL} is safe, and the branching vector is not worse than $(1, 2)$.
\end{lemma}

\begin{proof}
    The edge $vt$ is not marked; otherwise, Step~\ref{CHORDAL REDUCTION: DOMINATING TERMINALS} should be applied.
    To prove this lemma, we show that if $t$ is not in the solution, there also is a solution that does not contain the vertex $v$.

    If $t$ is not in a solution $S$, then $S$ contains at least $|Q| - 2$ vertices in $Q$.
    Otherwise, three vertices in $Q \setminus S$ (including the terminal $t$) would be left and form a $T$-cycle.
    We can replace $S\cap Q$ with the $|Q| - 2$ vertices $Q\setminus \Set{t, v}$ to get another solution $S' = (S \setminus Q) \cup (Q\setminus \Set{t, v})$.
    This is because in $G - S'$, $v$ is only adjacent to $t$, and $v$ will not be in any $T$-cycle or marked edge.

    Since the size of $Q$ is at least $4$, we can decrease $k$ by at least $|Q| - 2 \geq 2$.
    We get a branching vector $(1, 2)$.
\end{proof}

\paragraph*{Reductions Based on Small Separators}

Now, we introduce some reduction rules based on separators of size one or two.
Let $Q$ be a small separator in the graph.
As mentioned in Subsection~\ref{CHORDAL AND SPLIT}, due to the properties of the chordal graph, we know that $Q$ forms a clique.
Our reduction rules aim to replace certain connected components $Z$ of $G - Q$ with a single terminal or two terminals, potentially marking some newly created edges in the process.
If $Q$ is a separator of size $1$, we can reduce the instance to one of two cases in Fig.~\ref{FIG: SEPQRATOR1}(b) and (c).
This property is formally described in Lemma~\ref{STRUCTURE: SEPQRATOR1}.
Recall that $s(\mathcal{I})$ denotes the size of the minimum solution to the instance $\mathcal{I}$.

\begin{lemma} \label{STRUCTURE: SEPQRATOR1}
    Assume that $Q = \Set{v}$ is a separator of size $1$.
    Let $Z$ be a connected component of $G - Q$, $\mathcal{I}_{1}$ be the instance induced by $Z$, and $\mathcal{I}_{2}$ be the instance induced by $Z \cup Q$.
    \begin{enumerate}
        \item If $s(\mathcal{I}_{1}) = s(\mathcal{I}_{2})$, then $s(\mathcal{I}) = s(\mathcal{I}_{1}) + s(\mathcal{\bar{I}}_{1})$, where $\mathcal{\bar{I}}_{1}$ is the instance obtained by deleting $Z$; and
        \item if $s(\mathcal{I}_{1}) + 1 = s(\mathcal{I}_{2})$, then $s(\mathcal{I}) = s(\mathcal{I}_{1}) + s(\mathcal{\bar{I}}_{2})$, where $\mathcal{\bar{I}}_{2}$ is the instance obtained by replacing $Z$ with a single terminal $t$ in $\mathcal{I}$ and marking the new edge $vt$.
    \end{enumerate}
\end{lemma}

\begin{proof}
    Let $S$ be a minimum solution to instance $\mathcal{I}$.
    It satisfies that $S \setminus Z$ is a solution to the instance $\mathcal{\bar{I}}_{1}$, which is induced by $G - Z$.
    Clearly, we also have that $|S \cap Z| \geq s(\mathcal{I}_{1})$ and $|S \cap (Z \cup Q)| \geq s(\mathcal{I}_{2})$.
    Let $S_{i}$ be a minimum solution to instance $\mathcal{I}_{i}$.
    
    Suppose that $s(\mathcal{I}_{1}) = s(\mathcal{I}_{2})$.
    Clearly, $(S \setminus Z) \cup S_{2}$ is also a solution to $\mathcal{I}$ since $\Set{v}$ is a separator.
    In addition, we have $| (S \setminus Z) \cup S_{2}| \leq |S| - s(\mathcal{I}_{1}) + s(\mathcal{I}_{2}) = |S|$, which yields that $| (S \setminus Z) \cup S_{2}|$ is a minimum solution.
    We know that $S \setminus Z$ is a solution to the instance $\mathcal{\bar{I}}_{1}$.
    Then we derive that $s(\mathcal{I}) = s(\mathcal{I}_{1}) + s(\mathcal{\bar{I}}_{1})$.
    
    Suppose that $s(\mathcal{I}_{1}) + 1 = s(\mathcal{I}_{2})$.
    Apparently, $(S \setminus Z) \cup \Set{v} \cup S_{1}$ is a solution to $\mathcal{I}$ since $\Set{v}$ is a separator.
    In addition, we have $| (S \setminus (Z \cup Q)) \cup \Set{v} \cup S_{1}| \leq |S| - s(\mathcal{I}_{2}) + 1 + s(\mathcal{I}_{1}) = |S|$, which leads that $(S \setminus Z) \cup \Set{v} \cup S_{1}$ is a minimum solution.
    Notice that $(S \setminus (Z \cup Q)) \cup \Set{v}$ is a minimum solution to $\mathcal{\bar{I}}_{2}$.
    Therefore, we obtain that $s(\mathcal{I}) = s(\mathcal{\bar{I}}_{2}) + s(\mathcal{I}_{1})$.
\end{proof}

\begin{figure}[t!]
    \centering
    \subfloat[The instance $\mathcal{I}$.]{
        \begin{tikzpicture}
            [
            scale = 1,
            nonterminal/.style={draw, shape = circle, fill = white, inner sep = 2pt},
            terminal/.style={draw, shape = circle, fill = black, inner sep = 2pt},
            ]
            \node (Z) at(0, -1.2) {$Z$};
            \node[inner sep = 0pt] (x1) at(-1, 1) {};
            \node[inner sep = 0pt] (x2) at(-0.33, 1) {};
            \node[inner sep = 0pt] (x3) at(0.33, 1) {};
            \node[inner sep = 0pt] (x4) at(1, 1) {};

            \node[inner sep = 0pt] (y1) at(-1, -1) {};
            \node[inner sep = 0pt] (y2) at(-0.33, -1) {};
            \node[inner sep = 0pt] (y3) at(0.33, -1) {};
            \node[inner sep = 0pt] (y4) at(1, -1) {};
            \node[nonterminal, label={[xshift = 4mm, yshift = -3mm]$v$}] (v) at(0, 0) {};

            \draw
            (v) -- (x1) (v) -- (x2) (v) -- (x3) (v) -- (x4)
            (v) -- (y1) (v) -- (y2) (v) -- (y4);
            \draw[red, ultra thick]
            (v) -- (y3);
            \draw[dashed] (0, 1) circle [x radius = 1.5, y radius = 0.6];
            \draw[dashed] (0, -1) circle [x radius = 1.5, y radius = 0.6];
        \end{tikzpicture}
    }
    \subfloat[The instance $\mathcal{\bar{I}}_{1}$.]{
        \begin{tikzpicture}
            [
            scale = 1,
            nonterminal/.style={draw, shape = circle, fill = white, inner sep = 2pt},
            terminal/.style={draw, shape = circle, fill = black, inner sep = 2pt},
            ]
            \node[inner sep = 0pt] (x1) at(-1, 1) {};
            \node[inner sep = 0pt] (x2) at(-0.33, 1) {};
            \node[inner sep = 0pt] (x3) at(0.33, 1) {};
            \node[inner sep = 0pt] (x4) at(1, 1) {};

            \node[nonterminal, label={[xshift = 4mm, yshift = -3mm]$v$}] (v) at(0, 0) {};

            \draw
            (v) -- (x1) (v) -- (x2) (v) -- (x3) (v) -- (x4);
            \draw[dashed] (0, 1) circle [x radius = 1.5, y radius = 0.6];
            \draw[dashed, white] (0, -1) circle [x radius = 1.5, y radius = 0.6];
        \end{tikzpicture}
    }
    \subfloat[The instance $\mathcal{\bar{I}}_{2}$.]{
        \begin{tikzpicture}
            [
            scale = 1,
            nonterminal/.style={draw, shape = circle, fill = white, inner sep = 2pt},
            terminal/.style={draw, shape = circle, fill = black, inner sep = 2pt},
            ]
            \node[inner sep = 0pt] (x1) at(-1, 1) {};
            \node[inner sep = 0pt] (x2) at(-0.33, 1) {};
            \node[inner sep = 0pt] (x3) at(0.33, 1) {};
            \node[inner sep = 0pt] (x4) at(1, 1) {};

            \node[terminal, label={[yshift = -7mm]$t$}] (t) at(0, -1) {};
            \node[nonterminal, label={[xshift = 4mm, yshift = -3mm]$v$}] (v) at(0, 0) {};

            \draw
            (v) -- (x1) (v) -- (x2) (v) -- (x3) (v) -- (x4);
            \draw[dashed] (0, 1) circle [x radius = 1.5, y radius = 0.6];
            \draw[dashed, white] (0, -1) circle [x radius = 1.5, y radius = 0.6];

            \draw[red, ultra thick]
            (v) -- (t);
        \end{tikzpicture}
    }
    \caption{The instances $\mathcal{I}$, $\mathcal{\bar{I}}_{1}$ and $\mathcal{\bar{I}}_{2}$.
    The black vertices denote terminals, and the white vertices denote non-terminals.
    $Q = \Set{v}$ is a separator of size one in $G$.
    In (a), the connected component $Z$ in $G - Q$ is in the dashed ellipse.}
    \label{FIG: SEPQRATOR1}
\end{figure}

If $Q$ is a separator of size $2$, we can reduce the instance to one of six different cases in Fig.~\ref{FIG: SEPQRATOR2} (b) to (g).
This property is formally described in the following lemma.

\begin{lemma} \label{STRUCTURE: SEPQRATOR2}
    Assume that $Q = \Set{v, u}$ is a separator of size $2$.
    Let $Z$ be a connected component of $G - Q$.
    For each subset $Q' \subseteq Q$, let $\mathcal{I}(Q')$ be the instance induced by $Z \cup Q'$.
    If 
    $$s(\mathcal{I}(\varnothing)) + x = s(\mathcal{I}(\Set{v})) + y = s(\mathcal{I}(\Set{u})) + z = s(\mathcal{I}(\Set{v,u})),$$
    then $s(\mathcal{I}) = s(\mathcal{I}(\varnothing)) + s(\mathcal{I}'_{xyz})$, where
    \begin{enumerate}
        \item $\mathcal{I}'_{1} = \mathcal{I}'_{000}$ is the instance obtained by deleting $Z$;
        \item $\mathcal{I}'_{2} = \mathcal{I}'_{111}$ is the instance obtained by replacing $Z$ with a single terminal $t$ adjacent to $v$ and $u$ in $\mathcal{I}$ and not marking any edge incident on $t$;
        \item $\mathcal{I}'_{3} = \mathcal{I}'_{110}$ is the instance obtained by replacing $Z$ with a single terminal $t$ adjacent to $v$ and $u$ in $\mathcal{I}$ and only marking the new edge $vt$;
        \item $\mathcal{I}'_{4} = \mathcal{I}'_{101}$ is the instance obtained by replacing $Z$ with a single terminal $t$ adjacent to $v$ and $u$ in $\mathcal{I}$ and only marking the new edge $ut$;
        \item $\mathcal{I}'_{5} = \mathcal{I}'_{011}$ is the instance obtained by replacing $Z$ with a single terminal $t$ adjacent to $v$ and $u$ in $\mathcal{I}$ and marking the two new edges $vt$ and $ut$;
        \item $\mathcal{I}'_{6} = \mathcal{I}'_{211}$ is the instance obtained by deleting $Z$ from $\mathcal{I}$ and introducing two vertices $t_{1}$ and $t_{2}$ such that $t_{1}$ is only adjacent to $v$ via a marked edge $vt_{1}$ and $t_{2}$ is only adjacent to $u$ via a marked edge $ut_{2}$.
    \end{enumerate}
\end{lemma}

\begin{proof}
    In the $i$th condition, for each subset $Q' \subseteq Q$, let $S_{i}(Q')$ be the minimum solutions to the instance $\mathcal{I}(Q')$, respectively.
    For every subset $Q' \subseteq Q$, define $f_{i}(Q') = |S_{i}(Q')|$.
    Moreover, assume $Z'_{i}$ ($i \in [6]$) is the vertex subset satisfying that $\mathcal{I}'_{i}$ is the instance after replacing $Z$ with the subgraph induced by $Z'_{i}$.
    For every subset $Q' \subseteq Q$, let $S''_{i}(Q')$ be a minimum solution to the instance induced by $Z'_{i} \cup Q'$ and define $g_{i}(Q') \coloneq |S''_{i}(Q')|$.
    One can easily check that for every subset $Q' \subseteq Q$, it holds 
    \begin{equation} \label{EQ: I - I'}
        f_{i}(Q') - g_{i}(Q') = f_{i}(\varnothing) - g_{i}(\varnothing) = s(\mathcal{I}(\varnothing)).
    \end{equation}
    
    Let $S$ be a minimum solution to instance $\mathcal{I}$, and define $Q_{S} \coloneq Q \setminus S$.
    Then we have $|S \cap (Z \cup Q_{S})| \geq f_{i}(Q_{S})$.
    Observe that $S \setminus (Z \cup Q_{S}) \cup S_{i}(Q_{S})$ is also a solution to $\mathcal{I}$ since $Q$ is a separator.
    Additionally, 
    \begin{equation*}
        |S \setminus (Z \cup Q_{S}) \cup S_{i}(Q_{S})| = |S \setminus (Z \cup Q_{S})| + |S_{i}(Q_{S})| \leq |S| - f_{i}(Q_{S}) + f_{i}(Q_{S}) = |S|.
    \end{equation*}
    Hence, $S \setminus (Z \cup Q_{S}) \cup S_{i}(Q_{S})$ is a minimum solution to $\mathcal{I}$, and it holds
     \begin{equation} \label{EQ: I}
        |S \setminus (Z \cup Q_{S}) \cup S_{i}(Q_{S})| = |S| - f_{i}(Q_{S}) + f_{i}(Q_{S}) = |S|.
    \end{equation}
    
    Similarly, it is clear that $S'_{i} = S \setminus (Z \cup Q_{S}) \cup S_{i}(Q''_{S})$ is also a solution to instance $\mathcal{I}'_{i}$.
    It follows that
    \begin{equation} \label{EQ: I'}
        |S'_{i}| = |S \setminus (Z \cup Q_{S}) \cup S''_{i}(Q_{S})| = |S| - f_{i}(Q_{S}) + g_{i}(Q_{S}).
    \end{equation}
    Combine \eqref{EQ: I - I'}, \eqref{EQ: I}, and \eqref{EQ: I'}, we conclude that
    \begin{equation*}
        s(\mathcal{I}) - s(\mathcal{I}'_{i}) = |S| - |S'_{i}| = f_{i}(Q_{S}) - g_{i}(Q_{S}) = s(\mathcal{I}(\varnothing)).
    \end{equation*}
    and then the desired result follows.
\end{proof}

\begin{figure}[t!]
    \centering
    \subfloat[The instance $\mathcal{I}$.]{
        \begin{tikzpicture}
            [
            scale = 1,
            nonterminal/.style={draw, shape = circle, fill = white, inner sep = 2pt},
            terminal/.style={draw, shape = circle, fill = black, inner sep = 2pt},
            ]
            \node (Z) at(0, -1.2) {$Z$};
            \node[inner sep = 0pt] (x1) at(-1, 1) {};
            \node[inner sep = 0pt] (x2) at(-0.33, 1) {};
            \node[inner sep = 0pt] (x3) at(0.33, 1) {};
            \node[inner sep = 0pt] (x4) at(1, 1) {};

            \node[inner sep = 0pt] (y1) at(-1, -1) {};
            \node[inner sep = 0pt] (y2) at(-0.33, -1) {};
            \node[inner sep = 0pt] (y3) at(0.33, -1) {};
            \node[inner sep = 0pt] (y4) at(1, -1) {};
            \node[nonterminal, label={[xshift = -4mm, yshift = -3mm]$u$}] (u) at(-0.5, 0) {};
            \node[nonterminal, label={[xshift = 4mm, yshift = -3mm]$v$}] (v) at(0.5, 0) {};

            \draw
            (u) -- (v)
            (u) -- (x1) (u) -- (x2) (u) -- (x3)
            (u) -- (y1) (u) -- (y2)
            (v) -- (x2) (v) -- (x3) (v) -- (x4)
            (v) -- (y2) (v) -- (y4);
            \draw[red, ultra thick]
            (v) -- (y3);
            \draw[dashed] (0, 1) circle [x radius = 1.5, y radius = 0.6];
            \draw[dashed] (0, -1) circle [x radius = 1.5, y radius = 0.6];
        \end{tikzpicture}
    }

    \subfloat[The instance $\mathcal{I}'_{1}$.]{
        \begin{tikzpicture}
            [
            scale = 1,
            nonterminal/.style={draw, shape = circle, fill = white, inner sep = 2pt},
            terminal/.style={draw, shape = circle, fill = black, inner sep = 2pt},
            ]
            \node[inner sep = 0pt] (x1) at(-1, 1) {};
            \node[inner sep = 0pt] (x2) at(-0.33, 1) {};
            \node[inner sep = 0pt] (x3) at(0.33, 1) {};
            \node[inner sep = 0pt] (x4) at(1, 1) {};

            \node[nonterminal, label={[xshift = -4mm, yshift = -3mm]$u$}] (u) at(-0.5, 0) {};
            \node[nonterminal, label={[xshift = 4mm, yshift = -3mm]$v$}] (v) at(0.5, 0) {};

            \draw
            (u) -- (v)
            (u) -- (x1) (u) -- (x2) (u) -- (x3)
            (v) -- (x2) (v) -- (x3) (v) -- (x4);
            \draw[dashed] (0, 1) circle [x radius = 1.5, y radius = 0.6];
            \draw[dashed, white] (0, -1) circle [x radius = 1.5, y radius = 0.6];
        \end{tikzpicture}
    }
    \subfloat[The instance $\mathcal{I}'_{2}$.]{
        \begin{tikzpicture}
            [
            scale = 1,
            nonterminal/.style={draw, shape = circle, fill = white, inner sep = 2pt},
            terminal/.style={draw, shape = circle, fill = black, inner sep = 2pt},
            ]
            \node[inner sep = 0pt] (x1) at(-1, 1) {};
            \node[inner sep = 0pt] (x2) at(-0.33, 1) {};
            \node[inner sep = 0pt] (x3) at(0.33, 1) {};
            \node[inner sep = 0pt] (x4) at(1, 1) {};

            \node[terminal, label={[yshift = -7mm]$t$}] (t) at(0, -1) {};
            \node[nonterminal, label={[xshift = -4mm, yshift = -3mm]$u$}] (u) at(-0.5, 0) {};
            \node[nonterminal, label={[xshift = 4mm, yshift = -3mm]$v$}] (v) at(0.5, 0) {};

            \draw
            (u) -- (v)
            (u) -- (x1) (u) -- (x2) (u) -- (x3)
            (v) -- (x2) (v) -- (x3) (v) -- (x4)
            (t) -- (u) (t) -- (v);
            \draw[dashed] (0, 1) circle [x radius = 1.5, y radius = 0.6];
            \draw[dashed, white] (0, -1) circle [x radius = 1.5, y radius = 0.6];
        \end{tikzpicture}
    }
    \subfloat[The instance $\mathcal{I}'_{3}$.]{
        \begin{tikzpicture}
            [
            scale = 1,
            nonterminal/.style={draw, shape = circle, fill = white, inner sep = 2pt},
            terminal/.style={draw, shape = circle, fill = black, inner sep = 2pt},
            ]
            \node[inner sep = 0pt] (x1) at(-1, 1) {};
            \node[inner sep = 0pt] (x2) at(-0.33, 1) {};
            \node[inner sep = 0pt] (x3) at(0.33, 1) {};
            \node[inner sep = 0pt] (x4) at(1, 1) {};

            \node[terminal, label={[yshift = -7mm]$t$}] (t) at(0, -1) {};
            \node[nonterminal, label={[xshift = -4mm, yshift = -3mm]$u$}] (u) at(-0.5, 0) {};
            \node[nonterminal, label={[xshift = 4mm, yshift = -3mm]$v$}] (v) at(0.5, 0) {};

            \draw
            (u) -- (v)
            (u) -- (x1) (u) -- (x2) (u) -- (x3)
            (v) -- (x2) (v) -- (x3) (v) -- (x4)
            (t) -- (v);
            \draw[dashed] (0, 1) circle [x radius = 1.5, y radius = 0.6];
            \draw[dashed, white] (0, -1) circle [x radius = 1.5, y radius = 0.6];
            \draw[red, ultra thick]
            (u) -- (t);
        \end{tikzpicture}
    }

    \subfloat[The instance $\mathcal{I}'_{4}$.]{
        \begin{tikzpicture}
            [
            scale = 1,
            nonterminal/.style={draw, shape = circle, fill = white, inner sep = 2pt},
            terminal/.style={draw, shape = circle, fill = black, inner sep = 2pt},
            ]
            \node[inner sep = 0pt] (x1) at(-1, 1) {};
            \node[inner sep = 0pt] (x2) at(-0.33, 1) {};
            \node[inner sep = 0pt] (x3) at(0.33, 1) {};
            \node[inner sep = 0pt] (x4) at(1, 1) {};

            \node[terminal, label={[yshift = -7mm]$t$}] (t) at(0, -1) {};
            \node[nonterminal, label={[xshift = -4mm, yshift = -3mm]$u$}] (u) at(-0.5, 0) {};
            \node[nonterminal, label={[xshift = 4mm, yshift = -3mm]$v$}] (v) at(0.5, 0) {};

            \draw
            (u) -- (v)
            (u) -- (x1) (u) -- (x2) (u) -- (x3)
            (v) -- (x2) (v) -- (x3) (v) -- (x4)
            (t) -- (u);
            \draw[dashed] (0, 1) circle [x radius = 1.5, y radius = 0.6];
            \draw[dashed, white] (0, -1) circle [x radius = 1.5, y radius = 0.6];
            \draw[red, ultra thick]
            (v) -- (t);
        \end{tikzpicture}
    }
    \subfloat[The instance $\mathcal{I}'_{5}$.]{
        \begin{tikzpicture}
            [
            scale = 1,
            nonterminal/.style={draw, shape = circle, fill = white, inner sep = 2pt},
            terminal/.style={draw, shape = circle, fill = black, inner sep = 2pt},
            ]
            \node[inner sep = 0pt] (x1) at(-1, 1) {};
            \node[inner sep = 0pt] (x2) at(-0.33, 1) {};
            \node[inner sep = 0pt] (x3) at(0.33, 1) {};
            \node[inner sep = 0pt] (x4) at(1, 1) {};

            \node[terminal, label={[yshift = -7mm]$t$}] (t) at(0, -1) {};
            \node[nonterminal, label={[xshift = -4mm, yshift = -3mm]$u$}] (u) at(-0.5, 0) {};
            \node[nonterminal, label={[xshift = 4mm, yshift = -3mm]$v$}] (v) at(0.5, 0) {};

            \draw
            (u) -- (v)
            (u) -- (x1) (u) -- (x2) (u) -- (x3)
            (v) -- (x2) (v) -- (x3) (v) -- (x4);
            \draw[dashed] (0, 1) circle [x radius = 1.5, y radius = 0.6];
            \draw[dashed, white] (0, -1) circle [x radius = 1.5, y radius = 0.6];
            \draw[red, ultra thick]
            (u) -- (t) (v) -- (t);
        \end{tikzpicture}
    }
    \subfloat[The instance $\mathcal{I}'_{6}$.]{
        \begin{tikzpicture}
            [
            scale = 1,
            nonterminal/.style={draw, shape = circle, fill = white, inner sep = 2pt},
            terminal/.style={draw, shape = circle, fill = black, inner sep = 2pt},
            ]
            \node[inner sep = 0pt] (x1) at(-1, 1) {};
            \node[inner sep = 0pt] (x2) at(-0.33, 1) {};
            \node[inner sep = 0pt] (x3) at(0.33, 1) {};
            \node[inner sep = 0pt] (x4) at(1, 1) {};

            \node[terminal, label={[yshift = -7mm]$t_{1}$}] (t1) at(-0.5, -1) {};
            \node[terminal, label={[yshift = -7mm]$t_{2}$}] (t2) at(0.5, -1) {};
            \node[nonterminal, label={[xshift = -4mm, yshift = -3mm]$u$}] (u) at(-0.5, 0) {};
            \node[nonterminal, label={[xshift = 4mm, yshift = -3mm]$v$}] (v) at(0.5, 0) {};

            \draw
            (u) -- (v)
            (u) -- (x1) (u) -- (x2) (u) -- (x3)
            (v) -- (x2) (v) -- (x3) (v) -- (x4);
            \draw[dashed] (0, 1) circle [x radius = 1.5, y radius = 0.6];
            \draw[dashed, white] (0, -1) circle [x radius = 1.5, y radius = 0.6];
            \draw[red, ultra thick]
            (u) -- (t1) (v) -- (t2);
        \end{tikzpicture}
    }
    \caption{The instances $\mathcal{I}$ and $\mathcal{I'}_{i}$ for $i \in [6]$.
    The black vertices denote terminals, and the white vertices denote non-terminals.
    $Q = \Set{u, v}$ is a separator of size two in $G$.
    In (a), the connected component $Z$ in $G - Q$ is in the dashed ellipse.}
    \label{FIG: SEPQRATOR2}
\end{figure}
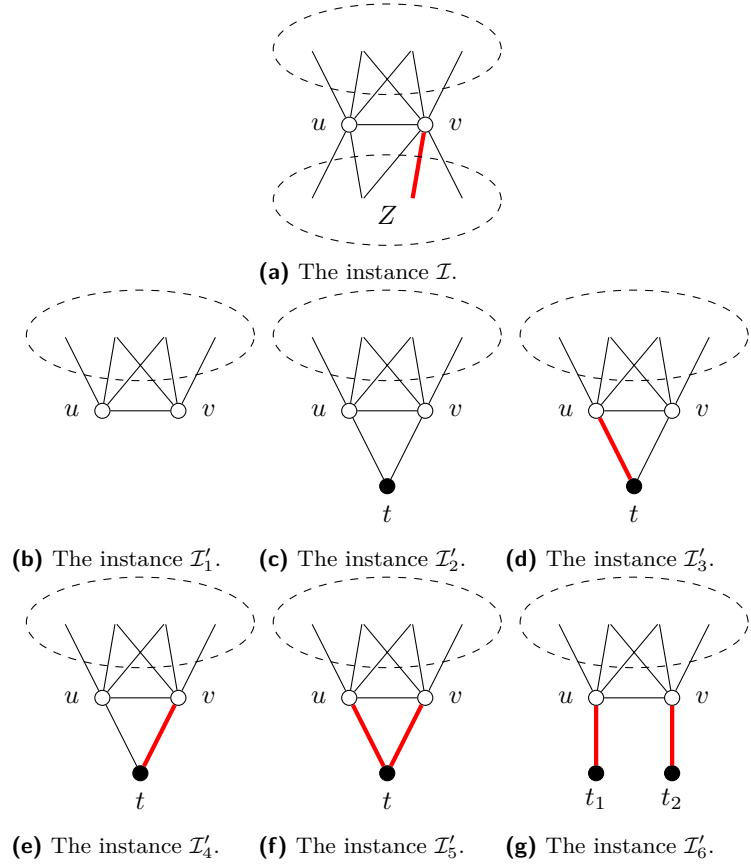

Based on Lemmas~\ref{STRUCTURE: SEPQRATOR1} and~\ref{STRUCTURE: SEPQRATOR2}, we can utilize the following reduction step to further simplify the instance.

\begin{chordalstep} \label{CHORDAL REDUCTION: SMALL SEPARATOR}
    If there is a separator $Q$ of size $1$ or $2$ with a connected component $Z$ of $G - Q$ such that the size of the minimum solution to the instance induced by $Z$ is at most $5$, we replace $Z$ with at most two vertices according to Lemmas~\ref{STRUCTURE: SEPQRATOR1} and~\ref{STRUCTURE: SEPQRATOR2}.
\end{chordalstep}

We require the size of the minimum solution to the instance induced by $Z$ to be bounded by a constant so that this step can be executed in polynomial time.

\begin{definition} [Thin Instances]
    An instance of \textsc{SFVS-C} is called \emph{thin} if none of the first eight steps in this subsection can be applied on the instance.
\end{definition}

\begin{lemma} \label{STRUCTURE: THIN}
    The following properties hold for a thin instance of \textsc{SFVS-C}.
    \begin{enumerate}[\textup{(}a\textup{)}]
        \item There is no isolated vertex in $G$ and $M$ is a matching in the graph; \label{ITEM: MATCHING}
        \item the size of every simplicial clique is at least $4$; \label{ITEM: CLIQUE SIZE4}
        \item every simplicial clique contains only one simplicial vertex, which is also the unique terminal in the simplicial clique; and \label{ITEM: SIMPLIEIAL = TERMINAL}
        \item for every separator $Q$ of size $1$ or $2$ and each connected component $Z$ in $G - Q$, the minimum solution to the instance induced by $Z$ is no less than $5$. \label{ITEM： SMALL SEPARATOR}
    \end{enumerate}
\end{lemma}

\begin{proof}   
    If a non-terminal is not adjacent to any terminal or is not incident to any marked edge, it should be deleted by Step~\ref{CHORDAL REDUCTION: 0TERMINAL}.
    Besides, if there exists a vertex adjacent to at least two marked edges, Step~\ref{CHORDAL REDUCTION: DEGREE2 0MARKED} would be applied.
    Thus we derive that Item~(\ref{ITEM: MATCHING}) is correct.

    Now we prove the correctness of Item~(\ref{ITEM: CLIQUE SIZE4}).
    For a simplicial clique $Q$, let $v$ be a simplicial vertex in $Q$.
    Because of Steps~\ref{CHORDAL REDUCTION: 0TERMINAL},~\ref{CHORDAL REDUCTION: BRIDGE}, and~\ref{CHORDAL REDUCTION: DEGREE2 0MARKED}, we have the degree $\deg_{G}(v) \geq 2$.
    Suppose that $\deg_{G}(v) = 2$.
    If $v$ is not adjacent to a marked edge, Step~\ref{CHORDAL REDUCTION: DEGREE2 0MARKED} can be applied.
    If $v$ is adjacent to a marked edge, Step~\ref{CHORDAL REDUCTION: DEGREE2 1MARKED} can be applied.
    If $v$ is adjacent to two marked edges, Step~\ref{CHORDAL BRANCH: DEGREE2 MARED} can be applied.
    This leads to a contradiction.
    Thus, we have $\deg_{G}(v) \geq 3$, which implies that $|Q| \geq 4$.
    
    Next, we consider Item~(\ref{ITEM: SIMPLIEIAL = TERMINAL}).
    Let $u$ be a simplicial vertex in a simplicial clique $Q$.
    If $u$ is a non-terminal incident to a marked edge, Step~\ref{CHORDAL REDUCTION: DOMINATING TERMINALS} is applicable.
    In addition, we have that $Q \cap T \neq \varnothing$ since Step~\ref{CHORDAL REDUCTION: 0TERMINAL} cannot be applied.  
    We have shown that $|Q| \geq 4$.
    If there is a terminal $t \neq u$, Step~\ref{CHORDAL BRANCH: DEGREE3 SIMPLICIAL} can be applied.
    It follows that $u$ must be a terminal.
    In addition, if another simplicial vertex $u' \neq u$ exists, Step~\ref{CHORDAL BRANCH: DEGREE3 SIMPLICIAL} can also be applied.
    Consequently, $u$ is the unique simplicial vertex in $Q$, and $u$ is also a terminal.
    Furthermore, $u$ is not adjacent to other terminals; otherwise, Step~\ref{CHORDAL CORRECTNESS: DEGREE3 SIMPLICIAL} is applicable.
    This leads to the correctness of Item~(\ref{ITEM: SIMPLIEIAL = TERMINAL}).

    We finally show the correctness of Item~(\ref{ITEM： SMALL SEPARATOR}).
    Let $Q$ be a separator of size $1$ or $2$.
    Assume to the contrary that there is a connected component $Z$ in $G - Q$ such that the size of the minimal solution to the instance induced by $Z$ is less than $5$.
    Based on the replacing operation in Step~\ref{CHORDAL REDUCTION: SMALL SEPARATOR}, we know that $C$ is a singleton set.
    Furthermore, we know that the unique vertex in $Z$ is a terminal of degree at most $2$, which implies that $Z \cup Q$ is a simplicial clique of size only three, a contradiction.
\end{proof}

\begin{lemma} \label{THINvsGOOD}
    For a thin instance $\mathcal{I}$, if every terminal in $G$ is simplicial, we can polynomially reduce $\mathcal{I}$ to an equivalent good instance $\mathcal{I}'$ of \textsc{SFVS-S} without increasing $k$.
\end{lemma}

\begin{proof}
    We construct $\mathcal{I}'$ from $\mathcal{I} = (G, T, M, k)$ by doing the following.
    For each marked edge $ab$ with two non-terminals, add a new terminal $t$ only adjacent to $a$ and $b$ and unmark the edge $ab$.
    Next, make the subgraph induced by non-terminals a clique by adding an edge to every pair of disconnected non-terminals.
    Eventually, we have the instance $\mathcal{I}' = (G' = (V', E'), T', M', k')$, where $k' = k$.
    
    Observe that two disconnected non-terminals in $G$ are not adjacent to the same terminal since every terminal is simplicial. 
    Thus, the two instances, $\mathcal{I}$ and $\mathcal{I}'$, are equivalent since every added edge is not contained in any $T'$-triangle in $G'$.
    It is easy to see the parameter $k$ does not change.
    
    Next, we show that $\mathcal{I}'$ is a good instance of \textsc{SFVS-S}.
    Any terminal $t_{1}$ in the original graph $G$ is simplicial, and $t_{1}$ is not adjacent to other terminals by Lemma~\ref{STRUCTURE: THIN}(\ref{ITEM: SIMPLIEIAL = TERMINAL}).
    For any newly added terminal $t_{2}$, the degree of $t_{2}$ is two, and its two neighbours are both non-terminals by the construction.
    Hence, the terminal set $G'[T']$ is an independent set.
    Thus, the resulting graph is a split graph, and $(T', V' \setminus T')$ forms a split partition of $G'$.
    Besides, we have handled marked edges with two endpoints being two non-terminals in constructing $\mathcal{I}'$.
    Notice that the degree of every terminal is at least two, and $M'$ is a matching, which leads to the auxiliary subgraph of $G'$ being an empty graph.
    Therefore, the DM reduction cannot be applied on $\mathcal{I}'$, and we obtain that $\mathcal{I}'$ is a good instance.
\end{proof}

The last step in this subsection is based on Lemma~\ref{THINvsGOOD}.
A terminal is called an \emph{inner terminal} if it is not simplicial.

\begin{chordalstep} \label{CHORDAL REDUCTION: THIN TO GOOD}
    If there is no inner terminal (i.e., all terminals are simplicial), reduce the instance to a good instance of \textsc{SFVS-S} by Lemma~\ref{THINvsGOOD}, and then call \texttt{GoodAlg} to solve it.
\end{chordalstep}

\begin{lemma} \label{TIME: THIN + GOOD}
    If the input instance has no inner terminal, the first nine steps of \rm\texttt{WholeAlg} will solve the instance in time $\mathcal{O}^{*}(\alpha^{k} + 1.6181^{k})$ if \textsc{SFVS-S} can be solved in time $\mathcal{O}^{*}(\alpha^{k})$.
\end{lemma}

In the case that the instance still has some inner terminals, the algorithm will go to the next steps in the next subsection.

\subsection{Part~\Rome{2}: The Divide-and-Conquer Procedure} \label{DIVIDING PROCEDURE}

In this subsection, we present the second part of \texttt{WholeAlg}, which focuses on the divide-and-conquer procedure to handle inner terminals.
This part follows Part~\Rome{1} of \texttt{WholeAlg}.
Hence, we always assume that the input graph in this part is thin and contains inner terminals.

\begin{definition} [Dividing Separators and Good Components]
    A separator $Q$ is called a dividing separator if
    \begin{enumerate}[\textup{(}i\textup{)}]
        \item the separator $Q$ contains an inner terminal;
        \item there exists an edge $Q_{1}Q_{2}$ in the clique tree $\mathcal{T}_{G}$ such that $Q = Q_{1} \cap Q_{2}$; and
        \item there is a connected component $X_{Q}$ in $G - Q$, such that the sub-instance induced by $G[X_{Q} \cup Q]$ contains no inner terminal.
    \end{enumerate}
The connected component $X_{Q}$ in \textup{(}iii\textup{)} is called \emph{good} \textup{(}with respect to $Q$\textup{)}.
\end{definition}

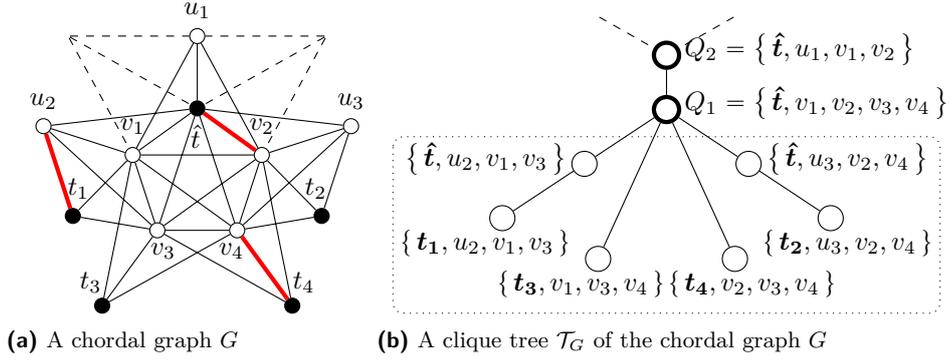
\begin{figure}[t!]
    \centering
    \subfloat[A chordal graph $G$]{
        \begin{tikzpicture}
            [
            scale = 1.05,
            nonterminal/.style={draw, shape = circle, fill = white, inner sep = 2pt},
            terminal/.style={draw, shape = circle, fill = black, inner sep = 2pt},
            ]
            \node[inner sep = 0pt] (w1) at(-1.618, 1.764) {};
            \node[inner sep = 0pt] (w2) at(1.618, 1.764) {};
            \node[nonterminal, label={[yshift = 0.5mm]$v_{1}$}] (v1) at(-0.809, 0.263) {};
            \node[nonterminal, label={[yshift = 0.5mm]$v_{2}$}] (v2) at(0.809, 0.263) {};
            \node[nonterminal, label={[xshift = 0.75mm, yshift = -6mm]$v_{3}$}] (v3) at(-0.5, -0.6885) {};
            \node[nonterminal, label={[xshift = -0.75mm, yshift = -6mm]$v_{4}$}] (v4) at(0.5, -0.6885) {};
            \node[nonterminal, label={$u_{1}$}] (u1) at(0, 1.764) {};
            
            \node[terminal, label={[xshift = -0.2mm, yshift = -7.5mm]$\hat{t}$}] (t0) at(0, 0.851) {};
            \node[nonterminal, label={$u_{2}$}] (u2) at(-1.927, 0.6265) {};
            \node[nonterminal, label={$u_{3}$}] (u3) at(1.927, 0.6265) {};
            \node[terminal, label={[xshift = 0.75mm]$t_{1}$}] (t1) at(-1.559, -0.5068) {};
            \node[terminal, label={[xshift = -0.75mm]$t_{2}$}] (t2) at(1.559, -0.5068) {};
            \node[terminal, label={[xshift = -1.5mm, yshift = -0.5mm]$t_{3}$}] (t3) at(-1.191, -1.64) {};
            \node[terminal, label={[xshift = 1.5mm, yshift = -0.5mm]$t_{4}$}] (t4) at(1.191, -1.64) {};
            
            \draw
            (u1) -- (t0) (u1) -- (v1) (u1) -- (v2)
            (t0) -- (v1) (t0) -- (v3) (t0) -- (v4)
            (v1) -- (v2) (v1) -- (v3) (v1) -- (v4) (v2) -- (v3) (v2) -- (v4) (v3) -- (v4)
            (u2) -- (t0) (u2) -- (v1) (u2) -- (v3) (u3) -- (t0) (u3) -- (v2) (u3) -- (v4)
            (t3) -- (v1) (t3) -- (v3) (t3) -- (v4) (t4) -- (v3) (t4) -- (v2)
            (t1) -- (v1) (t1) -- (v3) (t2) -- (v2) (t2) -- (v4) (t2) -- (u3);
            
            \draw[dashed]
            (w1) -- (u1) (w1) -- (t0) (w1) -- (v1) 
            (w2) -- (u1) (w2) -- (t0) (w2) -- (v2);

            \draw[red, ultra thick] (t0) -- (v2) (u2) -- (t1) (v4) -- (t4);
        \end{tikzpicture}
        }
        \hspace{5mm}
        \subfloat[The subgraph $G - Q$]{
        \begin{tikzpicture}
            [
            scale = 1.05,
            nonterminal/.style={draw, shape = circle, fill = white, inner sep = 2pt},
            terminal/.style={draw, shape = circle, fill = black, inner sep = 2pt},
            ]
            \node[inner sep = 0pt] (w1) at(-1.618, 1.764) {};
            \node[inner sep = 0pt] (w2) at(1.618, 1.764) {};
            \node[nonterminal, label={[xshift = 0.75mm, yshift = -6mm]$v_{3}$}] (v3) at(-0.5, -0.6885) {};
            \node[nonterminal, label={[xshift = -0.75mm, yshift = -6mm]$v_{4}$}] (v4) at(0.5, -0.6885) {};
            \node[nonterminal, label={$u_{1}$}] (u1) at(0, 1.764) {};
            
            \node[nonterminal, label={$u_{2}$}] (u2) at(-1.927, 0.6265) {};
            \node[nonterminal, label={$u_{3}$}] (u3) at(1.927, 0.6265) {};
            \node[terminal, label={[xshift = 0.75mm]$t_{1}$}] (t1) at(-1.559, -0.5068) {};
            \node[terminal, label={[xshift = -0.75mm]$t_{2}$}] (t2) at(1.559, -0.5068) {};
            \node[terminal, label={[xshift = -1.5mm, yshift = -0.5mm]$t_{3}$}] (t3) at(-1.191, -1.64) {};
            \node[terminal, label={[xshift = 1.5mm, yshift = -0.5mm]$t_{4}$}] (t4) at(1.191, -1.64) {};
            
            \draw
            (v3) -- (v4) (u2) -- (v3) (u3) -- (v4)
            (t3) -- (v3) (t3) -- (v4) (t4) -- (v3) 
            (t1) -- (v3) (t2) -- (v4) (t2) -- (u3);
            
            \draw[dashed]
            (w1) -- (u1) (w2) -- (u1);

            \draw[red, ultra thick] (u2) -- (t1) (v4) -- (t4);
        \end{tikzpicture}
        }

        \subfloat[A clique tree $\mathcal{T}_{G}$ of the chordal graph $G$]{
        \begin{tikzpicture}
            [
            scale = 0.72,
            clique/.style = {circle, draw}
            ]
            \node (w1) at(-1.4, 1.8) {};
            \node (w2) at(1.4, 1.8) {};
            \node[clique, ultra thick, label = {[xshift = 17.5mm, yshift = -5mm]$Q_{2} = \Set{\boldsymbol{\hat{t}}, u_{1}, v_{1}, v_{2}}$}] (Q2) at(0, 1) {};
            \node[clique, ultra thick, label = {[xshift = 19.75mm, yshift = -5mm]$Q_{1} = \Set{\boldsymbol{\hat{t}}, v_{1}, v_{2}, v_{3}, v_{4}}$}] (Q1) at(0, 0) {};
            \node[clique, label = {[xshift = -13mm, yshift = -4.5mm]$\Set{\boldsymbol{\hat{t}}, u_{2}, v_{1}, v_{3}}$}] (Q3) at(-2.4, -1.2) {};
            \node[clique, label = {[xshift = 13mm, yshift = -4.5mm]$\Set{\boldsymbol{\hat{t}}, u_{3}, v_{2}, v_{4}}$}] (Q4) at(2.4, -1.2) {};
            \node[clique, label = {[xshift = -2.25mm, yshift = -8mm]$\Set{\boldsymbol{t_{3}}, v_{1}, v_{3}, v_{4}}$}] (Q5) at(-1.5, -2.5) {};
            \node[clique, label = {[xshift = 2.25mm, yshift = -8mm]$\Set{\boldsymbol{t_{4}}, v_{2}, v_{3}, v_{4}}$}] (Q6) at(1.5, -2.5) {};
            
            \node[clique, label = {[xshift = -2.25mm, yshift = -8mm]$\Set{\boldsymbol{t_{1}}, u_{2}, v_{1}, v_{3}}$}] (Q31) at(-5, -2.5) {};
            \node[clique, label = {[xshift = 2.25mm, yshift = -8mm]$\Set{\boldsymbol{t_{2}}, u_{3}, v_{2}, v_{4}}$}] (Q41) at(5, -2.5) {};
                
            \draw[dashed]
            (w1) -- (Q2) (w2) -- (Q2);
            \draw
            (Q1) -- (Q2) (Q1) -- (Q3) (Q1) -- (Q4) (Q1) -- (Q5) (Q1) -- (Q6)
            (Q3) -- (Q31) (Q4) -- (Q41);
            
            \draw[rounded corners, dotted] 
            (-7, -0.5) rectangle (7, -3.75);
        \end{tikzpicture}
        }
        
        \subfloat[The subgraph induced by $X_{1}$]{
        \begin{tikzpicture}
            [
            scale = 1.05,
            nonterminal/.style={draw, shape = circle, fill = white, inner sep = 2pt},
            terminal/.style={draw, shape = circle, fill = black, inner sep = 2pt},
            ]
            \node[nonterminal, label={[yshift = 0.5mm]$v_{1}$}] (v1) at(-0.809, 0.263) {};
            \node[fill = white] (w) at(0, 1.2) {};
            
            \node[terminal, label={[xshift = -0.2mm, yshift = -7.5mm]$\hat{t}$}] (t0) at(0, 0.851) {};
            \node[nonterminal, label={$u_{2}$}] (u2) at(-1.927, 0.6265) {};
            \node[nonterminal, label={$u_{3}$}] (u3) at(1.927, 0.6265) {};
            \node[terminal, label={[xshift = 0.75mm]$t_{1}$}] (t1) at(-1.559, -0.5068) {};
            \node[terminal, label={[xshift = -0.75mm]$t_{2}$}] (t2) at(1.559, -0.5068) {};
            \node[terminal, label={[xshift = -1.5mm, yshift = -0.5mm]$t_{3}$}] (t3) at(-1.191, -1.64) {};
            \node[terminal, label={[xshift = 1.5mm, yshift = -0.5mm]$t_{4}$}] (t4) at(1.191, -1.64) {};
            
            \draw
            (t0) -- (v1)
            (u2) -- (t0) (u2) -- (v1) (u3) -- (t0)
            (t3) -- (v1) (t1) -- (v1) (t2) -- (u3);

            \draw[red, ultra thick] (u2) -- (t1);
            
        \end{tikzpicture}
        }
        \hspace{5mm}
        \subfloat[The subgraph induced by $X_{2}$]{
        \begin{tikzpicture}
            [
            scale = 1.05,
            nonterminal/.style={draw, shape = circle, fill = white, inner sep = 2pt},
            terminal/.style={draw, shape = circle, fill = black, inner sep = 2pt},
            ]
            \node[nonterminal, label={[yshift = 0.5mm]$v_{2}$}] (v2) at(0.809, 0.263) {};
            \node[fill = white] (w) at(0, 1.2) {};
            
            \node[terminal, label={[xshift = -0.2mm, yshift = -7.5mm]$\hat{t}$}] (t0) at(0, 0.851) {};
            \node[nonterminal, label={$u_{2}$}] (u2) at(-1.927, 0.6265) {};
            \node[nonterminal, label={$u_{3}$}] (u3) at(1.927, 0.6265) {};
            \node[terminal, label={[xshift = 0.75mm]$t_{1}$}] (t1) at(-1.559, -0.5068) {};
            \node[terminal, label={[xshift = -0.75mm]$t_{2}$}] (t2) at(1.559, -0.5068) {};
            \node[terminal, label={[xshift = -1.5mm, yshift = -0.5mm]$t_{3}$}] (t3) at(-1.191, -1.64) {};
            \node[terminal, label={[xshift = 1.5mm, yshift = -0.5mm]$t_{4}$}] (t4) at(1.191, -1.64) {};
            
            \draw
            (t0) -- (v2)
            (u2) -- (t0) (u3) -- (t0) (u3) -- (v2) 
            (t4) -- (v2) (t2) -- (v2) (t2) -- (u3);

            \draw[red, ultra thick] (u2) -- (t1);
            
        \end{tikzpicture}
        }

        \subfloat[The subgraph induced by $X_{3}$]{
        \begin{tikzpicture}
            [
            scale = 1.05,
            nonterminal/.style={draw, shape = circle, fill = white, inner sep = 2pt},
            terminal/.style={draw, shape = circle, fill = black, inner sep = 2pt},
            ]
            \node[nonterminal, label={[xshift = 0.75mm, yshift = -6mm]$v_{3}$}] (v3) at(-0.5, -0.6885) {};
            \node[fill = white] (w) at(0, 1.2) {};
            
            \node[terminal, label={[xshift = -0.2mm, yshift = -7.5mm]$\hat{t}$}] (t0) at(0, 0.851) {};
            \node[nonterminal, label={$u_{2}$}] (u2) at(-1.927, 0.6265) {};
            \node[nonterminal, label={$u_{3}$}] (u3) at(1.927, 0.6265) {};
            \node[terminal, label={[xshift = 0.75mm]$t_{1}$}] (t1) at(-1.559, -0.5068) {};
            \node[terminal, label={[xshift = -0.75mm]$t_{2}$}] (t2) at(1.559, -0.5068) {};
            \node[terminal, label={[xshift = -1.5mm, yshift = -0.5mm]$t_{3}$}] (t3) at(-1.191, -1.64) {};
            \node[terminal, label={[xshift = 1.5mm, yshift = -0.5mm]$t_{4}$}] (t4) at(1.191, -1.64) {};
            
            \draw
            (t0) -- (v3)
            (u2) -- (t0) (u2) -- (v3) (u3) -- (t0)
            (t3) -- (v3) (t4) -- (v3) (t1) -- (v3) (t2) -- (u3);

            \draw[red, ultra thick] (u2) -- (t1);
            
        \end{tikzpicture}
        }
        \hspace{5mm}
        \subfloat[The subgraph induced by $X_{4}$]{
        \begin{tikzpicture}
            [
            scale = 1.05,
            nonterminal/.style={draw, shape = circle, fill = white, inner sep = 2pt},
            terminal/.style={draw, shape = circle, fill = black, inner sep = 2pt},
            ]
            \node[nonterminal, label={[xshift = -0.75mm, yshift = -6mm]$v_{4}$}] (v4) at(0.5, -0.6885) {};
            \node[fill = white] (w) at(0, 1.2) {};
            
            \node[terminal, label={[xshift = -0.2mm, yshift = -7.5mm]$\hat{t}$}] (t0) at(0, 0.851) {};
            \node[nonterminal, label={$u_{2}$}] (u2) at(-1.927, 0.6265) {};
            \node[nonterminal, label={$u_{3}$}] (u3) at(1.927, 0.6265) {};
            \node[terminal, label={[xshift = 0.75mm]$t_{1}$}] (t1) at(-1.559, -0.5068) {};
            \node[terminal, label={[xshift = -0.75mm]$t_{2}$}] (t2) at(1.559, -0.5068) {};
            \node[terminal, label={[xshift = -1.5mm, yshift = -0.5mm]$t_{3}$}] (t3) at(-1.191, -1.64) {};
            \node[terminal, label={[xshift = 1.5mm, yshift = -0.5mm]$t_{4}$}] (t4) at(1.191, -1.64) {};
            
            \draw
            (t0) -- (v4)
            (u2) -- (t0) (u3) -- (t0) (u3) -- (v4)
            (t3) -- (v4) (t2) -- (v4) (t2) -- (u3);

            \draw[red, ultra thick] (u2) -- (t1) (v4) -- (t4);
            
        \end{tikzpicture}
        }
        
    \caption{In (a), (b), (d), (e), (f), and (g), the black vertices denote terminals $\Set{\hat{t}} \cup \Set{t_{i}}_{i = 1}^{4}$, and the white vertices denote non-terminals $\Set{u_{i}}_{i = 1}^{3} \cup \Set{v_{i}}_{i = 1}^{4}$.
    The thick red edges denote marked edges, and the dashed edges are optional.
    In (c), the nodes represent maximal cliques in $G$.
    The terminal $\hat{t}$ is an inter terminal, $Q = Q_{1} \cap Q_{2} = \Set{\hat{t}, v_{1}, v_{2}}$ is a dividing clique, where $Q_{1}$ and $Q_{2}$ are nodes with thick borders.
    A good connected component $X_{Q} = \Set{t_{i}}_{i = 1}^{4} \cup \Set{u_{2}, u_{3}, v_{3}, v_{4}}$ (w.r.t. $Q$) is formed by the union of the cliques in dotted boxes minus $Q$.
    It satisfies that $X_{0} = (X_{Q} \setminus Q_{1}) \cup \Set{\hat{t}} = \Set{t_{i}}_{i = 1}^{4} \cup \Set{\hat{t}, u_{2}, u_{3}}$ and $s_{0} = 1$.
    For the sub-instance $\mathcal{I}_{i}$ induced by $G[X_{i}] = G[X_{0} \cup \Set{v_{i}}]$ ($i \in [4]$), $s_{1} = s_{3} = 1$ and $s_{2} = s_{4} = 2$.
    Thus it holds $U_{0} = \Set{v_{1}, v_{3}}$ and $U_{1} = \Set{v_{2}, v_{4}}$.}
    \label{DIVIDING CLIQUES}
\end{figure}

When there is an inner terminal, the dividing separators exist.
We start by constructing a clique tree $\mathcal{T}_{G}$ of the graph $G$ with the set of maximal cliques $\mathcal{Q}_{G}$, which can be done in linear time by using the algorithm in~\cite{jctGavril74}.
Observe that there exist $n - 1$ edges in the clique tree, and each dividing separator corresponds to an edge in the clique tree.
By checking each edge in the clique tree, we can identify a dividing separator $Q$ and a good component $X_{Q}$ with respect to $Q$ in polynomial time.

The subsequent steps of the algorithm are based on a dividing separator.
We always use the following notations.
Let $Q$ be a dividing separator that contains an inner terminal $\hat{t}$, and suppose $Q = Q_{1} \cap Q_{2}$ for an edge $Q_{1}Q_{2}$ in the clique tree.
Let $X_{Q} \subseteq V$ be a good component w.r.t $Q$.
We can see that neither $Q_{1}$ nor $Q_{2}$ is simplicial since each contains the inner terminal $\hat{t}$.
Without loss of generality, we assume that $Q \subsetneq Q_{1} \subsetneq X_{Q} \cup Q$.
It follows that at least one simplicial vertex exists in $X_{Q} \setminus Q_{1}$ since every leaf node in the clique tree contains a simplicial vertex.

Let $Q_{1} \setminus \Set{\hat{t}} = \Set{v_{1}, v_{2}, \dots, v_{\ell}}$ and $X_{0} \coloneq (X_{Q} \setminus Q_{1}) \cup \Set{\hat{t}}$.
For each vertex $v_{i} \in Q_{1} \setminus \Set{\hat{t}}$, define $X_{i} \coloneq X_{0} \cup \Set{v_{i}}$.
Let $\mathcal{I}_{i}$ denote the sub-instance induced by $G[X_{i}]$, i.e., $\mathcal{I}_{i} = (G[X_{i}],  X_{i} \cap T,  M \cap E(G[X_{i}]), k)$.
Since $G[X_{i}]$ is a subgraph of $G[X_{Q} \cup Q]$, we know that in $\mathcal{I}_{i}$, each terminal is simplicial.
Therefore, the first nine steps of \texttt{WholeAlg} can solve each $\mathcal{I}_{i}$ in time $\mathcal{O}^{*}(\alpha^{k} + 1.6181^{k})$ by Lemma~\ref{TIME: THIN + GOOD}.
The size of the minimum solution to $\mathcal{I}_{i}$ is expressed as $s_{i} \coloneq s(\mathcal{I}_{i})$.
Refer to Fig.~\ref{DIVIDING CLIQUES} for illustrations of these concepts.

\begin{chordalstep} \label{CHORDAL BRANCH: DIVID}
    Solve each instance $\mathcal{I}_{i}$ for $0 \leq i \leq \ell$.
    If anyone returns No, the algorithm returns No to indicate the original instance is a No-instance.
    Otherwise, we can compute the value $s_{i}$ for $0 \leq i \leq \ell$.
\end{chordalstep}

If \textsc{SFVS-S} can be solved in time $\mathcal{O}^{*}(\alpha^{k})$, Step~\ref{CHORDAL BRANCH: DIVID} can be executed in time $\mathcal{O}^{*}(\alpha^{k} + 1.6181^{k})$ since each instance in Step~\ref{CHORDAL BRANCH: DIVID} takes $\mathcal{O}^{*}(\alpha^{s_{i}} + 1.6181^{s_{i}}) \leq \mathcal{O}^{*}(\alpha^{k} + 1.6181^{k})$ time, and there are $\ell + 1 = \mathcal{O}(n)$ instances.

Observe that for each index $i \in [\ell]$, it holds $s_{0} \leq s_{i} \leq s_{0} + 1$ because $G[X_{i}]$ contains only one more vertex than $G[X_{0}]$.
Define $U_{0} \coloneq \{ v_{i} \in Q_{1} \setminus \Set{\hat{t}} : s_{i} = s_{0} \}$ and $U_{1} \coloneq \{ v_{i} \in Q_{1} \setminus \Set{\hat{t}} : s_{i} = s_{0} + 1 \}$ (cf. Fig.~\ref{DIVIDING CLIQUES}).

Note that for any separator $Q'$ and any component $X'$ of $G - Q'$, $S$ is a solution to $\mathcal{I}$ if and only if $S \cap (X' \cup Q')$ is a solution to the instance induced by $X' \cup Q'$ and $S \setminus X'$ is a solution to the instance induced by $G - X'$.
This property is important and will be repeatedly used in the analysis of our algorithm.
We first give a structural property of $U_{1}$.

\begin{lemma} \label{SOLUTION SIZE1}
    If the instance $\mathcal{I}$ has a minimum solution without $\hat{t}$, there exists a minimum solution $S$ containing $U_{1}$ but not containing $\hat{t}$.
\end{lemma}

\begin{proof}
    We prove this lemma by contradiction.
    Suppose that $S$ is a minimum solution to $\mathcal{I} = (G, T, M, k)$ not containing a vertex $v_{i} \in U_{1}$ and the terminal $\hat{t} \in T$.
    In this case, $S$ must contain $Q_{1} \setminus \Set{\hat{t}, v_{i}}$; otherwise, if we remove $S$ from the graph, there would be a triangle containing $\hat{t}$ and $v_{i}$ in the remaining graph.
    Thus, we have
    \begin{equation*}
        |S \cap Q_{1}| = |Q_{1}| - 2.
    \end{equation*}
    Since $S \cap X_{i}$ is also a solution to $\mathcal{I}_{i} = (G[X_{i}],  X_{i} \cap T,  M \cap E(G[X_{i}]), k)$, we can establish that
    \begin{equation*}
        |S \cap X_{i}| \geq s_{i} = s_{0} + 1.
    \end{equation*}
    Note that $X_{i} = X_{0} \cup \Set{v_{i}}$ and $v_{i} \not\in S$.
    Thus, we get
    \begin{equation*}
        S \cap X_{0} = S \cap X_{i}.
    \end{equation*}
    By combining the above three relations, we can obtain that
    \begin{equation*}
        |S \cap (Q_{1} \cup X_{0})| = |Q_{1}| - 2 + s_{i} \geq s_{0} + |Q_{1}| - 1.
    \end{equation*}
    
    Let $S_{0}$ be a minimum solution to $\mathcal{I}_{0}$.
    It is easy to see that $S' = (S \setminus (Q_{1} \cup X_{0})) \cup (S_{0} \cup Q_{1} \setminus \Set{\hat{t}})$ is still a solution to $\mathcal{I}$.
    Indeed, on the one hand, we have $S \setminus X_{0} \subseteq S'$, which indicates that $S'  \setminus X_{0}$ is a solution to the instance induced by $G - X_{0}$.
    On the other hand, we have $S_{0} \cup Q_{1} \setminus \Set{\hat{t}} \subseteq S'$, which indicates that $S' \cap (Q_{1} \cup X_{0})$ is a solution to the instance induced by $Q_{1} \cup X_{0}$.
    Besides, by the construction of $S'$, it follows that
    \begin{equation*}
        |S'| \leq |S| - (s_{0} + |Q_{1}| - 1) + (|S_{0}| + |Q_{1}| - 1) = |S|,
    \end{equation*}
    which yields that $S$ is a minimum solution to $\mathcal{I}$.
    
    Finally, due to the construction of $S'$, we have $S' \cap Q_{1} = Q_{1} \setminus \Set{\hat{t}}$, which leads $S'$ is a minimum solution to $\mathcal{I}$ containing $Q_{1} \setminus \Set{\hat{t}}$.
\end{proof}

Based on Lemmas~\ref{SOLUTION SIZE1}, we have the following branching steps.

\begin{chordalstep} \label{CHORDAL BRANCH: BIG S0+U1}
    If $s_{0} + |U_{1}| \geq 2$, branch into two instances by either
    \begin{itemize}
        \item removing $\hat{t}$ and decreasing $k$ by $1$; or
        \item removing $X_{0} \setminus \Set{\hat{t}}$ and $U_{1}$, and decreasing $k$ by $s_{0} + |U_{1}| \geq 2$.
    \end{itemize}
\end{chordalstep}

\begin{lemma} \label{CHORDAL CORRECTNESS: BIG S0+U1}
    Step~\ref{CHORDAL BRANCH: BIG S0+U1} is safe, and the branching vector is not worse than $(1, 2)$.
\end{lemma}

\begin{proof}
    If $\hat{t} \in Q \subseteq Q_{1}$ is in the solution $S$, then $S \setminus \Set{\hat{t}}$ is a solution to the instance induced by $G - \hat{t}$; otherwise, $\hat{t}$ is not contained in the solution.

    Let $\mathcal{I}' = (G - (X_{0} \setminus \Set{\hat{t}}) - U_{1}, T \setminus ((X_{0} \setminus \Set{\hat{t}}) \cup U_{1}), k - s_{0} - |U_{1}|)$.
    We only need to show $\mathcal{I}'$ has a solution without $\hat{t}$ if and only if $\mathcal{I}'$ has a solution without $\hat{t}$.
    Let $S$ be a solution to the instance $\mathcal{I}$ such that $\hat{t} \not\in S$.
    Based on Lemma~\ref{SOLUTION SIZE1}, we assume that $U_{1} \subseteq S$.
    Notice that $S \cap X_{0}$ is a solution to the instance $\mathcal{I}_{0}$, and its size satisfies $|S \cap X_{0}| \geq s_{0}$.
    Additionally, $S \setminus (X_{0} \cup U_{1})$ is also a solution to the sub-instance $\mathcal{I}'$.
    Consequently, the size of the minimum solution to $\mathcal{I}'$ holds that
    \begin{equation*}
        s(\mathcal{I}') \leq |S \setminus(X_{0} \cup U_{1})| = |S| - |U_{1}| - |S \cap X_{0}| \leq |S| - |U_{1}| - s_{0} \leq k - (s_{0} + |U_{1}|).
    \end{equation*}

    Conversely, assume that $S'$ is a solution to the instance $\mathcal{I}'$ such that $\hat{t} \not\in S'$ and $|S'| \leq k - (s_{0} + |U_{1}|)$.
    We now show the existence of a solution $S$ to instance $\mathcal{I}$ with $|S| = k$.
    Since $S'$ must intersect all $T$-triangles in the clique $Q_{1} \setminus U_{1}$ and $\hat{t} \not\in S'$, there can be at most one vertex in $U_{0} \setminus S'$.
    We consider two cases.

    \textbf{Case 1:} $U_{0} \setminus S'$ is empty.
    In this case, we define $S = S' \cup S_{0} \cup U_{1}$, where $S_{0}$ is a minimum solution to the instance $\mathcal{I}_{0}$.
    On one hand, $S' \cup U_{1}$ is a solution to the instance induced by $G - X_{0}$.
    On the other hand, $S_{0}$ is a solution to the instance induced by $X_{0}$.
    We conclude that $S$ is a solution to instance $\mathcal{I}$.

    \textbf{Case 2:} $U_{0} \setminus S'$ is non-empty.
    Assume that $v_{j} \in Q_{1} \setminus \Set{\hat{t}}$ is the unique vertex in $U_{0} \setminus S'$.
    In this case, we let $S = S' \cup S_{0} \cup U_{1}$, where $S_{j}$ is a minimum solution to the instance $\mathcal{I}_{j}$.
    On the one hand, $S' \cup U_{1}$ is a solution to the instance induced by $G - X_{0}$.
    On the other hand, $S_{j}$ is a solution to the instance induced by $X_{j}$.
    We conclude that $S$ is a solution to instance $\mathcal{I}$.
    
    Finally, observe that $|S_{0}| = s_{0} = s_{j} = |S_{j}|$ since $v_{j} \in U_{0}$.
    Therefore, for either of the above two cases, the size of $S$ satisfies that
    \begin{equation*}
        |S| = |S'| + |U_{1}| + s_{0} = k - (|U_{1}| + s_{0}) + |U_{1}| + s_{0} = k.
    \end{equation*}
    Therefore, Step~\ref{CHORDAL BRANCH: BIG S0+U1} is safe.
    
    Finally, we obtain the branching vector $(1, s_{0} + |U_{1}|)$, which is not worse than $(1, 2)$.
\end{proof}

\begin{chordalstep} \label{CHORDAL BRANCH: BIG S0+Q1}
    If $s_{0} + |Q_{1}| \geq 4$ and $U_{0} \setminus Q \neq \varnothing$, we pick a vertex $v_{i} \in U_{0} \setminus Q$ and branch into two instances by either
    \begin{itemize}
        \item removing $\hat{t}$ and decreasing $k$ by $1$; or
        \item removing $(Q_{1} \cup X_{0}) \setminus \Set{\hat{t}, v_{i}}$ and decreasing $k$ by $s_{0} + |Q_{1}| - 2 \geq 2$.
    \end{itemize}
\end{chordalstep}

\begin{lemma} \label{CHORDAL CORRECTNESS: BIG S0+Q1}
    Step~\ref{CHORDAL BRANCH: BIG S0+Q1} is safe, and the branching vector is not worse than $(1, 2)$.
\end{lemma}

\begin{proof}
    If $\hat{t} \in Q$ is in a minimum solution $S$, then $S \setminus \Set{\hat{t}}$ is a minimum solution to the instance induced by $G - \hat{t}$; otherwise, $\hat{t}$ is not contained in the solution.
    Furthermore, for any solution $S$ to the input instance $\mathcal{I} = (G, T, M, k)$, set $Q_{1} \setminus S$ contains at most two vertices, leading that $|Q_{1}| - 2 \leq |S \cap Q_{1}|$.
    Recall that $v_{i}$ is a vertex in $U_{0} \setminus Q$ and suppose that $v_{i}$ is included in $S$.
    
    We now show $S' = (S \setminus X_{i}) \cup (Q_{1} \setminus \Set{\hat{t}, v_{i}}) \cup S_{i}$ is also a solution to $\mathcal{I}$.
    Notice that $\Set{\hat{t}}$ is a separator in $G - S'$.
    Since $S \setminus X_{i} \subseteq S'$, we have that $S' \setminus X_{i}$ is the solution to the instance induced by $V \setminus X_{i}$; and $S_{i} \subseteq S'$ implies that $S' \cap X_{i}$ is the solution to the instance $\mathcal{I}_{i}$.
    As a result, $S'$ is a solution to $\mathcal{I}$.

    On the other hand, $S \cap X_{i}$ is a solution to $\mathcal{I}_{i}$, which means that $s_{i} \leq |S \cap X_{i}|$.
    Thus, we have
    \begin{equation*}
        |S'| = |S \setminus (X_{i} \cup Q_{1})| + (|Q_{1}| - 2) + s_{i} = |S \setminus (X_{i} \cup Q_{1})| + |S \cap Q_{1}| + |S \cap X_{i}| \leq |S|.
    \end{equation*}
    which leads that $S'$ is a minimum solution to $\mathcal{I}$.
    
    Furthermore, according to the construction of $S'$, we know that there always exists a minimum solution containing $Q_{1} \setminus \Set{\hat{t}, v_{i}}$ and $S_{i}$.
    After removing $Q_{1} \setminus \Set{\hat{t}, v_{i}}$, $\Set{\hat{t}}$ is a separator in the remaining graph.
    Hence, we can further safely remove the rest vertices in $X_{i}$.

    Finally, we get the branching vector $(1, s_{0} + |Q_{1}| - 2)$.
\end{proof}

\begin{lemma} \label{STRUCTURE: FISH}
    For a thin instance that Steps~\ref{CHORDAL REDUCTION: THIN TO GOOD},~\ref{CHORDAL BRANCH: DIVID},~\ref{CHORDAL BRANCH: BIG S0+U1} and~\ref{CHORDAL BRANCH: BIG S0+Q1} cannot be applied, the following properties hold \textup{(}cf. Fig~\ref{FIG: FISH}\textup{)}.
    \begin{enumerate}[\textup{(}a\textup{)}]
        \item The value $s_{0} = 0$; \label{ITEM: S0=0}
        \item the set $U_{1}$ is equal to $Q_{1} \setminus Q$ and it contains exactly one vertex; and \label{ITEM: U1=Q1-Q=1}
        \item the set $X_{Q} \setminus Q_{1}$ is non-empty, and every vertex in $X_{Q} \setminus Q_{1}$ is simplicial. \label{ITEM: X1-Q1 SIMPLICIAL}
    \end{enumerate}
\end{lemma}

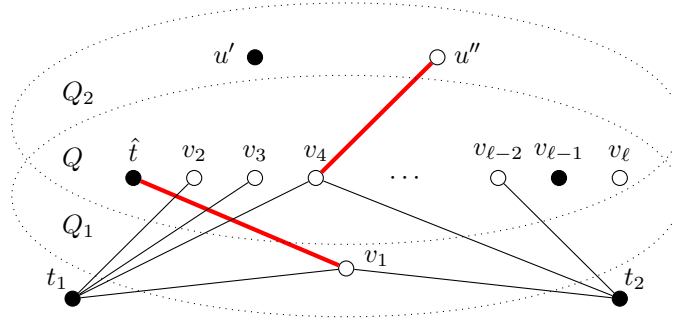
\begin{figure}[t!]
    \centering
        \begin{tikzpicture}
            [
            scale = 0.8,
            nonterminal/.style={draw, shape = circle, fill = white, inner sep=2pt},
            terminal/.style={draw, shape = circle, fill = black, inner sep=2pt},
            ]
            \node[terminal, label={[xshift = -4mm, yshift = -3mm]$u'$}] (u') at (4, 2) {};
            \node[nonterminal, label={[xshift = 4mm, yshift = -3mm]$u''$}] (u'') at (7, 2) {};

            \node[] (Q) at(1, 0.3) {$Q$};
            \node[] (Q2) at(1.1, 1.4) {$Q_{2}$};
            \node[] (Q1) at(1.1, -0.8) {$Q_{1}$};
            \node[terminal, label={[xshift = 0mm, yshift = 0mm]$\hat{t}$}] (t) at(2, 0) {};
            \node[nonterminal, label={[yshift = 0mm]$v_{2}$}] (u2) at(3, 0) {};
            \node[nonterminal, label={[xshift = 0mm, yshift = 0mm]$v_{3}$}] (u3) at(4, 0) {};
            \node[nonterminal, label={[xshift = 0mm, yshift = 0mm]$v_{4}$}] (u4) at(5, 0) {};
            \node[] (uu) at(6.5, 0) {$\cdots$};
            \node[nonterminal, label={[yshift = 0mm]$v_{\ell - 2}$}] (ul-2) at(8, 0) {};
            \node[terminal, label={[yshift = 0mm]$v_{\ell - 1}$}] (ul-1) at(9, 0) {};
            \node[nonterminal, label={[yshift = 0mm]$v_{\ell}$}] (ul) at(10, 0) {};

            \node[nonterminal, label={[xshift = 4mm, yshift = -2mm]$v_{1}$}] (u1) at(5.5, -1.5) {};
            \node[terminal, label={[xshift = -2mm, yshift = -1mm]$t_{1}$}] (t1) at(1, -2) {};
            \node[terminal, label={[xshift = 2mm, yshift = -1mm]$t_{2}$}] (t2) at(10, -2) {};

            \draw[ultra thick, red]
            (t) -- (u1) (u'') -- (u4);
            \draw
            (t1) -- (u2) (t1) -- (u3) (t1) -- (u4) (t1) -- (u1)
            (t2) -- (ul-2) (t2) -- (u4) (t2) -- (u1);
            \draw[dotted] (5.5, 0.9) circle [x radius = 5.5, y radius = 2];
            \draw[dotted] (5.5, -0.3) circle [x radius = 5.5, y radius = 2];
        \end{tikzpicture}
        \caption{A part of an instance that Steps~\ref{CHORDAL BRANCH: BIG S0+U1} and~\ref{CHORDAL BRANCH: BIG S0+Q1} cannot be applied, where black vertices denote the terminals $\Set{\hat{t}, t_{1}, t_{2}, v_{\ell - 1}, u'}$, and white vertices denote the non-terminals $\Set{v_{1}, v_{2}, \ldots, v_{\ell}, u''}$.
        $Q_{1} = \Set{\hat{t}} \cup \Set{v_{i}}_{i = 1}^{\ell}$ and $Q_{2} = \Set{\hat{t}, u', u''} \cup \Set{v_{i}}_{i = 2}^{\ell}$ are two maximal cliques, which are denoted by dotted ellipses.
        The unmarked edges between vertices in $Q_{1}$ or $Q_{2}$ are not presented, the marked edges denote thick edges.
        $Q = Q_{1} \cap Q_{2}$ is a dividing separator and $X_{Q} = \Set{\hat{t}, v_{1}, t_{1}, t_{2}}$ is a good component.
        Besides, it holds $X_{0} = \Set{\hat{t}, t_{1}, t_{2}}$ and $s_{0} = 0$.
        In the maximal clique $Q_{1}$, $U_{1} = \Set{v_{1}} = Q_{1} \setminus Q$ and $U_{0} = \Set{v_{i}}_{i = 2}^{\ell}$.
        In $X_{Q} \setminus Q_{1} = \Set{t_{1}, t_{2}}$, $t_{1}$ and $t_{2}$ are simplicial.}
        \label{FIG: FISH}
\end{figure}

\begin{proof}
    First, we know $s_{0} + |U_{1}| \leq 1$ since Step~\ref{CHORDAL BRANCH: BIG S0+U1} cannot be applied.
    If $s_{0} + |Q_{1}| \leq 3$, we can derive that $|Q| = |Q_{1} \cap Q_{2}| \leq |Q_{1}| - 1 \leq 2$.
    Furthermore, the size of the minimum solution to the instance induced by $X_{Q} \cup (Q_{1} \setminus Q)$ is bounded by $s_{0} + |Q_{1}| \leq 3$.
    This contradicts Lemma~\ref{STRUCTURE: THIN}(\ref{ITEM： SMALL SEPARATOR}).
    Therefore, according to the condition of Step~\ref{CHORDAL BRANCH: BIG S0+Q1}, we can conclude that $U_{0} \setminus Q$ is empty.
    It follows that $U_{1} = Q_{1} \setminus Q$, and the size of $U_{1}$ is exactly one, which yields that Item~(\ref{ITEM: U1=Q1-Q=1}) is correct.
    Additionally, by the relation $s_{0} + |U_{1}| \leq 1$, we have $s_{0} = 0$ and Item~(\ref{ITEM: S0=0}) is correct.
    
    Next, if $X_{Q} \setminus Q_{1}$ is empty, the sole vertex in $U_{1} = Q_{1} \setminus Q$ is simplicial, implying that $u$ is a terminal.
    However, $\hat{t}$ is a terminal adjacent to $u$, contradicting Lemma~\ref{STRUCTURE: THIN}(\ref{ITEM: SIMPLIEIAL = TERMINAL}).
    
    Finally, we assume to the contrary that a non-simplicial vertex exists in $X_{Q} \setminus Q_{1}$.
    In this case, there must be a simplicial vertex $t \in X_{Q} \setminus Q_{1}$ such that at least one neighbour of $t$ does not belong to $Q_{1}$. 
    Suppose $u$ is one such neighbour of $t$ and we will consider three cases.

    \textbf{Case 1:} $\hat{t} \in N_{G}(t)$.
    For this case, $\hat{t}$ and $u$ are adjacent since $t$ is simplicial.
    Thus, vertices $\hat{t}$, $t$, and $u$ form a $T$-triangle in $G[X_{0}]$, contradicting that $s_{0} = 0$.

    \textbf{Case 2:} there exists a vertex $v_{i} \in N_{G}(t) \cap Q$ ($i \in [\ell]$).
    For this case,  the vertices $v_{i}$, $t$ and $u$ form a $T$-triangle in $G[X_{i}]$, contradicting that $s_{i} = 0$.

    \textbf{Case 3:} no neighbour of $t$ belongs to $Q_{1}$.
    According to Lemma~\ref{STRUCTURE: THIN}(\ref{ITEM: CLIQUE SIZE4}), there must exist another vertex $u' \in N_{G}(t)$ distinct from $u$ and $v_{1}$.
    It follows that the vertices $u$, $u'$, and $t$ form a $T$-triangle, also leading to the contradiction that $s_{0} = 0$.

    Therefore, we conclude that Item~(\ref{ITEM: X1-Q1 SIMPLICIAL}) is correct. 
\end{proof}

\begin{chordalstep} \label{CHORDAL BRANCH: FISH}
    Let $U_{1} = Q_{1} \setminus Q = \Set{v_{1}}$.
    If there exists a vertex $u$ distinct from $\hat{t}$ adjacent to $v_{1}$ with a marked edge, branch into two instances by either
    \begin{itemize}
        \item removing $v_{1}$ and decreasing $k$ by $1$; or
        \item removing $\Set{\hat{t}, u}$ and decreasing $k$ by $2$.
    \end{itemize}
    Otherwise, we branch into two instances by either
    \begin{itemize}
        \item removing $\hat{t}$ and decreasing $k$ by $1$; or
        \item removing $(Q_{1} \cup Q_{2}) \setminus Q$ and decreasing $k$ by $|Q_{1} \cup Q_{2}| - |Q| = 1 + |Q_{2} \setminus Q| \geq 2$  (cf. Fig~\ref{FIG: FISH}).
    \end{itemize}
\end{chordalstep}

Note that Step~\ref{CHORDAL BRANCH: FISH} is always applicable after Step~\ref{CHORDAL BRANCH: BIG S0+Q1}.
One of two branching rules will be executed depending on whether $v_{1}$ is adjacent to a vertex distinct from $\hat{t}$ with a marked edge.

\begin{lemma} \label{CHORDAL CORRECTNESS: FISH}
    Step~\ref{CHORDAL BRANCH: FISH} is safe, and the branching vector is not worse than $(1, 2)$.
\end{lemma}

\begin{proof}
    Observe that $N_{M}(v_{1}) = \Set{u}$ by Lemma~\ref{STRUCTURE: THIN}(\ref{ITEM: MATCHING}).
    If $v_{1} \in U_{1}$ is included in a solution, then $S \setminus \Set{v_{1}}$ is a solution to the instance induced by $G - v_{1}$.
    Otherwise, $v_{1}$ does not belong to any solution, and thus $u$ belongs to every solution since edge $uv_{1}$ is marked.
    Additionally, referring to Lemma~\ref{SOLUTION SIZE1}, there must exist a solution either including $\hat{t}$ or excluding $\hat{t}$ but containing $U_{1} = \Set{v_{1}}$.
    In the case where $v_{1}$ does not belong to any solution, we can derive that there exists a solution containing both $\hat{t}$ and $u$.
    From now on, we have proved that the first branching rule is safe, and it is evident that the branching vector is $(1, 2)$.

    Now, we show the safeness of the second branching rule.
    Observe that for any minimum solution $S$ not containing $\hat{t}$, we have $|Q_{1} \setminus S| \leq 2$ and $|Q_{2} \setminus S| \leq 2$.
    We will prove that if there exists a minimal solution $S$ including all vertices in $Q$ except $\hat{t}$, there exists another minimum solution containing $\hat{t}$.
    This implies that either there is a minimal solution including $\hat{t}$ or including $(Q_{1} \cup Q_{2}) \setminus Q$.
    
    Recall that $s_{1}$ is the size of the minimum solution to the instance induced by $G[X_{1}] = G[X_{0} \cup \Set{v_{1}}]$.
    Notice that Steps~\ref{CHORDAL BRANCH: BIG S0+U1} and~\ref{CHORDAL BRANCH: BIG S0+Q1} cannot be applied.
    We have $s_{0} = 0$ and $s_{1} = 1$ by Lemma~\ref{STRUCTURE: FISH}.
    Given a minimum solution $S$, if $Q \setminus S = \Set{\hat{t}}$, then it follows that $|S \cap X_{1}| \geq s_{1} = 1$.
    Consider the set $S' = (S \setminus X_{1}) \cup \Set{\hat{t}}$.
    We claim that $S'$ is also a minimum solution.
    On one hand, one can easily find that $|S'| \leq |S \setminus X_{1}| + 1 \leq |S|$.
    On the other hand, observe that the separator $Q$ is entirely included in $S'$.
    According to Lemma~\ref{STRUCTURE: FISH}(\ref{ITEM: X1-Q1 SIMPLICIAL}), all vertices distinct from $v_{1}$ in $X_{Q} \setminus Q = X_{1} \setminus Q$ are simplicial.
    This indicates that $G[X_{Q} \setminus Q]$ is a star, which contains no $T$-triangle.
    By the assumption, $v_{1}$ is not adjacent to the vertices in $X_{Q} \setminus Q$ via marked edges, which yields that no marked edge appears in $G[X_{Q} \setminus Q]$.
    Since $Q \subseteq S'$, we only need to show $S' \setminus (X_{Q} \setminus Q)$ is a solution to the instance induced by $G - (X_{Q} \setminus Q)$.
    Observe that $S \setminus X_{Q} \subseteq S'$, which means that $S' \setminus (X_{Q} \setminus Q)$ is a solution.
    Therefore, we can conclude that $S'$ is a minimum solution containing $\hat{t}$.
    
    Finally, since $Q_{1} \setminus Q$ and $Q_{2} \setminus Q$ are both non-empty, we have $|Q_{1} \cup Q_{2}| - |Q| = |Q_{1} \setminus Q| + |Q_{2} \setminus Q| \geq 2$, implying that the branching vector is not worse than $(1, 2)$.
\end{proof}

\begin{theorem} \label{SFVS CHORDAL}
    {\rm\texttt{WholeAlg}} solves \textsc{SFVS-C} in time $\mathcal{O}^{*}(\alpha^{k} + 1.6181^{k})$ if \textsc{SFVS-S} can be solved in time $\mathcal{O}^{*}(\alpha^{k})$.
\end{theorem}

\begin{proof}
    For every step in Part~\Rome{1} of \texttt{WholeAlg}, the reduction rules do not increase the value $k$, and the branching vectors of the branching rule are not worse than $(1, 2)$.
    Thus the branching factors of the branching rules in Part~\Rome{1} of \texttt{WholeAlg} is $1.6181$.

    Furthermore, Steps~\ref{CHORDAL BRANCH: BIG S0+U1} and~\ref{CHORDAL BRANCH: BIG S0+Q1} rely on a preprocessing that invokes the sub-algorithm \texttt{GoodAlg}.
    Thus the Divide-and-Conquer procedure of \texttt{WholeAlg} takes $(\alpha^{k} + 1.6181^{k})n^{c}$ time to compute $U_{0}$ and $U_{1}$ for some constant $c$.
    Then we get the following recurrence relation
    \begin{equation*}
        R(k) \leq (\alpha^{k} + 1.6181^{k}) \cdot n^{c} + R(k - 1) + R(k - 2).
    \end{equation*}
    One can easily verify that $R(k) = \mathcal{O}^{*}(\alpha^{k} + 1.6181^{k})$.
\end{proof}

\begin{corollary}
    \textsc{SFVS-C} can be solved in time $\mathcal{O}^{*}(1.8192^{k})$.
\end{corollary}

\section{Exact Algorithms for \textsc{SFVS in Chordal Graphs}}

In this section, we show that our result (i.e., the $\mathcal{O}^{*}(1.8192^{k})$-time parameterized algorithm) can be used to obtain fast exact algorithms with respect to the number $n$ of vertices.

The polynomial-time reduction between \textsc{Generalized SFVS in Chordal Graphs} (\textsc{Generalized SFVS-C}) and \textsc{SFVS-C} introduced in Preliminaries changes the number of vertices in the graph.
So here, we need to distinguish the original version \textsc{SFVS-C} and the generalized version \textsc{Generalized SFVS-C}.

Fomin et al.~\cite{jacmFominGLS19} introduced a framework for converting an $\mathcal{O}^{*}(\alpha^{k})$-algorithm for a so-called ``subset problem'' to an $\mathcal{O}^{*}((2 - \frac{1}{\alpha})^{n})$-algorithm for the same problem.
\textsc{Generalized SFVS-C} satisfies the conditions.
Hence, by using this result together with the $\mathcal{O}^{*}(1.8192^{k})$-time algorithm for \textsc{Generalized SFVS-C}, we directly get the following result.

\begin{theorem}
    \textsc{Generalized SFVS-C} can be solved in time $\mathcal{O}(1.4504^{n})$.
\end{theorem}

Next, we show that the running time bound can be further improved for the original \textsc{SFVS-S} and \textsc{SFVS-C} without marked edges.

We first consider \textsc{SFVS-S}.
For an instance $\mathcal{I} = (G = (V, E), T, k)$, we first find a split partition $(I, K)$ of the vertex set $V$ in linear time, where $K$ is a clique, and $I$ is an independent set, and do the following steps in order.

\begin{exactstep} \label{EXACT REDUCTION: SMALL K}
    If  $|K| \leq 15$, solve the instance directly in polynomial time.
\end{exactstep}

\begin{exactstep} \label{EXACT BRANCH: TERMINALS IN K}
    If $K \cap T \neq \varnothing$, pick a terminal $t \in K \cap T$ and branch into three instances by either
    \begin{itemize}
        \item removing $t$, and decreasing $n$ by $1$;
        \item removing any a subset $K' \subseteq K \setminus \Set{t}$ of size $10$, and decreasing $n$ by $7$;
        \item removing $K'' = K \setminus (\Set{t} \cup K')$, and decreasing $n$ by $|K| - 8$.
    \end{itemize}
\end{exactstep}

The safeness of Step~\ref{EXACT BRANCH: TERMINALS IN K} is based on the following observation.
If $t$ is not contained in the solution, at most one vertex in the clique $K \setminus \Set{t}$ is not contained in the solution.
We partition $K \setminus \Set{t}$ into two parts $K'$ and $K''$.
At least one of $K'$ and $K''$ should be in the solution.
Thus, we can do the branch.
The branching vector is $(1, 7, |K| - 8)$, where $|K| \geq 15$.
So the branching factor is at most $1.3422$.

\begin{exactstep} \label{EXACT REDUCTION: BIG I AND K}
    If $|I| \leq k$ or $|K| \leq k$, return Yes.
\end{exactstep}

We can execute this step based on the following observation.
If $|I| \leq k$, $I$ is a solution since there is no terminal in $K$ after Step~\ref{EXACT BRANCH: TERMINALS IN K}.
If $|K| \leq k$, $K$ is a solution since $I$ is an independent set.

\begin{exactstep} \label{EXACT REDUCTION: PARAMETERIZED ALGORITHM}
    Solve the problem in $\mathcal{O}^{*}(1.8192^{k})$ time by Lemma~\ref{TIME: GOOD}.
\end{exactstep}

In Step~\ref{EXACT REDUCTION: PARAMETERIZED ALGORITHM}, we have that $|I| > k$ and $|K| > k$ and thus $n > 2k$.
This implies that the problem can be solved in time $\mathcal{O}^{*}(1.8192^{k}) \leq \mathcal{O}^{*}(1.8192^{n/2}) = \mathcal{O}(1.3488^{n})$.

We get the following result.

\begin{theorem}
    The original \textsc{SFVS-S} can be solved in time $\mathcal{O}(1.3488^{n})$.
\end{theorem}

We also improve the running time bound for \textsc{SFVS-C}.

\begin{theorem}
    The original \textsc{SFVS-C} can be solved in time $\mathcal{O}(1.3788^{n})$.
\end{theorem}

\begin{proof}
    First, we show that the problem can be solved in time $\mathcal{O}^{*}(2^{n - k})$.
    If $k \geq |T|$, we know that $T$ is a solution to $\mathcal{I}$, leading that $\mathcal{I}$ is a Yes-instance.
    Thus we assume that $k < |T|$.
    Next, we enumerate all subsets $P$ of non-terminals in $\mathcal{O}^{*}(2^{n - |T|}) = \mathcal{O}^{*}(2^{n - k})$ time.
    For each subset $P$, we delete $P$ from the graph and mark all other non-terminals as undeletable.
    Then we iteratively pick a pair of adjacent undeletable vertices $v$ and $u$: delete the terminal that forms a $T$-triangle with $v$ and $u$, and then merge $v$ and $u$ into one until there is no pair of adjacent undeletable vertices.
    
    Now, \textsc{SFVS-C} with the constraint that undeletable vertices are not allowed to be selected into the solution is equal to the weighted version of \textsc{FVS in Chordal Graphs}: we only need to set a very large weight to undeletable vertices to make sure that no solution will select an undeletable vertex.
    Note that weighted \textsc{FVS in Chordal Graphs} can be solved in polynomial time~\cite{iplYannakakisG87}.
    Thus, we can solve our problem in time $\mathcal{O}^{*}(2^{n - k})$.

    By doing a trivial trade-off between the $\mathcal{O}^{*}(2^{n - k})$-time algorithm and the $\mathcal{O}^{*}(1.8192^{k})$-time algorithm in Theorem~\ref{SFVS CHORDAL}, we know that \textsc{SFVS-C} can be solved in time $\mathcal{O}(1.3788^{n})$.
\end{proof}

\section{Prize-Collecting Maximum Independent Set in Hypergraphs}

Although we study the subset feedback vertex set problem in graph subclasses, \textsc{SFVS-S} already generalizes other interesting problems.

Several graph connectivity problems~\cite{damHalldorssonL09,sodaChekuriX17,talgFoxPZ23} can be modeled as natural problems in hypergraphs.
In hypergraphs, an edge can connect any number of vertices, whereas in an ordinary graph, an edge connects exactly two vertices.
Given a hypergraph $H$, the set of vertices and hyperedges are denoted by $V(H)$ and $E(H)$, respectively.
The maximum independent set problem in hypergraphs aims to find a maximum vertex subset $I \subseteq V(H)$ such that every hyperedge contains at most one vertex from $I$.
The maximum independent set problem in hypergraphs can be easily reduced to the maximum independent set problem in ordinary graphs: we only need to replace each hyperedge $e \in E(H)$ with a clique formed by the vertices in $e$ to get an ordinary graph.
In terms of exact algorithms, we may not need to distinguish this problem in hypergraphs and ordinary graphs. 
However, the prize-collecting version in hypergraphs becomes interesting, which allows us to violate the independent constraint with penalty.
As mentioned above, the prize-collecting version of many central $\NP$-hard problems has drawn certain attention recently.

\defproblem{\textsc{Prize-Collecting Maximum Independent Set in hypergraphs (PCMIS)}}
{A hypergraph $H$ and an integer $p$.}
{Determine whether there is a subset of vertices $I \subseteq V(H)$ of the prize at least $p$, where the prize of $I$ is the size of $I$ minus the number of hyperedges that contain at least two vertices from $I$.}

\begin{lemma}
    \textsc{PCMIS} is polynomially solvable fot $p \leq 1$, and \textsc{PCMIS} is $\NP$-hard for each constant $p \geq 2$.
\end{lemma}

\begin{proof}
    By definition, any singleton set has a prize of $1$.
    Therefore, \textsc{PCMIS} is polynomially solvable when $p \leq 1$.

    We will prove the $\NP$-hardness of \textsc{PCMIS} with $p \geq 2$ by reducing from the maximum independent set problem in ordinary undirected graphs.
    Let $(G, k)$ be an instance of the maximum independent set problem.
    We construct an instance $(H,p)$ of \textsc{PCMIS}, where $p\geq 2$ is a constant. Since $p$ is a constant, we can assume that $k\geq p$.

    Suppose $|V(G)| = n$ and $|E(G)| = m$.
    We now construct a hypergraph $H$ with $n$ vertices and $nm + k - p$ hyperedges.
    Specifically, for each vertex $v \in V(G)$, we introduce a vertex $v'$, and thus we obtain $V(H) = \{ v' : v \in V(G) \}$.
    Next, for each edge $uv \in E(G)$, we introduce $n$ identical hyperedges $e^{uv}_{i} = \Set{u, v}$ ($i \in [n]$); we also add $k - p$ identical hyperedges  $e'_{i} = V(H)$ ($i \in [k - p]$).
    Hence, $H$ contains $nm + (k - p)$ hyperedges: $E(H) = E'_{1} \cup E'_{2}$, where $E'_{1} = \{ e^{uv}_{i} = \Set{u, v} : uv \in E(G),~ i \in [n]\}$ and $E'_{2} = \{ e'_{i} = V(H) : i \in [k - p]\}$. 

    Finally, we show $(G, k)$ is a Yes-instance if and only if $(H, p)$ is a Yes-instance.
    On the one hand, let $I \subseteq V(G)$ be an independent set of $G$ with the size $k$.
    Let $I' = \{ v' : v \in I \}$, and we have that every hyperedge in $E'_{1}$ contains at most one vertex from $I'$.
    Additionally, each hyperedge in $E'_{2}$ contains exactly $k$ vertices from $I'$.
    Since $k \geq p \geq 2$, we derive that the prize of $I'$ is $k - |E'_{2}| = k - (k - p) = p$, which means that $(H, p)$ is a Yes-instance.

    On the other hand, let $I' \subseteq V(H)$ be a vertex subset of $H$ with the prize $p$.
    Let $X$ be the set of hyperedges containing at least two vertices from $I'$.
    We have that $p = |I'| - |X|$.
    Note that if one hyperedge $e^{vu}_{i} \in E'_{1}$ ($i \in [n]$) is in $X$, then all the $n$ hyperedges identical with $e^{vu}_{i}$ should be in $X$, which will make $p \leq |I'| - n \leq 0$, a contradiction. Therefore, we derive that $I = \{ v \in V(G) : v' \in I' \}$ is an independent set in $G$ and $X \cap E'_{1}=\emptyset$. Since $X \subseteq E'_{2}$,
    we have $|I| = |I'| = p + |X| \geq p + (k - p) \geq k$.
    We conclude that $I$ is an independent set of $G$ with size at least $k$, leading that $(G, k)$ is a Yes-instance.
\end{proof}

Previously, no exact algorithm for \textsc{PCMIS} faster than $\mathcal{O}^{*}(2^{n})$ is known.
We show that \textsc{PCMIS} parameterized by the size of the vertex set $n$ and \textsc{SFVS-S} parameterized by the size of the solution has a strong relation, which leads us to break the ``$2^{n}$-barrier'' for \textsc{PCMIS}.

\begin{lemma} \label{LEM: SFVSS=PCMIS}
    For any constant $\alpha > 1$, an $\mathcal{O}^{*}(\alpha^{k})$-time algorithm for \textsc{SFVS-S} leads to an $\mathcal{O}^{*}(\alpha^{n - p})$-algorithm for \textsc{PCMIS}.
\end{lemma}

\begin{proof}
    For an instance $(H, p)$ of \textsc{PCMIS}, we construct an instance $(G, T, M, k)$ of \textsc{SFVS-S}.
    Suppose that $H$ contains $n$ vertices and $m$ hyperedges.
    We construct a split graph $G$ as follows.
    We first introduce a clique with vertices $\{v' : v \in V(H) \}$; then for each hyperedge $e \in E(H)$, introduce a new terminal $t'_{e}$ whose neighbors are exactly the vertices in the clique corresponding to the vertices in hyperedge $e$.
    The terminal set is set as $T = \{ t'_{e} : e \in E(H) \}$, the marked edge set is set as $M = \varnothing$, and let $k = n - p$.
    
    We have the key idea: every hyperedge in a hypergraph contains at most one vertex if and only if the corresponding split graph of the hypergraph contains no $T$-triangle.
    Let $S$ be a solution to $(G, T, M, k)$ containing $m'$ terminals and $n'$ non-terminals.
    We can see that $I = \{ v : v' \in V(G) \setminus (T \cup S)\}$ is a vertex set with the prize at least $(n - n') - m' = n - k = p$ in $(H, p)$.
    As for the opposite direction, suppose $I$ is a solution to $(H, p)$ of size $n'$, and the prize of $I$ is $p$.
    We can derive that there are at most $m' = n' - p$ hyperedges containing at least two vertices from $I$.
    This means that in $G$, we can remove $m'$ terminals and $n - n'$ non-terminals to obtain a subgraph without any $T$-triangle, leading that $(G, T, M, k)$ has a solution of size $(n - n') + m' = n - p = k$.
    Therefore, $(G, T, M, k)$ is a Yes-instance if and only if $(H, p)$ is a Yes-instance.
    We finish the proof of Lemma~\ref{LEM: SFVSS=PCMIS}.
\end{proof}

Based on Lemma~\ref{TIME: GOOD}, we obtain an exact algorithm for \textsc{PCMIS} breaking the $2^{n}$ barrier.

\begin{corollary}
    \textsc{PCMIS} can be solved in time $\mathcal{O}^{*}(1.8192^{n})$.
\end{corollary}

\section{Conclusion}

In this paper, we broke the ``$2^{k}$-barrier'' for \textsc{SFVS in Chordal Graphs}.
As a corollary, we obtained an exact algorithm faster than $\mathcal{O}^{*}(2^{n})$ for \textsc{Prize-Collecting Maximum Independent Set in hypergraphs}.
To achieve this breakthrough, we introduced a new measure based on the Dulmage-Mendelsohn decomposition.
This measure served as the basis for designing and analyzing an algorithm that addresses a crucial sub-case.
Furthermore, we analyzed the whole algorithm using the traditional measure $k$, employing various techniques such as a divide-and-conquer approach and reductions based on small separators.
The bottleneck of our algorithm occurs when dealing with \textsc{SFVS in Split Graphs}.

We think it is interesting to break the ``$2^{k}$-barrier'' or ``$2^{n}$-barrier'' for more important problems, say the \textsc{Steiner Tree} problem and TSP.
For \textsc{SFVS} in general graphs, the best result is $\mathcal{O}^{*}(4^{k})$~\cite{siamcompIwataWY16,focsIwataYY18}.
It will also be interesting to reduce the gap between the results in general graphs and chordal graphs.

\bibliography{bib_SFVSC}

\begin{thebibliography}{10}

\bibitem{jcssAbu-Khzam10}
Faisal~N. Abu{-}Khzam.
\newblock A kernelization algorithm for $d$-hitting set.
\newblock {\em Journal of Computer and System Sciences}, 76(7):524--531, 2010.
\newblock \href {https://doi.org/10.1016/J.JCSS.2009.09.002}
  {\path{doi:10.1016/J.JCSS.2009.09.002}}.

\bibitem{stocBlauthN23}
Jannis Blauth and Martin N{\"{a}}gele.
\newblock An improved approximation guarantee for prize-collecting {TSP}.
\newblock In {\em Proceedings of the 55th Annual {ACM} Symposium on Theory of
  Computing, {STOC} 2023}, pages 1848--1861, 2023.
\newblock \href {https://doi.org/10.1145/3564246.3585159}
  {\path{doi:10.1145/3564246.3585159}}.

\bibitem{ijfcsBodlaender94}
Hans~L. Bodlaender.
\newblock On disjoint cycles.
\newblock {\em International Journal of Foundations of Computer Science},
  5(1):59--68, 1994.
\newblock \href {https://doi.org/10.1142/S0129054194000049}
  {\path{doi:10.1142/S0129054194000049}}.

\bibitem{dmBuneman74}
Peter Buneman.
\newblock A characterisation of rigid circuit graphs.
\newblock {\em Discrete Mathematics}, 9(3):205--212, 1974.
\newblock \href {https://doi.org/10.1016/0012-365X(74)90002-8}
  {\path{doi:10.1016/0012-365X(74)90002-8}}.

\bibitem{algorithmicaCaoCL15}
Yixin Cao, Jianer Chen, and Yang Liu.
\newblock On feedback vertex set: New measure and new structures.
\newblock {\em Algorithmica}, 73(1):63--86, 2015.
\newblock \href {https://doi.org/10.1007/s00453-014-9904-6}
  {\path{doi:10.1007/s00453-014-9904-6}}.

\bibitem{sodaChekuriX17}
Chandra Chekuri and Chao Xu.
\newblock Computing minimum cuts in hypergraphs.
\newblock In {\em Proceedings of the 28th Annual {ACM-SIAM} Symposium on
  Discrete Algorithms, {SODA} 2017}, pages 1085--1100, 2017.
\newblock \href {https://doi.org/10.1137/1.9781611974782.70}
  {\path{doi:10.1137/1.9781611974782.70}}.

\bibitem{jcssChenFLLV08}
Jianer Chen, Fedor~V. Fomin, Yang Liu, Songjian Lu, and Yngve Villanger.
\newblock Improved algorithms for feedback vertex set problems.
\newblock {\em Journal of Computer and System Sciences}, 74(7):1188--1198,
  2008.
\newblock \href {https://doi.org/10.1016/J.JCSS.2008.05.002}
  {\path{doi:10.1016/J.JCSS.2008.05.002}}.

\bibitem{jcssChenK03}
Jianer Chen and Iyad~A. Kanj.
\newblock Constrained minimum vertex cover in bipartite graphs: complexity and
  parameterized algorithms.
\newblock {\em Journal of Computer and System Sciences}, 67(4):833--847, 2003.
\newblock \href {https://doi.org/10.1016/J.JCSS.2003.09.003}
  {\path{doi:10.1016/J.JCSS.2003.09.003}}.

\bibitem{siamdmCyganPPW13}
Marek Cygan, Marcin Pilipczuk, Michal Pilipczuk, and Jakub~Onufry Wojtaszczyk.
\newblock Subset feedback vertex set is fixed-parameter tractable.
\newblock {\em {SIAM} Journal on Discrete Mathematics}, 27(1):290--309, 2013.
\newblock \href {https://doi.org/10.1137/110843071}
  {\path{doi:10.1137/110843071}}.

\bibitem{iwpecDehneFRS04}
Frank K. H.~A. Dehne, Michael~R. Fellows, Frances~A. Rosamond, and Peter Shaw.
\newblock Greedy localization, iterative compression, modeled crown reductions:
  New {$\FPT$} techniques, an improved algorithm for set splitting, and a novel
  {$2k$} kernelization for vertex cover.
\newblock In {\em Proceedings of the 1st International Workshop of
  Parameterized and Exact Computation, {IWPEC} 2004}, pages 271--280, 2004.

\bibitem{DagSemProcDeamineHM09}
Erik~D. Demaine, MohammadTaghi Hajiaghayi, and D\'{a}niel Marx.
\newblock Open problems from dagstuhl seminar 09511.
\newblock In {\em Parameterized complexity and approximation algorithms},
  volume 9511 of {\em Dagstuhl Seminar Proceedings (DagSemProc)}, pages 1--14,
  2009.
\newblock \href {https://doi.org/10.4230/DagSemProc.09511.1}
  {\path{doi:10.4230/DagSemProc.09511.1}}.

\bibitem{amhaajDirac61}
Gabriel~Andrew Dirac.
\newblock On rigid circuit graphs.
\newblock {\em Abhandlungen aus dem Mathematischen Seminar der Universitat
  Hamburg}, 25(1):71--76, 1961.
\newblock \href {https://doi.org/10.1007/BF02992776}
  {\path{doi:10.1007/BF02992776}}.

\bibitem{jdaDomGHNT10}
Michael Dom, Jiong Guo, Falk H{\"{u}}ffner, Rolf Niedermeier, and Anke
  Tru{\ss}.
\newblock Fixed-parameter tractability results for feedback set problems in
  tournaments.
\newblock {\em Journal of Discrete Algorithms}, 8(1):76--86, 2010.
\newblock \href {https://doi.org/10.1016/j.jda.2009.08.001}
  {\path{doi:10.1016/j.jda.2009.08.001}}.

\bibitem{sctDowneyF92}
Rodney~G. Downey and Michael~R. Fellows.
\newblock Fixed-parameter tractability and completeness {III:} some structural
  aspects of the {$\W$} hierarchy.
\newblock In {\em Workshop on Structure and Complexity Theory}, pages 191--225,
  1992.
\newblock \href {https://doi.org/10.5555/183589.183729}
  {\path{doi:10.5555/183589.183729}}.

\bibitem{springerDowneyF99}
Rodney~G. Downey and Michael~R. Fellows.
\newblock {\em Parameterized Complexity}.
\newblock Monographs in Computer Science. Springer, Berlin, 1999.
\newblock \href {https://doi.org/10.1007/978-1-4612-0515-9}
  {\path{doi:10.1007/978-1-4612-0515-9}}.

\bibitem{canjmathDulmageM58}
A.~L. Dulmage and N.~S. Mendelsohn.
\newblock Coverings of bipartite graphs.
\newblock {\em Canadian Journal of Mathematics}, 10:517--534, 1958.

\bibitem{tranrscDulmageA59}
Andrew~L Dulmage.
\newblock A structure theory of bipartite graphs of finite exterior dimension.
\newblock {\em The Transactions of the Royal Society of Canada, Section III},
  53:1--13, 1959.

\bibitem{siamcompEvenNZ00}
Guy Even, Joseph Naor, and Leonid Zosin.
\newblock An $8$-approximation algorithm for the subset feedback vertex set
  problem.
\newblock {\em {SIAM} Journal on Computing}, 30(4):1231--1252, 2000.
\newblock \href {https://doi.org/10.1137/S0097539798340047}
  {\path{doi:10.1137/S0097539798340047}}.

\bibitem{jacmFominGLS19}
Fedor~V. Fomin, Serge Gaspers, Daniel Lokshtanov, and Saket Saurabh.
\newblock Exact algorithms via monotone local search.
\newblock {\em Journal of the {ACM}}, 66(2):8:1--8:23, 2019.
\newblock \href {https://doi.org/10.1145/3284176} {\path{doi:10.1145/3284176}}.

\bibitem{algorithmicaFominHKPV14}
Fedor~V. Fomin, Pinar Heggernes, Dieter Kratsch, Charis Papadopoulos, and Yngve
  Villanger.
\newblock Enumerating minimal subset feedback vertex sets.
\newblock {\em Algorithmica}, 69(1):216--231, 2014.
\newblock \href {https://doi.org/10.1007/s00453-012-9731-6}
  {\path{doi:10.1007/s00453-012-9731-6}}.

\bibitem{talgFominLLSTZ19}
Fedor~V. Fomin, Tien{-}Nam Le, Daniel Lokshtanov, Saket Saurabh, St{\'{e}}phan
  Thomass{\'{e}}, and Meirav Zehavi.
\newblock Subquadratic kernels for implicit $3$-hitting set and $3$-set packing
  problems.
\newblock {\em {ACM} Transactions on Algorithms}, 15(1):13:1--13:44, 2019.
\newblock \href {https://doi.org/10.1145/3293466} {\path{doi:10.1145/3293466}}.

\bibitem{talgFoxPZ23}
Kyle Fox, Debmalya Panigrahi, and Fred Zhang.
\newblock Minimum cut and minimum $k$-cut in hypergraphs via branching
  contractions.
\newblock {\em {ACM} Transactions on Algorithms}, 19(2):13:1--13:22, 2023.

\bibitem{talgFukunaga17}
Takuro Fukunaga.
\newblock Spider covers for prize-collecting network activation problem.
\newblock {\em {ACM} Transactions on Algorithms}, 13(4):49:1--49:31, 2017.
\newblock \href {https://doi.org/10.1145/3132742} {\path{doi:10.1145/3132742}}.

\bibitem{pjmFulkersonDG65}
Delbert Fulkerson and Oliver Gross.
\newblock Incidence matrices and interval graphs.
\newblock {\em Pacific Journal of Mathematics}, 15(3):835--855, 1965.
\newblock \href {https://doi.org/10.2140/pjm.1965.15.835}
  {\path{doi:10.2140/pjm.1965.15.835}}.

\bibitem{jctGavril74}
F\v{a}nic\v{a} Gavril.
\newblock The intersection graphs of subtrees in trees are exactly the chordal
  graphs.
\newblock {\em Journal of Combinatorial Theory, Series B}, 16(1):47--56, 1974.
\newblock \href {https://doi.org/10.1016/0095-8956(74)90094-X}
  {\path{doi:10.1016/0095-8956(74)90094-X}}.

\bibitem{jcssGuoGHNW06}
Jiong Guo, Jens Gramm, Falk H{\"{u}}ffner, Rolf Niedermeier, and Sebastian
  Wernicke.
\newblock Compression-based fixed-parameter algorithms for feedback vertex set
  and edge bipartization.
\newblock {\em Journal of Computer and System Sciences}, 72(8):1386--1396,
  2006.
\newblock \href {https://doi.org/10.1016/j.jcss.2006.02.001}
  {\path{doi:10.1016/j.jcss.2006.02.001}}.

\bibitem{damHalldorssonL09}
Magn{\'{u}}s~M. Halld{\'{o}}rsson and Elena Losievskaja.
\newblock Independent sets in bounded-degree hypergraphs.
\newblock {\em Discrete Applied Mathematics}, 157(8):1773--1786, 2009.

\bibitem{mstHolsK18}
Eva{-}Maria~C. Hols and Stefan Kratsch.
\newblock A randomized polynomial kernel for subset feedback vertex set.
\newblock {\em Theory of Computing Systems}, 62(1):63--92, 2018.
\newblock \href {https://doi.org/10.1007/S00224-017-9805-6}
  {\path{doi:10.1007/S00224-017-9805-6}}.

\bibitem{siamcompHopcroftK73}
John~E. Hopcroft and Richard~M. Karp.
\newblock An $n^{5/2}$ algorithm for maximum matchings in bipartite graphs.
\newblock {\em {SIAM} Journal on Computing}, 2(4):225--231, 1973.
\newblock \href {https://doi.org/10.1137/0202019} {\path{doi:10.1137/0202019}}.

\bibitem{mstHuffnerKMN10}
Falk H{\"{u}}ffner, Christian Komusiewicz, Hannes Moser, and Rolf Niedermeier.
\newblock Fixed-parameter algorithms for cluster vertex deletion.
\newblock {\em Theory of Computing Systems}, 47(1):196--217, 2010.
\newblock \href {https://doi.org/10.1007/S00224-008-9150-X}
  {\path{doi:10.1007/S00224-008-9150-X}}.

\bibitem{icalpIwata17}
Yoichi Iwata.
\newblock Linear-time kernelization for feedback vertex set.
\newblock In {\em Proceedings of the 44th International Colloquium on Automata,
  Languages, and Programming, {ICALP} 2017}, volume~80, pages 68:1--68:14,
  2017.
\newblock \href {https://doi.org/10.4230/LIPIcs.ICALP.2017.68}
  {\path{doi:10.4230/LIPIcs.ICALP.2017.68}}.

\bibitem{algorithmicaIwataK21}
Yoichi Iwata and Yusuke Kobayashi.
\newblock Improved analysis of highest-degree branching for feedback vertex
  set.
\newblock {\em Algorithmica}, 83(8):2503--2520, 2021.
\newblock \href {https://doi.org/10.1007/s00453-021-00815-w}
  {\path{doi:10.1007/s00453-021-00815-w}}.

\bibitem{siamcompIwataWY16}
Yoichi Iwata, Magnus Wahlstr{\"{o}}m, and Yuichi Yoshida.
\newblock Half-integrality, {LP}-branching, and {$\FPT$} algorithms.
\newblock {\em {SIAM} Journal on Computing}, 45(4):1377--1411, 2016.
\newblock \href {https://doi.org/10.1137/140962838}
  {\path{doi:10.1137/140962838}}.

\bibitem{focsIwataYY18}
Yoichi Iwata, Yutaro Yamaguchi, and Yuichi Yoshida.
\newblock {$0$/$1$/All} {CSPs}, half-integral {$A$}-path packing, and
  linear-time {$\FPT$} algorithms.
\newblock In {\em 59th {IEEE} Annual Symposium on Foundations of Computer
  Science, {FOCS} 2018}, pages 462--473, 2018.
\newblock \href {https://doi.org/10.1109/FOCS.2018.00051}
  {\path{doi:10.1109/FOCS.2018.00051}}.

\bibitem{iwpecKanjPS04}
Iyad~A. Kanj, Michael~J. Pelsmajer, and Marcus Schaefer.
\newblock Parameterized algorithms for feedback vertex set.
\newblock In {\em Proceedings of the 1st International Workshop of
  Parameterized and Exact Computation, {IWPEC} 2004}, volume 3162, pages
  235--247, 2004.
\newblock \href {https://doi.org/10.1007/978-3-540-28639-4\_21}
  {\path{doi:10.1007/978-3-540-28639-4\_21}}.

\bibitem{cocoKarp72}
Richard~M. Karp.
\newblock Reducibility among combinatorial problems.
\newblock In {\em Proceedings of a symposium on the Complexity of Computer
  Computations}, The {IBM} Research Symposia Series, pages 85--103, 1972.
\newblock \href {https://doi.org/10.1007/978-1-4684-2001-2\_9}
  {\path{doi:10.1007/978-1-4684-2001-2\_9}}.

\bibitem{iplKociumakaP14}
Tomasz Kociumaka and Marcin Pilipczuk.
\newblock Faster deterministic feedback vertex set.
\newblock {\em Information Processing Letters}, 114(10):556--560, 2014.
\newblock \href {https://doi.org/10.1016/j.ipl.2014.05.001}
  {\path{doi:10.1016/j.ipl.2014.05.001}}.

\bibitem{stacsKumarL16}
Mithilesh Kumar and Daniel Lokshtanov.
\newblock Faster exact and parameterized algorithm for feedback vertex set in
  tournaments.
\newblock In {\em Proceedings of the 33rd Symposium on Theoretical Aspects of
  Computer Science, {STACS} 2016}, volume~47, pages 49:1--49:13, 2016.
\newblock \href {https://doi.org/10.4230/LIPICS.STACS.2016.49}
  {\path{doi:10.4230/LIPICS.STACS.2016.49}}.

\bibitem{elsevierLocaszP86}
L{\'a}szl{\'o} Lov{\'a}sz and Michael~D. Plummer.
\newblock {\em Matching theory}, volume 121 of {\em North-Holland Mathematics
  Studies}.
\newblock Elsevier Science Ltd., London, 1 edition, 1986.

\bibitem{damPapadopoulosT19}
Charis Papadopoulos and Spyridon Tzimas.
\newblock Polynomial-time algorithms for the subset feedback vertex set problem
  on interval graphs and permutation graphs.
\newblock {\em Discrete Applied Mathematics}, 258:204--221, 2019.
\newblock \href {https://doi.org/10.1016/j.dam.2018.11.017}
  {\path{doi:10.1016/j.dam.2018.11.017}}.

\bibitem{algorithmicaPedrosaR22}
Lehilton Lelis~Chaves Pedrosa and Hugo Kooki~Kasuya Rosado.
\newblock A $2$-approximation for the $k$-prize-collecting steiner tree
  problem.
\newblock {\em Algorithmica}, 84(12):3522--3558, 2022.
\newblock \href {https://doi.org/10.1007/S00453-021-00919-3}
  {\path{doi:10.1007/S00453-021-00919-3}}.

\bibitem{algorithmicaPhilipRST19}
Geevarghese Philip, Varun Rajan, Saket Saurabh, and Prafullkumar Tale.
\newblock Subset feedback vertex set in chordal and split graphs.
\newblock {\em Algorithmica}, 81(9):3586--3629, 2019.
\newblock \href {https://doi.org/10.1007/S00453-019-00590-9}
  {\path{doi:10.1007/S00453-019-00590-9}}.

\bibitem{talgRamanSS06}
Venkatesh Raman, Saket Saurabh, and C.~R. Subramanian.
\newblock Faster fixed parameter tractable algorithms for finding feedback
  vertex sets.
\newblock {\em {ACM} Transactions on Algorithms}, 2(3):403--415, 2006.
\newblock \href {https://doi.org/10.1145/1159892.1159898}
  {\path{doi:10.1145/1159892.1159898}}.

\bibitem{siamcompRoseTL76}
Donald~J. Rose, Robert~Endre Tarjan, and George~S. Lueker.
\newblock Algorithmic aspects of vertex elimination on graphs.
\newblock {\em {SIAM} Journal on Computing}, 5(2):266--283, 1976.
\newblock \href {https://doi.org/10.1137/0205021} {\path{doi:10.1137/0205021}}.

\bibitem{seiccgtcStephaneH77}
F\"{o}ldes Stephane and Peter Hammer.
\newblock Split graphs.
\newblock In {\em Proceedings of the 8th south-east Combinatorics, Graph
  Theory, and Computing, {SEICCGTC} 1977}, pages 311--315, 1977.

\bibitem{talgThomasse10}
St{\'{e}}phan Thomass{\'{e}}.
\newblock A $4k^{2}$ kernel for feedback vertex set.
\newblock {\em {ACM} Transactions on Algorithms}, 6(2):32:1--32:8, 2010.
\newblock \href {https://doi.org/10.1145/1721837.1721848}
  {\path{doi:10.1145/1721837.1721848}}.

\bibitem{cocoonTianXY23}
Kangyi Tian, Mingyu Xiao, and Boting Yang.
\newblock Parameterized algorithms for cluster vertex deletion on degree-$4$
  graphs and general graphs.
\newblock In {\em Proceedings of the 29th Annual International Computing and
  Combinatorics, {COCOON} 2023 Part {I}}, volume 14422 of {\em LNCS}, pages
  182--194, 2023.
\newblock \href {https://doi.org/10.1007/978-3-031-49190-0\_13}
  {\path{doi:10.1007/978-3-031-49190-0\_13}}.

\bibitem{phdWahlstrom07}
Magnus Wahlstr{\"{o}}m.
\newblock {\em Algorithms, measures and upper bounds for satisfiability and
  related problems}.
\newblock PhD thesis, Link{\"{o}}ping University, Link{\"{o}}ping Sweden, 2007.

\bibitem{phdWalter72}
James~Richard Walter.
\newblock {\em Representations of rigid cycle graphs}.
\newblock PhD thesis, Wayne State University, 1972.

\bibitem{iplYannakakisG87}
Mihalis Yannakakis and Fanica Gavril.
\newblock The maximum $k$-colorable subgraph problem for chordal graphs.
\newblock {\em Information Processing Letters}, 24(2):133--137, 1987.
\newblock \href {https://doi.org/10.1016/0020-0190(87)90107-4}
  {\path{doi:10.1016/0020-0190(87)90107-4}}.

\end{thebibliography}

\appendix

\end{document}